\documentclass[10pt,a4paper]{article}
\usepackage[latin1]{inputenc}
\usepackage[T1]{fontenc}
\usepackage{indentfirst}
\usepackage{amsmath,amssymb,amsthm,amsfonts}
\usepackage{calligra,mathrsfs}
\usepackage{graphicx,color}
\usepackage{xcolor}
\usepackage[all,cmtip]{xy}
\usepackage{bbm}
\usepackage{hyperref}
\usepackage{mathtools}
\usepackage{verbatim}
\usepackage{stmaryrd}
\usepackage[shortlabels]{enumitem}
\usepackage{bm}
\usepackage{url}
\usepackage{comment}
\usepackage{slashed}
\usepackage{stmaryrd}
\usepackage{caption}
\usepackage{subcaption}
\usepackage{scalerel}
\usepackage{multicol}
\usepackage{tabularx}
\usepackage{tikz}
\usetikzlibrary{patterns,patterns.meta}
\usepackage{subcaption}
\usepackage{multirow}

\usepackage{tikz-cd}
\usetikzlibrary{matrix}
\usetikzlibrary{positioning}
\usetikzlibrary{arrows}
\usetikzlibrary{backgrounds}

\usepackage{caption}
\hypersetup{colorlinks=true,linkcolor=black,urlcolor=lcolor,citecolor=black}

\definecolor{amethyst}{rgb}{0.6, 0.4, 0.8}
\definecolor{alizarin}{rgb}{0.82, 0.1, 0.26}
\definecolor{darkorchid}{rgb}{0.6, 0.2, 0.8}
\definecolor{lcolor}{rgb}{0.6,0.3,0.3}
\definecolor{caribbeangreen}{rgb}{0.0, 0.8, 0.6}
\definecolor{flamingopink}{rgb}{0.99, 0.56, 0.67}
\definecolor{hollywoodcerise}{rgb}{0.96, 0.0, 0.63}
\definecolor{airforceblue}{rgb}{0.36, 0.54, 0.66}
\definecolor{applegreen}{rgb}{0.55, 0.71, 0.0}
\definecolor{ballblue}{rgb}{0.13, 0.67, 0.8}
\definecolor{brinkpink}{rgb}{0.98, 0.38, 0.5}
\definecolor{bluegray}{rgb}{0.4, 0.6, 0.8}
\definecolor{impx}{rgb}{0.96, 0.0, 0.63}
\definecolor{junglegreen}{rgb}{0.16, 0.67, 0.53}
\definecolor{applegreen}{rgb}{0.55, 0.71, 0.0}
\definecolor{darkpastelgreen}{rgb}{0.01, 0.75, 0.24}
\definecolor{parisgreen}{rgb}{0.31, 0.78, 0.47}
\definecolor{ufogreen}{rgb}{0.24, 0.82, 0.44}
\definecolor{bluebell}{rgb}{0.64, 0.64, 0.82}
\definecolor{cornflowerblue}{rgb}{0.39, 0.4, 0.93}
\definecolor{pastelred}{rgb}{1.0, 0.41, 0.38}
\definecolor{carminepink}{rgb}{0.92, 0.3, 0.26}
\definecolor{darkpastelpurple}{rgb}{0.59, 0.44, 0.84}
\definecolor{cerulean}{rgb}{0.0, 0.48, 0.65}
\definecolor{darkcerulean}{rgb}{0.03, 0.27, 0.49}
\definecolor{ceruleanblue}{rgb}{0.16, 0.32, 0.75}
\definecolor{green(munsell)}{rgb}{0.0, 0.66, 0.47}
\definecolor{bleudefrance}{rgb}{0.19, 0.55, 0.91}
\definecolor{antiquefuchsia}{rgb}{0.57, 0.36, 0.51}
\definecolor{deeplilac}{rgb}{0.6, 0.33, 0.73}
\definecolor{ferngreen}{rgb}{0.31, 0.47, 0.26}
\definecolor{qcolor}{rgb}{0.19,0.55,0.91}
\definecolor{hcolor}{rgb}{0.9,0.2,0.5}
\definecolor{scolor}{rgb}{0.9,0.5,0.2}
\definecolor{cyan}{rgb}{0.0,0.8,1.0}
\definecolor{cyan(process)}{rgb}{0.0, 0.72, 0.92}
\definecolor{babyblue}{rgb}{0.54, 0.81, 0.94}
\definecolor{triangleHor}{rgb}{0.93, 0.53, 0.18}
\definecolor{triangleVer}{rgb}{0.0, 0.5, 1.0}

\newcommand{\ulinsol}{\tc{magenta}{\udata}_*}
\newcommand{\vlinsol}{\tc{magenta}{\vdata_*}}

\newcommand{\rfix}{\tc{cyan}{r_*}}
\newcommand{\MinkThmsfix}{\tc{orange}{s_*}}

\newcommand{\tc}[2]{#2}

\newcommand{\step}{\vskip 3mm}

\renewcommand{\Im}{\textnormal{Im}}
\renewcommand{\Re}{\textnormal{Re}}
\newcommand{\p}{\partial}
\newcommand{\Z}{\mathbbm{Z}}
\newcommand{\N}{\mathbbm{N}}
\newcommand{\R}{\mathbbm{R}}

\newcommand{\C}{\mathbbm{C}}
\newcommand{\one}{\mathbbm{1}}
\newcommand{\eps}{\epsilon}

\renewcommand{\div}{\textnormal{div}}
\newcommand{\divST}{\widehat{\div}}

\newcommand{\CURL}{\textnormal{CURL}}

\newcommand{\image}{\textnormal{image}}

\renewcommand{\subset}{\subseteq}

\newcommand{\RNum}[1]{\mathrm{\uppercase\expandafter{\romannumeral #1\relax}}}

\newcommand{\xvecjap}{\langle \vec{x} \rangle}
\DeclareMathOperator{\arctanh}{arctanh}

\DeclareMathOperator{\Lie}{\mathcal{L}}
\renewcommand{\diamond}{\mathcal{D}}

\newcommand{\diamonddata}{\underline{\diamond\mkern-4mu}\mkern4mu }

\newcommand{\gmink}{\eta}
\newcommand{\gminkdata}{\underline{\eta}}

\newcommand{\MCinf}{\mathcal{F}}
\DeclareMathOperator{\MC}{\mathcal{F}}

\newcommand{\conenb}[1]{\tc{applegreen}{#1}}
\newcommand{\Kil}{\tc{cornflowerblue}{\mathfrak{K}}}
\newcommand{\confKilS}{\tc{orange}{\underline{\mathfrak{K}}}}

\newcommand{\transl}{\tc{cornflowerblue}{\mathfrak{t}}}
\newcommand{\boosts}{\tc{cornflowerblue}{\mathfrak{b}}}
\newcommand{\dens}[1]{|\Omega|^{#1}}
\newcommand{\I}{\mathcal{I}}
\newcommand{\Idata}{\underline{\I}}

\newcommand{\gx}{\mathfrak{g}}
\newcommand{\lx}{\mathfrak{L}}
\newcommand{\lxdata}{\underline{\lx}}

\newcommand{\anchorg}{\rho_{\gx}}

\newcommand{\OmegaC}{\Omega_{\C}}
\newcommand{\otimesCinf}{\otimes_{C^\infty}}
\newcommand{\dg}{d_{\gx}}
\newcommand{\dI}{d_{\I}}

\newcommand{\ddR}{d}

\newcommand{\dl}{d_{\lx}}
\newcommand{\uo}{u_0}
\newcommand{\uI}{u_{\I}}

\newcommand{\co}{c_{0}}
\newcommand{\ci}{c_{\ix}}
\newcommand{\cp}{c_{+}}
\newcommand{\cm}{c_{-}}
\newcommand{\cpm}{c_{\pm}}

\newcommand{\uodata}{\udata_0}
\newcommand{\uIdata}{\udata_{\I}}
\newcommand{\up}{u_{+}}
\newcommand{\um}{u_{-}}
\newcommand{\upm}{u_{\pm}}
\newcommand{\upbar}{\bar{u}_{+}}

\newcommand{\gxdata}{\underline{\mathfrak{g}}}
\newcommand{\Pconstraints}{\mkern1mu\underline{\mkern-1mu P\mkern-1mu}\mkern1mu}
\newcommand{\linPconstraints}{D\underline{P}}

\newcommand{\SPAN}{\textnormal{span}}

\newcommand{\diag}{\textnormal{diag}}

\newcommand{\tint}{\textstyle\int}
\newcommand{\tsum}{\textstyle\sum}

\newcommand{\gxG}{\mathfrak{g}_{\textnormal{G}}}
\newcommand{\gxdataG}{\gxdata_{\textnormal{G}}}
\newcommand{\OmegaG}{\Omega_{\textnormal{G}}}

\newcommand{\IG}{\I_{\textnormal{G}}}

\newcommand{\ngG}{\tc{lcolor}{n}}
\renewcommand{\ng}{\tc{lcolor}{m}}

\newcommand{\udata}{\mkern1mu\underline{\mkern-1mu u \mkern-1mu}\mkern1mu}
\newcommand{\vdata}{\mkern1mu\underline{\mkern-1muv\mkern-1mu}\mkern1mu}
\newcommand{\wdata}{\underline{w}}
\newcommand{\Ydata}{\underline{Y\mkern-5mu}\mkern2mu}
\newcommand{\kerr}{\mathrm{K}}
\newcommand{\kerrdata}{\underline{\mathrm{K}}}

\newcommand{\Cb}{C_b}

\newcommand{\Hb}{H_{b}}

\newcommand{\Dspcl}{\tc{darkpastelpurple}{\slashed{\Delta}}}

\newcommand{\Dspdata}{\tc{darkpastelpurple}{\underline{\Delta}}}

\newcommand{\otimesRR}{\otimes_{\R}}

\newcommand{\HdataApp}{\widetilde{H}_{\textnormal{data}}}

\newcommand{\Wconv}{\tc{antiquefuchsia}{\mathfrak{h}}}
\newcommand{\FP}{\Phi}
\newcommand{\RFP}{\mathfrak{r}}
\newcommand{\Csob}[1]{\tc{cerulean}{c}^{#1}_{\textnormal{sob}}}

\newcommand{\pdata}{\mkern2mu\underline{\mkern-2mu\pSHS\mkern-2mu}\mkern2mu}
\newcommand{\idata}{\underline{\iSHS}}
\newcommand{\rhohdata}{\mkern1mu\underline{\mkern-1mu \rhoh\mkern-1mu}\mkern1mu}

\newcommand{\Capply}[1]{\mathscr{#1}}
\newcommand{\epsapply}[1]{\tc{blue}{\slashed{\eps}}}

\newcommand{\cx}{\mathfrak{c}}
\newcommand{\dc}{d_{\cx}}
\newcommand{\cdata}{\tc{magenta}{\underline{c}}}
\newcommand{\cc}{\tc{ballblue}{c}}
\newcommand{\dd}{\mathrm{d}}
\newcommand{\lxc}{\mathfrak{l}}
\newcommand{\dlxc}{d_{\lxc}}

\newcommand{\uinj}{\underline{\sigma\mkern-2mu}\mkern2mu }
\newcommand{\ix}{\mathfrak{i}}
\newcommand{\iC}{\smash{\slashed{\ix}}}
\newcommand{\diC}{d_{\iC}}
\newcommand{\di}{d_{\ix}}

\newcommand{\uzeta}{\smash{\underline{\smash{\zeta}}}}
\newcommand{\rel}{\textnormal{rel}}
\newcommand{\cpt}{\textnormal{cpt}}
\newcommand{\STMat}{\widehat{\mathcal{M}}}
\newcommand{\VEC}{\mathcal{V}}

\newcommand{\dR}{\textnormal{dR}}

\newcommand{\Ccone}{\tc{orange}{\mathcal{C}}}
\newcommand{\isoDIV}{\tc{blue}{\phi}^{\ix}}
\newcommand{\Charge}{\mathfrak{C}}
\newcommand{\Chargei}{\mathfrak{C}_{\iC}}
\newcommand{\KS}{\tc{orange}{\mathbf{K}}}
\newcommand{\KSdata}{\underline{\KS}}
\newcommand{\KSconst}{\tc{magenta}{\kappa}}

\newcommand{\Gc}{G_{\cx}}
\newcommand{\Gi}{G_{\ix}}
\newcommand{\Gicomp}[1]{G_{\ix,#1}}
\newcommand{\Pic}{\Pi_{\cx}}

\newcommand{\brh}{\tc{orange}{\mathbbm{Y}}}
\newcommand{\rhoh}{\mathsf{Y}}
\newcommand{\Baux}{\tc{caribbeangreen}{B}'}
\newcommand{\B}{\tc{darkorchid}{B}}
\newcommand{\Bdata}{\tc{darkorchid}{\mkern2mu\underline{\mkern-2mu B\mkern-2mu}\mkern2mu}}

\newcommand{\CB}{\tc{hcolor}{C}}
\newcommand{\rhoB}{\tc{lcolor}{R}}

\renewcommand{\ss}{\tc{orange}{s}}

\newcommand{\isoI}[1]{\phi^{\I}_{#1}}
\newcommand{\Hinew}{\tc{orange}{H_{\ix}}}
\newcommand{\Hibig}{\tilde{H}_{\ix}}

\newcommand{\GDcomp}[1]{G_{D,#1}}

\newcommand{\KER}{K_{\kappa}}
\newcommand{\tree}{\tc{magenta}{T}}
\newcommand{\Ctree}{\tc{magenta}{\tilde{T}}}
\newcommand{\pSHS}{\tc{orange}{\mathbf{p}}}
\newcommand{\iSHS}{\tc{orange}{\mathbf{i}}}
\newcommand{\hSHS}{\tc{orange}{\mathbf{h}}}
\newcommand{\pitilde}{\tc{cyan(process)}{\tilde{\pi}_0}}
\newcommand{\CBdata}{\tilde{\Bdata}}
\newcommand{\Crhohdata}{\tilde{\rhohdata}}

\newcommand{\rr}{\tc{blue}{\eps}}
\newcommand{\FPspace}{\tc{blue}{\mathcal{X}}}
\newcommand{\tildeFPspace}{\tc{blue}{\widetilde{\FPspace}}}
\newcommand{\CFP}{C_{\Phi}}
\newcommand{\CFPapply}{C_{\Phi}}

\newtheoremstyle{mytheoremstyle}
{11pt}
{11pt}
{\itshape}
{20pt}
{\bfseries}
{.}
{.5em}
{}

\newtheoremstyle{myremarkstyle}
{10pt}
{10pt}
{}
{20pt}
{\itshape}
{.}
{.5em}
{}

\newtheoremstyle{myproofstyle}
{12pt}
{12pt}
{}
{20pt}
{\bf\itshape}
{.}
{.5em}
{}

\newtheoremstyle{mycommentstyle}
{5pt}
{5pt}
{\footnotesize\sffamily\color{airforceblue}}
{0pt}
{}
{}
{.5em}
{}

\newtheoremstyle{myissuestyle}
{12pt}
{12pt}
{\itshape\color{magenta}}
{20pt}
{\bfseries}
{.}
{.5em}
{}

\newtheoremstyle{myquestionstyle}
{12pt}
{12pt}
{\color{airforceblue}}
{20pt}
{\bfseries}
{.}
{.5em}
{}

\theoremstyle{mytheoremstyle}
\newtheorem{theorem}{Theorem}
\newtheorem*{theorem*}{Theorem}
\newtheorem{definition}{Definition}
\newtheorem{lemma}{Lemma}
\newtheorem{prop}[theorem]{Proposition}
\newtheorem{corollary}[lemma]{Corollary}

\theoremstyle{myissuestyle}

\theoremstyle{myquestionstyle}

\theoremstyle{myremarkstyle}
\newtheorem{remark}{Remark}

\theoremstyle{myproofstyle}
\newtheorem*{PROOF}{Proof}

\usetikzlibrary{decorations.markings}

\newcommand{\xshift}{-1}
\newcommand{\yshift}{3}
\newcommand{\yishift}{1.5}
\newcommand{\rcircle}{0.07}
\newcommand{\xshort}{4}
\newcommand{\dx}{1.25}
\newcommand{\dy}{1}

\newcommand{\haux}[5]{\draw [shorten >=\xshort,shorten <=\xshort,#3] (#1) -- (#2) node[midway,#4] {#5};}
\newcommand{\hl}[3]{\haux{#1}{#2}{}{anchor=east,xshift=-\xshift,yshift=\yshift}{#3}}
\newcommand{\hr}[3]{\haux{#1}{#2}{}{anchor=west,xshift=\xshift,yshift=\yshift}{#3}}
\newcommand{\il}[3]{\haux{#1}{#2}{}{anchor=east,xshift=-\xshift,yshift=\yishift}{#3}}
\newcommand{\ir}[3]{\haux{#1}{#2}{}{anchor=west,xshift=\xshift,yshift=\yishift}{#3}}
\newcommand{\pl}[3]{\haux{#1}{#2}{}{anchor=east,xshift=-\xshift,yshift=\yshift}{#3}}

\newcommand{\braux}[5]{\draw [black,fill=#1] (#2) circle (#3*\rcircle) node [#4] {#5};}
\newcommand{\br}[3]{\braux{qcolor}{#1}{1}{#2}{#3}}

\title{
Solutions of the constraints with controlled decay to Kerr, including Schwartz decay
}
\author{Andrea N\"utzi}
\date{}

\begin{document}
\maketitle

\begin{abstract}
We show that to every small and decaying solution of the linearized constraint equations about Minkowski spacetime, one can add a quadratically small correction to obtain a solution of the full constraint equations. Near spacelike infinity, the correction is given by Kerr black hole initial data, up to a term that decays faster than the linearized solution, and that has Schwartz decay if the linearized solution has Schwartz decay. Using a recent result, we obtain that the solutions of the Einstein equations with these initial data admit a regular conformal compactification along null and timelike infinity. The construction is based on a right inverse (up to necessary integrability conditions) for the linearized constraint operator about Minkowski initial data obtained previously, that has optimal mapping properties relative to weighted b-Sobolev spaces, where the weights measure decay towards infinity. On an algebraic level, we show that the constraint equations can be derived using the homotopy transfer theorem, rather than using the geometric Gauss and Codazzi equations.
\end{abstract}

\tableofcontents

\hyphenation{Min-kow-ski}

\section{Introduction}

In general relativity, the constraint equations are
the necessary and sufficient conditions on the initial data
for local solvability of the Einstein equations.
In this paper we construct solutions of the constraints
that are close to trivial Minkowski initial data,
and that, towards spacelike infinity, 
decay to the initial data of a Kerr black hole spacetime.
For such initial data we showed \cite{MinkowskiPaper}
that the corresponding solutions of the Einstein equations
admit, like Minkowski spacetime itself, a regular
conformal compactification at null and timelike infinity.

Recall that on $\R^3$ there is a trivial solution of the constraints
obtained by restricting Minkowski spacetime
to the time equal to zero level set;
infinity of $\R^3$ corresponds to spacelike infinity.
Informally,
in this paper 
we show: 
\begin{quote}
\textbf{If} $\ulinsol$ is a solution, on $\R^3$, of the 
linearized constraint equations about Minkowski initial data
that is small in suitable Sobolev norms,
that decays inverse polynomially 
with rate $\delta>1$ towards spacelike infinity,
and such that the mass generated by $\ulinsol$ is positive
and the linear momentum generated by $\ulinsol$
is small compared to the mass
(the mass and linear momentum are generated by the gravitational 
self-interaction of $\ulinsol$ and are quadratic in $\ulinsol$),
\textbf{then} there exists a solution 
$ 
\ulinsol + \KSdata + \cdata
$
of the constraint equations,
where $\cdata$ decays inverse polynomially with rate $\delta+1$,
and $\KSdata$ is the initial data of a Kerr-Schild
spacetime near spacelike infinity, smoothly extended inwards on $\R^3$,
whose mass and linear momentum
are close to those generated by $\ulinsol$.
Moreover, $\KSdata$ and $\cdata$ are globally quadratically small
in $\ulinsol$.
The fixed point iteration used to construct 
the solution
is independent of the decay parameter $\delta$,
and one has:
\vspace{-1mm}
\begin{itemize}
\item 
If $\ulinsol$ has inverse polynomial decay 
then so has $\cdata$. 
\item 
If $\ulinsol$ has Schwartz decay 
then so has $\cdata$.
\item 
If $\ulinsol$ has compact support 
then so has $\cdata$.
\end{itemize}
\end{quote}
This is made precise in Theorem \ref{thm:P=0} 
at the end of this introduction, see also Theorem \ref{thm:mainL-infinity}. 
In Corollary \ref{cor:ApplicationMinkowskiPaper} we show that
the initial data produced by Theorem \ref{thm:P=0}
satisfy the assumptions of \cite[Theorem 3]{MinkowskiPaper},
which then implies that their hyperbolic evolution
yields solutions of the Einstein equations that admit a 
regular conformal compactification along null and timelike infinity.

Our main tool to prove Theorem \ref{thm:P=0} 
is a right inverse (up to necessary integrability conditions)
for the linearized constraint operator,
constructed in \cite{HomotopyPaper},
that satisfies optimal weighted estimates,
where the weights measure decay towards infinity.
Given this right inverse, the solutions of the constraints
are obtained via renormalization of charges using the Kerr parameters
and a Banach fixed-point argument similar to \cite{MaoOhTao,Thesis}.
Note that \cite{MaoOhTao} uses a different, 
Bogovskii-type right inverse that does not
satisfy optimal weighted estimates \cite[Remark 5]{HomotopyPaper},
to construct solutions of the
constraints that are identical to Kerr near infinity.

In this paper we work with the formulation of the constraints
needed for \cite{MinkowskiPaper}.
However the results could be translated to the standard formulation
in terms of the metric and the second fundamental form.
Alternatively,
there is a version of the right inverse \cite[Remark 7]{HomotopyPaper}
that can be directly used for the standard formulation.
See also upcoming work of the authors of \cite{MaoOhTao}.

Solutions of the constraints that are identical to 
Kerr initial data near spacelike infinity,
and everywhere close to Minkowski data,
were first constructed using gluing techniques \cite{ChruscielDelayGluing,CorvinoGluing,CorvinoSchoenGluing}.
For such initial data,
the hyperbolic evolution admits a regular conformal compactification 
at null and timelike infinity 
\cite{Friedrichgeodesicallycomplete,
FriedrichConformalEinsteinEvolution},
see also 
\cite{
ChruscielDelayAsymptoticallySimpleSpacetimes,
CorvinoStabPenrose,
FriedrichPeelingQuestion,
PenroseAsymptoticSimplicity}.
Here, together with \cite{MinkowskiPaper},
we generalize these results to data that only decay 
rapidly (Schwartz class) or inverse polynomially to Kerr.

Solutions of the constraints that decay inverse polynomially 
to Kerr towards spacelike infinity
were constructed recently \cite{FangSzeftelTouati}.
The construction in \cite{FangSzeftelTouati} 
is based on a simplified conformal method
and requires inverting the Euclidean Laplacian.
Decay of the solutions is then obtained by 
a procedure that cancels finitely many slowly decaying terms
in the inverse of the Laplacian,
where the number of cancellations depends on the decay of the solution.
By comparison, we use a right inverse of the 
linearized constraint operator that preserves decay towards spacelike infinity,
and thus our fixed point iteration is independent
of the decay of the solution. 
In particular we also obtain
solutions of the constraints that decay rapidly (Schwartz class)
to Kerr.
\step
\textbf{Einstein equations and constraint equations.}
We work with the framework of \cite{Thesis,MinkowskiPaper}, 
see also Remark \ref{rem:refdefs}, where the Einstein equations
are a
system of partial differential equations on the fixed background manifold\footnote{%
The notation $\diamond$ is from \cite{Thesis,MinkowskiPaper}.
One can think of $\diamond$ as a diamond-like subset 
of the Einstein cylinder, this perspective is important
in \cite{MinkowskiPaper}, but not important in this paper.
} 
$\diamond=\R^4$ with 
global coordinates $x=(x^0,x^1,x^2,x^3)$.
The Minkowski metric is $\gmink = \eta_{\mu\nu} dx^{\mu}\otimes dx^\nu$
with $\eta_{\mu\nu}=\diag(-1,1,1,1)_{\mu\nu}$.
The Einstein equations are
\begin{equation}\label{eq:MC}
\dg u + \tfrac12[u,u] = 0
\end{equation}
The unknown $u=\uo\oplus u_{\I}$ has two components:
$\uo$ corresponds to orthonormal frame and connection, 
$\uI$ to Weyl curvature. 
It is an element in a space
\[ 
\gx^1(\diamond) = \lx^1(\diamond) \oplus \I^{2}(\diamond)
\]
where $\gx^1(\diamond)$, $\lx^1(\diamond)$, $\I^{2}(\diamond)$ 
are the modules of sections, over $\diamond$, of (trivial) rank 
50, 40, 10 vector bundles $\gx^1$, $\lx^1$, $\I^2$ on $\diamond$.
Concretely:
\begin{itemize}
\item 
$\lx^1(\diamond)=\Omega^1(\diamond)\otimesRR\Kil$
with $\Omega^1(\diamond)$ the smooth differential one-forms
on $\diamond$, and $\Kil$ the ten-dimensional Lie algebra
of Killing vector fields for $\gmink$.
The orthonormal frame associated to $\uo$ is given as follows:
Expand $\uo=\sum_{i=1}^{10}\omega_i\otimes\zeta_i$ 
with $\omega_i\in\Omega^1(\diamond)$
and $\zeta_1,\dots,\zeta_{10}$ a basis of $\Kil$.
Define the $C^\infty$-linear map 
$F_{\uo}:\Gamma(T\diamond)\to\Gamma(T\diamond)$, 
$X\mapsto \sum_{i=1}^{10} \omega_i(X)\zeta_i$.
If $\one+F_{\uo}$ is pointwise invertible, 
then it is an orthonormal frame for the metric 
$g$ on $\diamond$ defined by 
\begin{equation}\label{eq:g}
\smash{g^{-1} = (\one + F_{\uo})^{\otimes2} \gmink^{-1} }
\end{equation}
If $\uo=0$ then $g=\gmink$.
We do not make explicit the associated connection.
\item 
$\I^2(\diamond)$ is the rank 10 submodule of 
$\dens{-1}(\diamond)\otimesCinf 
S^2(\Omega^2(\diamond))$ of all sections that
satisfy the algebraic symmetries and traceless condition 
(relative to $\gmink$) of Weyl tensors.
Here $S^2$ is the symmetric tensor product over $C^\infty$, 
and $\dens{-1}(\diamond)$ are the sections 
of the minus one density bundle on $\diamond$\footnote{%
In our convention one integrates 4-densities on $\diamond$.}.
\end{itemize}
The operators $\dg$ and $[\cdot,\cdot]$ are linear respectively
bilinear first order differential operators, where $\dg$ describes
linearized gravity about Minkowski, $[\cdot,\cdot]$ describes
the gravitational interaction.
The left hand side of \eqref{eq:MC}
takes values in a space
\begin{equation}\label{eq:g2}
\smash{\gx^2(\diamond) = \lx^2(\diamond) \oplus \I^{3}(\diamond)}
\end{equation}
The bundles $\gx^2$, $\lx^2$, $\I^{3}$ have rank 76, 60, 16.

If $u$ solves \eqref{eq:MC} and $\one+F_{\uo}$ is invertible, then the metric \eqref{eq:g}
is Ricci-flat.
Informally, the component of \eqref{eq:MC} in $\lx^2(\diamond)$
requires that the connection be torsion-free
and that the Riemann curvature be equal to the Weyl curvature;
the component of \eqref{eq:MC} in $\I^{3}(\diamond)$ is the equation of motion
for the Weyl curvature.

Let $\diamonddata$ be the $x^0=0$ time slice on $\diamond$, 
\[ 
\diamonddata \;=\; \diamond \cap \{x^0=0\}
\]
Then $\diamonddata=\R^3$ with coordinates $\vec{x}=(x^1,x^2,x^3)$.
Denote by $\gxdata^1$ the pullback bundle of $\gx^1$
under the inclusion map $\diamonddata\hookrightarrow\diamond$,
and by $\gxdata^1(\diamonddata)$ its sections over $\diamonddata$.
We use analogous notation for $\lx^1$, $\I^2$.
Initial data for \eqref{eq:MC} is given by a section
\begin{equation}\label{eq:gli}
\udata\;\in\;
\gxdata^1(\diamonddata) = \lxdata^1(\diamonddata)\oplus\Idata^2(\diamonddata)
\end{equation}
that solves the constraint equations, which
are the necessary and sufficient condition on $\udata$
for local existence of a solution to \eqref{eq:MC} 
that restricts to $\udata$ along $\diamonddata$.
They are a nonlinear first order partial differential equation
\begin{equation}\label{eq:Constraints_Eq}
\Pconstraints(\udata) \;=\; 0
\end{equation}
where 
\begin{equation}\label{eq:Constraints_Pdef}
\Pconstraints(\udata) 
	\;=\;
	\left( dx^0 + \anchorg(u)(x^0) \right) \left( \dg u + \tfrac12[u,u] \right)|_{\diamonddata}
\end{equation}
On the right hand side,
$u\in\gx^1(\diamond)$ is any element that satisfies $u|_{\diamonddata}=\udata$
(the map $\Pconstraints$ is independent of the choice of $u$).
The factor $dx^0 + \anchorg(u)(x^0)$ is a differential one-form on $\diamond$, the anchor map $\anchorg$ is $C^\infty$-linear in the argument $u$.
This one-form is multiplied with $\dg u + \tfrac12[u,u]\in\gx^2(\diamond)$,
using the multiplication
\begin{align}\label{eq:mult}
\Omega^1(\diamond)\times \gx^2(\diamond) \to \gx^3(\diamond)
\end{align}
essentially given by the wedge product,
where $\gx^3(\diamond)$ 
is again the space of sections of a trivial vector bundle.
The resulting section in $\gx^3(\diamond)$ is then restricted to $\diamonddata$.
In this way, \eqref{eq:Constraints_Pdef} defines a cubic,
first order partial differential operator
$$
\Pconstraints:\;\gxdata^1(\diamonddata)\to\gxdata^3(\diamonddata)
$$

Minkowski initial data is the trivial solution 
$\udata=0$ of \eqref{eq:Constraints_Eq}.
The linearization of $\Pconstraints$ about $\udata=0$ is 
the linear first order partial differential operator
\begin{equation}\label{eq:linconstraints}
\linPconstraints (\udata) =  dx^0 (\dg u)|_{\diamonddata}
\end{equation}
Solutions of $\dg u=0$, hence of $\linPconstraints (\udata)=0$,
can be constructed using Fourier integrals over the momentum 
space light cone \cite[Section 4.2.4]{Thesis}.
\begin{remark}\label{rem:refdefs}
For the full definition of the bundles $\gx^k$, $\lx^k$, $\I^k$,
the operators $\dg$, $[\cdot,\cdot]$, $\anchorg$
and the multiplication \eqref{eq:mult},
see \cite[Introduction, Section 2.2, Theorem 4]{MinkowskiPaper};
for the definition of  $\gxdata^k$, $\lxdata^k$, $\Idata^k$
see \cite[Section 2.5, Definition 9]{MinkowskiPaper};
for the definition of $\Pconstraints$,
including independence of the choice of extension of $\udata$,
see \cite[Section 2.5, Definition 10, Lemma 5]{MinkowskiPaper};
for the relation of solutions of \eqref{eq:MC} to Ricci-flat metrics
see \cite[Introduction, Section 2.4, Proposition 5]{MinkowskiPaper}.
\end{remark}
\textbf{Bases and norms.}
We use a basis of 
$\gxdata^k(\diamonddata)
	=\lxdata^k(\diamonddata)\oplus\Idata^{k+1}(\diamonddata)$ 
that is homogeneous to leading order near infinity,
see Lemma \ref{lem:basisgdata}.
For example, the basis of $\lxdata^1(\diamonddata)$
is given by 
$dx^{\mu}/\xvecjap\otimes B^{\alpha\beta}$, $dx^{\mu}\otimes T_{\alpha}$
where 
$\xvecjap=\sqrt{1+|\vec{x}|^2}$ with $|\vec{x}|^2=\sum_{i=1}^3(x^i)^2$,
and where $B^{\alpha\beta}$, $T_{\alpha}$ are the boosts and translations \eqref{eq:BT}.
We use weighted b-Sobolev norms \cite{MelroseGreenBook}
for $\ss\in\Z_{\ge1}$ and $\delta\in\R$, given by 
\[ 
\smash{\|\udata\|_{\Hb^{\ss,\delta}(\diamonddata)}
=
\|\udata_0\|_{\Hb^{\ss,\delta}(\diamonddata)}
+
\|\udata_{\I}\|_{\Hb^{\ss-1,\delta}(\diamonddata)}
\qquad
\udata=\udata_0\oplus\udata_{\I}\in\gxdata^k(\diamonddata)}
\]
Here the norms of $\udata_0$ and $\udata_{\I}$ are
defined by differentiating their components with respect to the 
fixed bases at most $\ss$ respectively $\ss-1$
times relative to $\xvecjap\p_{x^i}$,
then multiply with $\xvecjap^{\delta}$, then take the
$L^2$-norm using the measure $dx^1dx^2dx^3/\xvecjap^3$.
The \smash{$\Cb^{\ss,\delta}$}-norms are defined similarly.
See Definition \ref{def:normsdata}.
\step
\textbf{Charges.}
Let $\udata\in\gxdata^1(\diamonddata)$ be a solution
of the linearized constraints $\linPconstraints(\udata)=0$,
and assume that
\smash{$\|\udata\|_{\Hb^{5,\delta}(\diamonddata)}<\infty$} 
for some $\delta>1$.
The nonlinear self-interaction of $\udata$ generates ten charges,
defined as follows.
Since $\udata$ solves the linearized constraints, 
it can be extended to $\diamond$ as a solution of the 
linearized Einstein equations, c.f.~\eqref{eq:d(1-hd)}. 
That is, there exists $u\in\gx^1(\diamond)$ that satisfies 
\begin{equation}\label{eq:extensionudata0}
u|_{\diamonddata}=\udata
\qquad\qquad
\dg u=0
\end{equation}
The self-interaction of $u$ is given by 
$-\frac12[u,u]\in \gx^2(\diamond)$.
Let $U_{\I}$ be its component in $\I^3(\diamond)$.
We have 
$U_{\I} = U_+\oplus U_-$ with $U_{\pm}\in \I^3_\pm(\diamond)$
\cite[Definition 2]{MinkowskiPaper}.
Here 
$$
\I_+^3(\diamond) 
\;\subset\;
\dens{-1}(\diamond)\otimesCinf
\OmegaC^3(\diamond)\otimesCinf 
\Omega^2_+(\diamond)$$
where $\OmegaC^k(\diamond)$ are the smooth complex differential
$k$-forms on $\diamond$,  
and $\Omega_+^2(\diamond)\subset\OmegaC^2(\diamond)$ 
are the self-dual forms relative to the conformal metric $[\gmink]$.
The pullback of $U_+$ along the inclusion map 
$\diamonddata\hookrightarrow\diamond$ yields an element in 
\begin{align}
\label{eq:isohodge}
\dens{-1}(\diamonddata)
\otimesCinf
\OmegaC^3(\diamonddata) 
\otimesCinf \OmegaC^2(\diamonddata)
\;\;\simeq\;\;
\OmegaC^3(\diamonddata) 
\otimesCinf \OmegaC^1(\diamonddata)
\end{align}
where $\OmegaC^k(\diamonddata)$ are the smooth 
complex differential $k$-forms on $\diamonddata$, where $\dens{-1}(\diamonddata)$
are the minus one densities on $\diamonddata$
(the pullback of densities uses $[\gmink]$),
and we use the canonical isomorphism
$\dens{-1}(\diamonddata)\otimesCinf\OmegaC^2(\diamonddata)
\simeq \OmegaC^1(\diamonddata)$
given by the Hodge star operator associated to $[\gminkdata]$
with $\gminkdata = \sum_{i=1}^3 dx^i\otimes dx^i$.
Given the pullback of $U_+$ to \eqref{eq:isohodge},
and using the conformal Killing vector fields for $\gminkdata$ in \eqref{eq:confKilS3},
one obtains ten real three-forms on $\diamonddata$ 
in the following way:
Into the second factor $\OmegaC^1(\diamonddata)$ in \eqref{eq:isohodge}
insert $\underline{S}$ (scaling) and then take the real part;
insert $\underline{L}^a$ (rotations)
and take minus the imaginary part,
insert $\underline{K}^a$ (special conformal transformations) 
and take half of the real part;
insert $\underline{K}^a$ and take minus half of the imaginary part.
The integrals of these ten three-forms over $\diamonddata$ 
are, by definition, the ten charges generated by the self-interaction of $\udata$.
By Remark \ref{rem:B2charge} 
they are independent of the choice of extension 
$u$ in \eqref{eq:extensionudata0},
and the integrals converge by Lemma \ref{lem:ChargeEst}, \eqref{eq:brdataest}.

The charges are equivalently given by 
\begin{equation}\label{eq:CPuu}
\smash{\Charge(\pSHS(-\tfrac12[u,u]))\;\in\;\R^{10}}
\end{equation}
The map $\pSHS$ is the pullback of forms and densities, 
see Definition \ref{def:pdef}.
It is part of a contraction from the complex $(\gx(\diamond),\dg)$  
on $\diamond$
to a complex on $\diamonddata$, see \eqref{eq:pih}.
The homology of this complex 
(restricted to decaying sections) 
in degree two is 10-dimensional,
and $\Charge$
establishes an isomorphism with $\R^{10}$, see Definition \ref{def:MLCA}.
\step
\textbf{Kerr near spacelike infinity.}
The family of Kerr-Schild spacetimes may be viewed as a family
of solutions of \eqref{eq:MC} defined near spacelike infinity.
Here we use an extension of these elements to $\diamond$,
using a smooth cutoff.
This construction is in Appendix \ref{sec:KSreparametrization}
and yields the following:
For $r_1,r_2>0$ define the cone%
\begin{equation}\label{eq:Ccone}
\Ccone_{r_1,r_2} 
	= 
	\big\{ (m,q)\in (0,r_1)\times \R^{9} \mid |q|< r_2 m \big\}
\end{equation}
There exists $M\in (0,2^{-11}]$ and a map
\begin{equation}\label{eq:KS}
\KS:\;\; 
\Ccone_{M,2^{-6}} 
\; \to \; 
\gx^1(\diamond)
\end{equation}
with the following properties, where $z\in\Ccone_{M,2^{-6}}$,
and where $y$ are coordinates near spacelike infinity
defined by Kelvin inversion $y=x/(-(x^0)^2+|\vec{x}|^2)$:
\begin{enumerate}[label=\textnormal{({K\arabic*})}]
\item \label{item:KSMC}
$\KS(z)$ solves \eqref{eq:MC} in the neighborhood of 
spacelike infinity given by $|y|\le \tfrac{1}{101}$.
\item \label{item:LKSsmoothnullinf}
$\KS(z)$ extends smoothly to future and past null infinity
(in the sense of \cite{Thesis,MinkowskiPaper}, 
in particular by viewing $\diamond$
as a subset of the Einstein cylinder).
\item \label{item:Kerr4dest}
For every $\ss\in\Z_{\ge0}$ one has the pointwise estimate
on $\diamond\cap\{|y|\le\frac{1}{101}\}$:
\begin{equation}\label{eq:KerrDecay}
|(|y|\p_y)^{\le\ss}\KS(z)|
\;\lesssim_{\ss}\;
|z| |y|(1+|\log|y||)
\end{equation}
which is understood as follows:
It holds for each component of $\KS(z)$
relative to the homogeneous basis in 
\cite[Definition 17]{MinkowskiPaper};
the notation $(|y|\p_y)^{\le\ss}$ means that
the components are differentiated at most $\ss$ times
with respect to $|y|\p_{y^0},\dots,|y|\p_{y^3}$.
The notation $\lesssim_{\ss}$
is explained in Remark \ref{rem:lesssim}.
\end{enumerate}
and furthermore the initial data 
$\KSdata(z)=\KS(z)|_{\diamonddata}$ satisfy:
\begin{enumerate}[label=\textnormal{({K\arabic*})},resume]
\item \label{item:IntroKSConstraints}
$\KSdata(z)$ solves the constraint equations \eqref{eq:Constraints_Eq}
on $|\vec{x}|\ge101$
(follows from \ref{item:KSMC}).
\item \label{item:IntroKSest}
For every $\ss\in\Z_{\ge1}$ there exists $\KSconst_{\ss}\ge1$
such that for all $z_1,z_2\in\smash{\Ccone_{M,2^{-6}}}$:
\begin{align}\label{eq:KSest12}
\smash{\|\KSdata(z_1)\|_{\Cb^{\ss,1}(\diamonddata)}
	 \le \KSconst_{\ss}|z_1|
	 \qquad
\|\KSdata(z_1)-\KSdata(z_2)\|_{\Cb^{\ss,1}(\diamonddata)}
	\le \KSconst_{\ss}|z_1-z_2|}
\end{align}
\item \label{item:IntroKSCharges}
$\Charge(\MCinf(\KSdata(z))) = z$, 
with $\Charge$ and $\MCinf$ in  Definition \ref{def:MLCA}
respectively \ref{def:MCinf}.
\end{enumerate}
\step

\textbf{Main Theorem.}
We use the norms in Definition \ref{def:normsdata}.
\begin{theorem}\label{thm:P=0}
For all 
\begin{align}\label{eq:P=0_parameters}
\smash{N\in\Z_{\ge8}
\qquad\qquad
\delta > 1
\qquad\qquad
b\ge1}
\end{align}
there exists $\eps\in(0,M]$ and $C\ge1$ 
such that for all 
$\ulinsol \in \Hb^{N,\delta}(\diamonddata,\gxdata^1)$:
If
\begin{enumerate}[label=\textnormal{({a\arabic*})}]
\item \label{item:MAINDPu0=0_MAINu0small}
$\linPconstraints(\ulinsol)=0$
and
\smash{$\|\ulinsol\|_{\Hb^{N,\delta}(\diamonddata)} \le \eps$}
\item \label{item:MAINu0cone}
The charge vector $z_*\in\R^{10}$ generated by $\ulinsol$,
see \eqref{eq:CPuu},
satisfies 
\begin{equation}\label{eq:z*assps}
\smash{z_*\in\Ccone_{\frac\eps2,\conenb{2^{-7}}}
\qquad\qquad
\|\ulinsol\|_{\Hb^{N,\delta}(\diamonddata)} \le b \smash{|z_*|^{\frac12}}}
\end{equation}
\end{enumerate}
Then there exist
\smash{$\cdata\in \Hb^{N,\delta+1}(\diamonddata,\gxdata^1)$}
and 
$z\in \Ccone_{\eps,\conenb{2^{-6}}}$
such that
\begin{align}\label{eq:Pu0Kv}
\Pconstraints(\ulinsol + \KSdata(z) + \cdata) &= 0
\end{align}
and
\smash{$\|\cdata\|_{\Cb^{N-2,1}(\diamonddata)}
+ \|\KSdata(z)\|_{\Cb^{N-2,1}(\diamonddata)} 
\le
\tfrac12 \|\ulinsol\|_{\Cb^{2,1}(\diamonddata)}$}
and such that:
\begin{itemize}
\item \textbf{Part 1.}
$\cdata$ and $\KSdata(z)$ are quadratically small in $\ulinsol$,
more precisely:
\begin{subequations}\label{eq:T1cKconcl}
\begin{align}
\|\cdata\|_{\Hb^{N,\delta+1}} 
	&\le C |z_*|
&
\|\cdata\|_{\Hb^{N,\delta+1}(\diamonddata)} 
	&\le C \|\ulinsol\|_{\Cb^{2,1}(\diamonddata)} \|\ulinsol\|_{\Hb^{2,\delta}(\diamonddata)}\\
|z-z_*| &\le C |z_*|^{\frac32}
&
\|\KSdata(z)\|_{\Cb^{N,1}(\diamonddata)} 
	&\le C \|\ulinsol\|_{\Cb^{2,1}(\diamonddata)} \|\ulinsol\|_{\Hb^{2,\delta}(\diamonddata)}
	\label{eq:K(z)estimqte}
\end{align}
\end{subequations}
\item 
\textbf{Part 2.}
For all $N'\in\Z_{\ge N}$, 
$\delta'\ge\delta$, $b'\ge 1$, if
\smash{$\|\ulinsol\|_{\Hb^{N',\delta'}(\diamonddata)} \le b'$} then
\begin{align}\label{eq:Part2c}
\|\cdata\|_{\Hb^{N',\delta'+1}(\diamonddata)}
  \lesssim_{N',\delta',b,b'} 
	\|\ulinsol\|_{\Hb^{N',1}(\diamonddata)}\|\ulinsol\|_{\Hb^{N',\delta'}(\diamonddata)}
\end{align}
In particular, if $\ulinsol$ decays rapidly then also $\cdata$
decays rapidly.
\item 
\textbf{Part 3.}
For all $r\ge101$, if $\ulinsol|_{|\vec{x}|\ge r}=0$
then $\cdata|_{|\vec{x}|\ge r}=0$.
\end{itemize}
\end{theorem}
The proof is in Section \ref{sec:ProofThm1}, and will 
only use the properties
\ref{item:IntroKSConstraints}, \ref{item:IntroKSest}, \ref{item:IntroKSCharges}
of $\KSdata$.
It will be proven as a corollary of Theorem \ref{thm:mainL-infinity},
which does not make explicit reference to the Kerr-Schild family $\KSdata$.
\begin{remark}\label{rem:v0u0exist}
It is easy to see that there exist linearized solutions $\ulinsol$
that satisfy the assumptions of Theorem \ref{thm:P=0}:
In fact, given $N$, $\delta$ one can construct 
$\vlinsol\in\Hb^{N,\delta}(\diamonddata,\gxdata^1)$
with $D\Pconstraints(\vlinsol)=0$
and such that the charge vector $z_*^v = (m_*^v,q_*^v)\in\R^{10}$
generated by $\vlinsol$ satisfies $m_*^v>0$, $q_*^v=0$.
This is obtained using a Fourier integral over the
momentum space light cone \cite[Section 4.2.4]{Thesis}, for which 
the charges \eqref{eq:CPuu} 
are computed in \cite[Section 4.5, 4.6]{Thesis}\footnote{
The calculations in \cite{Thesis} are for
linearized solutions that decay rapidly (Schwartz class)
towards spacelike infinity, but they may be generalized
to inverse polynomial decay.}.
To get $q_*^v=0$, use any nonzero data on the momentum space light cone
that is invariant under an appropriate 
finite subgroup of the rotation group on $\diamonddata=\R^3$,
e.g.~the tetrahedral group.
Set \smash{$b=\|\vlinsol\|_{\Hb^{N,\delta}(\diamonddata)}/|z_*^v|^{\frac12}+1$}.
Then, for all sufficiently small $\lambda>0$, 
the element $\ulinsol = \lambda\vlinsol$ (observe $z_* = \lambda^2 z_*^v$) 
satisfies the assumptions of Theorem \ref{thm:P=0}.
\end{remark}
The next corollary shows that Theorem \ref{thm:P=0} 
produces initial data for which 
\cite[Theorem 3]{MinkowskiPaper} applies.
Therefore the hyperbolic evolution of this data
admits a regular conformal compactification along null and
timelike infinity.
\begin{corollary}\label{cor:ApplicationMinkowskiPaper}
Let 
\begin{equation}\label{eq:inputT3}
N\in\Z_{\ge7}
\qquad\gamma\in(0,1]
\qquad\MinkThmsfix\in \smash{(0,\tfrac{1}{101}]}
\qquad b>0
\end{equation}
Let $\vlinsol\in\Hb^{N+6,\delta_N}(\diamonddata,\gxdata^1)$
with $\delta_N = \frac52+\gamma+N+3$ be a solution
of the linearized constraints, $D\Pconstraints(\vlinsol)=0$,
and such that the charge vector $z_*^v = (m_*^v,q_*^v)$ generated by $\vlinsol$
satisfies $|q_*^v|<\conenb{2^{-7}}m_*^v$
(such elements exist, see Remark \ref{rem:v0u0exist}).
Then there exists $\lambda_0\in(0,1]$ such that
for all $\lambda\in(0,\lambda_0]$: 
\begin{itemize}
\item 
The assumptions of Theorem \ref{thm:P=0} hold
for the parameters \eqref{eq:P=0_parameters} given by 
$N+6$, $\delta_N$,
$\|\vlinsol\|_{\Hb^{N+6,\delta_N}(\diamonddata)}/|z_*^v|^{\frac12}+1$
and for $\ulinsol=\lambda\vlinsol$.
Let $\cdata$, $z$ be as in \eqref{eq:Pu0Kv}.
\item
The assumptions of \cite[Theorem 3]{MinkowskiPaper}
hold for the parameters \cite[(32)]{MinkowskiPaper}
given by \eqref{eq:inputT3}, and 
\cite[(33)]{MinkowskiPaper} given by 
\smash{$\kerr = \KS(z)|_{\Dspcl_{\le\MinkThmsfix}}$}, 
$\udata = \ulinsol + \KSdata(z)+\cdata$.
\end{itemize}
Further, for all $N'\in\Z_{\ge N}$, if 
\smash{$\vlinsol\in\Hb^{N'+6,\delta_{N'}}(\diamonddata,\gxdata^1)$}
with $\delta_{N'} = \frac52+\gamma+N'+3$
then there exists $b'>0$ such that the assumptions of 
\cite[Part 2 of Theorem 3]{MinkowskiPaper} hold
for the parameters $k=N'$, $b'$.
\end{corollary}
The proof of Corollary \ref{cor:ApplicationMinkowskiPaper} 
is in Section \ref{sec:ProofThm1}.
\step
\newcommand{\RHS}{\tc{magenta}{\mathfrak{R}}}
\textbf{Proof sketch of Theorem \ref{thm:P=0}.} 
We use a complex\footnote{
The complex \eqref{eq:cintro} coincides with the complex
$(\gxdata,d_{\gxdata})$ in \cite[Lemma 85]{Thesis}.
Beware that in the present paper
(and \cite{MinkowskiPaper})
the notation $\gxdata$ is used for the pullback bundle of 
$\gx$ along $\diamonddata\hookrightarrow\diamond$.}
\begin{align}\label{eq:cintro}
\begin{aligned}
\begin{tikzpicture}[node distance = 10mm and 8mm, auto]
  \node (a0) {$\Omega^0(\diamonddata)\otimesRR\Kil$};
  \node (a1) [right=of a0] {$\Omega^1(\diamonddata)\otimesRR\Kil$};
  \node (a2) [right=of a1] {$\Omega^2(\diamonddata)\otimesRR\Kil$};
  \node (a3) [right=of a2] {$\Omega^3(\diamonddata)\otimesRR\Kil$};
  \node (b1) [below=of a1] {$\ix^2(\diamonddata)$};
  \node (b2) [below=of a2] {$\ix^3(\diamonddata)$};
  \draw[->] (a0) to node[above] {\footnotesize $\ddR\otimes\one$}   (a1);
  \draw[->] (a1) to node[above] {\footnotesize $\ddR\otimes\one$} (a2);
  \draw[->] (a2) to node[above] {\footnotesize $\ddR\otimes\one$} (a3);
  \draw[->] (b1) to node[above,yshift=0mm,xshift=-1mm] {\footnotesize $\di$} (b2);
  \draw[->] (b1) to node[midway,yshift=-5pt,black] 
  {\footnotesize $-\uinj$} (a2);
  \draw[->] (b2) to node[midway,yshift=-5pt,black] 
  {\footnotesize $+\uinj$} (a3);
\end{tikzpicture}
\end{aligned}
\end{align}
defined in Section \ref{sec:cdef}.
The map  
$\ddR$ is the de Rham differential, 
and, in a basis, the map $\di$ is given by two copies
of the divergence acting on symmetric traceless 3-by-3 matrices.
We denote the spaces in the four columns by 
$\cx^0(\diamonddata),\dots,\cx^3(\diamonddata)$,
e.g.~$\cx^1(\diamonddata)=
	(\Omega^1(\diamonddata)\otimesRR\Kil)\oplus\ix^2(\diamonddata)$,
and the differential by $\dc:\cx^k(\diamonddata)\to\cx^{k+1}(\diamonddata)$.

The map 
\begin{equation}\label{eq:dc12}
\dc:\cx^1(\diamonddata) \to \cx^2(\diamonddata)
\end{equation}
is related to 
$D\Pconstraints$ in \eqref{eq:linconstraints}:
relative to the basis in Lemma \ref{lem:basisgdata},
$D\Pconstraints$ is a strictly lower triangular 2-by-2 block matrix,
and \eqref{eq:dc12} is its nonzero block, see \eqref{eq:DPdc}.
The map $\dc$ loses one derivative and preserves decay
relative to the $\Hb^{N,\delta}$-norms,
see \eqref{eq:dcest}.

The space $\cx^1(\diamonddata)$ is isomorphic 
to a gauge subspace of $\gxdata^1(\diamonddata)$
via an injection $\idata$, see \eqref{eq:idatadef}. 
The element $\cdata$ in Theorem \ref{thm:P=0} is constructed 
as $\cdata = \idata \cc$ with $\cc\in\cx^1(\diamonddata)$.

In Section \ref{sec:homotopytransfer}
we equivalently rewrite \eqref{eq:Pu0Kv} 
as an equation with values in $\cx^2(\diamonddata)$.
This equation takes the form (absolutely converging sum)
\begin{equation}\label{eq:MCinfintro}
\dc \cc = - \left(
\Bdata_1(\KSdata(z))
+\textstyle \sum_{n\ge2}
\frac{1}{n!} \Bdata_n\big( (\ulinsol+\KSdata(z)+\idata\cc)^{\otimes n}\big)
\right)
\end{equation}
where each $\Bdata_n:\gxdata^1(\diamonddata)^{\otimes n} \to \cx^2(\diamonddata)$ is an $\R$-multilinear first 
order differential operator,
whose coefficients relative
to the basis that we will use are bounded.
The maps $\Bdata_n$ are defined in Proposition \ref{prop:Bproperties}.
Denote the right hand side of \eqref{eq:MCinfintro} by $\RHS$.
Since $\KSdata(z)$ solves the constraints near infinity,
$\RHS$ decays one order faster than $\ulinsol$.
To solve \eqref{eq:MCinfintro} for $\cc$ that decays
one order faster than $\ulinsol$, 
the right hand side $\RHS$ 
must satisfy two types of integrability conditions:
\begin{equation}\label{eq:Rint}
\dc \RHS = 0
\qquad
\Charge(\RHS) = 0
\end{equation}
where $\Charge$, 
which already appeared in \eqref{eq:CPuu},
is an isomorphism from the homology of \eqref{eq:cintro} 
in degree two (restricted to sections that decay) to $\R^{10}$.

The equation \eqref{eq:MCinfintro} is solved in 
Theorem \ref{thm:mainL-infinity}.
The main tool is a chain homotopy $\Gc$ for \eqref{eq:cintro}
that gains one derivative and preserves decay
relative to the $\Hb^{N,\delta}$-norms,
and further preserves compact support on large balls around the origin
(Proposition \ref{prop:Gxhomotopy}, based on \cite{HomotopyPaper}).
Applied to $\RHS$ it satisfies
\[ 
\dc(\Gc \RHS) = \RHS - I \Charge(\RHS)- \Gc(\dc \RHS)
\]
with $I:\R^{10}\to\cx^2(\diamonddata)$ injective.
In the proof of Theorem \ref{thm:mainL-infinity} we 
then solve
\begin{equation}\label{eq:FPintro}
\cc = \Gc \RHS
\qquad
\Charge(\RHS) = 0
\end{equation}
which is effectively a fixed point problem for 
$\cc$ and the Kerr parameters $z$.

A fixed point of \eqref{eq:FPintro} solves \eqref{eq:MCinfintro}
iff $\RHS$ satisfies the two integrability conditions \eqref{eq:Rint}.
The second integrability condition is contained in
\eqref{eq:FPintro}.
The first is obtained using the fact that the 
multilinear operators $\Bdata_n$ satisfy (higher) Jacobi identities,
which is the main motivation for rewriting \eqref{eq:Pu0Kv}
in the form \eqref{eq:MCinfintro}.

The equation \eqref{eq:MCinfintro} is derived as follows.
Linear symmetric hyperbolic gauge fixing yields a contraction from the complex
$(\gx(\diamond),\dg)$ to $(\cx(\diamonddata),\dc)$, 
see Section \ref{sec:SHScontraction} (based on \cite{Thesis,RTgLa1}).
The multilinear operators $\Bdata_n$ are then obtained
using homotopy transfer of the bracket $[\cdot,\cdot]$ along this contraction,
see Section \ref{sec:Bndef}.
Homotopy transfer is a standard homological algebra tool,
see e.g.~\cite{Vallette}, and the $\Bdata_n$ satisfy the higher
Jacobi identities used to show \eqref{eq:Rint} by standard theory.
The equivalence of \eqref{eq:Pu0Kv} and \eqref{eq:MCinfintro} 
is proven in Lemma \ref{lem:MCinfproperties}.

On an algebraic level, this paper thus shows that the constraint equations, 
which are normally derived using the 
(geometric) Gauss and Codazzi equations, 
can be derived independently using the (algebraic) homotopy transfer theorem,
see Remark \ref{rem:HPT}.
The constraint equations are then the Maurer-Cartan equation
in an L-infinity algebra \cite[Definition 4.1]{getzler}, 
see Remark \ref{rem:MCinLinf}.

\step
\textbf{Acknowledgement.}
Part of this work was done during my stay at EPFL in the Fall 2025.
I thank Michael Reiterer for discussions related to this project.
I thank Piotr Chru\'sciel and Erwann Delay for discussions about a related project that concerns the constraint equations on hyperbolic space.


\section{Bases and norms for initial data}
\label{sec:basisID}

We fix $C^\infty$-bases of  
$\gx^k(\diamond)$, $\lx^k(\diamond)$, $\I^{k}(\diamond)$ and of 
$\gxdata^k(\diamonddata)$, $\lxdata^k(\diamonddata)$, $\Idata^{k}(\diamonddata)$;
define the norms that we use on $\diamonddata$;
and show mapping properties of
$\dg$ and $[\cdot,\cdot]$.

We need some preliminaries from \cite{MinkowskiPaper},
see also Remark \ref{rem:refdefs}:
\begin{itemize}
\item 
Recall $\lx^k(\diamond) = \Omega^k(\diamond)\otimesRR\Kil$
with $\Kil$ the ten-dimensional Lie algebra of Killing fields
for the Minkowski metric.
We decompose $\Kil = \boosts\oplus\transl$ where 
\begin{align}\label{eq:BT}
\begin{aligned}
\boosts &= 
\SPAN_{\R}\{ B^{01},B^{02},B^{03},B^{23},B^{31},B^{12} \}\\
\transl &=
\SPAN_{\R}\{ T_0,T_1,T_2,T_3\}
\end{aligned}
\end{align}
using the boosts and translations 
$B^{\mu\nu} =  
	(x^\mu \eta^{\nu\sigma}-x^\nu \eta^{\mu \sigma})\p_{x^\sigma}$, 
	$T_\mu = \p_{x^\mu}$.
\item 
Recall the module multiplication maps
\cite[(49a), Definition 3, (55a)]{MinkowskiPaper}
\begin{align}\label{eq:modmult}
\begin{aligned}
\Omega^q(\diamond)\times\lx^k(\diamond)&\to\lx^{k+q}(\diamond)\\
\Omega^q(\diamond)\times\I^k(\diamond)&\to\I^{k+q}(\diamond)\\
\Omega^q(\diamond)\times\gx^k(\diamond)&\to\gx^{k+q}(\diamond)
\end{aligned}
\end{align}
essentially given by the wedge product.
Recall the Leibniz rules \cite[(56)]{MinkowskiPaper}%
\begin{subequations}\label{eq:Leibniz}
\begin{align}
\dl(\omega \uo) 
	&= (\ddR\omega)\uo +(-1)^{q}\omega (\dl\uI)
	\label{eq:dlLeib}\\
\dI(\omega \uI) 
	&= (\ddR\omega)\uI +(-1)^{q}\omega (\dI\uI) 
	\label{eq:dILeib}\\
\dg(\omega u) 
	&= (\ddR\omega)u +(-1)^{q}\omega (\dg u)
	\label{eq:dgLeib}\\
[u,\omega u'] 
	&= \anchorg(u)(\omega)u' + (-1)^{qk}\omega[u,u']
	\label{eq:anchorleib}
\end{align}
\end{subequations}
where 
$\omega\in\Omega^q(\diamond)$,
$\uo\in\lx^k(\diamond)$,
$\uI\in\I^k(\diamond)$,
$u\in\gx^k(\diamond)$,
$u'\in\gx^{k'}(\diamond)$,
and where $\ddR$ is the de Rham differential.
\item 	
Recall $\I^k(\diamond) 
= (\I^k_+(\diamond)\oplus\I^k_-(\diamond))_\R$
in \cite[Definition 2]{MinkowskiPaper}.
Let $m,n$ be any two symmetric traceless 
3-by-3 matrices 
with entries in $C^\infty(\diamond)$,
and let $v,w\in C^\infty(\diamond,\R^3)$ be any two column vectors.
Then 
\begin{equation}\label{eq:isoIdef}
\isoI{2}(m\oplus n) \in \I^{2}(\diamond)
\qquad
\isoI{3}(v\oplus w) \in \I^{3}(\diamond)
\end{equation}
where we define
\begin{align*}
\isoI{2}(m\oplus n) 
&= 
-4(m_{jk}+in_{jk})
	\mu_{\gmink}^{-1}\otimes
	\theta_+^j
	\otimes 
	\theta_+^k
	\ \oplus
	\ \text{cc}
	\\
\isoI{3}(v\oplus w) &=
-4i(v_j + iw_j) 
	\mu_{\gmink}^{-1}\otimes
	\big(dx^1\wedge dx^2\wedge dx^3\otimes \theta_+^j \\
	&\qquad\qquad\qquad\qquad- \tfrac{i}{2} 
	\tsum_{\substack{b=1,2,3 \\ b\neq j}}
	dx^0\wedge dx^j\wedge dx^b\otimes \theta_+^b\big)
\ \oplus \ \text{cc}
\end{align*}
The notation $u\oplus \text{cc}$ means $u\oplus\overline{u}$;
we implicitly sum over $j,k=1,2,3$;
$\theta_+^1 = dx^0\wedge dx^1 + i dx^2\wedge dx^3$
and analogously for cyclic permutations of $123$; $\mu_{\gmink}^{-1}\in\dens{-1}(\diamond)$ 
is the $-1$ density associated to $\gmink$;
and $m_{jk},n_{jk}$ are the components of the matrices $m,n$
and $v_j,w_j$ the components of $v,w$.
\item 
Define 
$e_1=(1,0,0)$, $e_2=(0,1,0)$, $e_3=(0,0,1)$
and define 
\begin{align}\label{eq:hdef}
\begin{aligned}
h_1 &= 
\left(\begin{smallmatrix}
0&1&0\\
1&0&0\\
0&0&0
\end{smallmatrix}\right) &
h_2 &= 
\left(\begin{smallmatrix}
0&0&0\\
0&0&1\\
0&1&0
\end{smallmatrix}\right) &
h_3 &= 
\left(\begin{smallmatrix}
0&0&1\\
0&0&0\\
1&0&0
\end{smallmatrix}\right)  \\
h_4 &= 
\left(\begin{smallmatrix}
1&0&0\\
0&-1&0\\
0&0&0
\end{smallmatrix}\right) &
h_5 &= 
\tfrac{1}{\sqrt{3}}\left(\begin{smallmatrix}
1&0&0\\
0&1&0\\
0&0&-2
\end{smallmatrix}\right)
\end{aligned}
\end{align}
\item Define the numbers:
\begin{equation}
\label{eq:nm}
\renewcommand{\arraystretch}{1.1}
\begin{array}{c|ccccc}
k & 0 & 1 & 2 & 3 & 4 \\
\hline
\ngG_k^{\Omega} & 1&3&3&1&0\\
\ngG_k^{\I} &0&0&10&6&0\\
\ngG_k & 10 & 40 & 36 & 10 & 0\\
\end{array} 
\qquad
\begin{array}{c|ccccc}
k & 0 & 1 & 2 & 3 & 4 \\
\hline
\ng_k^{\Omega} & 1&4&6&4&1\\
\ng_k^{\I} &0&0&10&16&6\\
\ng_k & 10 & 50 & 76 & 46 & 10\\
\end{array} 
\end{equation}
e.g.~$\ng_2=76$.
Observe
$\ng_k^{\Omega}=\ngG_k^{\Omega}+\ngG_{k-1}^{\Omega}$,
$\ng_k^{\I}=\ngG_k^{\I}+\ngG_{k-1}^{\I}$,
$\ng_k=\ngG_k+\ngG_{k-1}$.
\end{itemize}

Define the submodules 
\begin{equation}\label{eq:GaugeSpaces1}
\OmegaG^k(\diamond)\subset\Omega^k(\diamond)
\qquad
\IG^k(\diamond)\subset\I^k(\diamond)
\qquad
\gxG^k(\diamond)\subset\gx^k(\diamond)
\end{equation}
where $\OmegaG^k(\diamond)$ is the $C^\infty(\diamond)$-span
of all $dx^{i_1}\wedge\cdots\wedge dx^{i_k}$, $1\le i_1<\cdots<i_k\le3$;
further $\IG^k(\diamond)=\image(\isoI{k})$ for $k=2,3$
and $\IG^4(\diamond)=\{0\}$; further 
\begin{equation}\label{eq:gGsum}
\gxG^k(\diamond)
= (\OmegaG^k(\diamond)\otimesRR\Kil)\oplus\IG^{k+1}(\diamond)
\end{equation}
The definition of these modules is motivated by gauge fixing,
see Lemma \ref{lem:pih}.

The following basis elements are
homogeneous to leading order near spacelike infinity,
relative to the $\R_+$-action in \cite[Section 2.3]{MinkowskiPaper}.
Recall $\xvecjap=\sqrt{1+|\vec{x}|^2}$.

\newcommand\Tstrut{\rule{0pt}{3ex}}         
\newcommand\Bstrut{\rule[-1.8ex]{0pt}{0pt}}   
\begin{table}[t]
\def\arraystretch{1.3}
\centering
\footnotesize{
\setlength{\tabcolsep}{2pt}
\begin{tabular}{c|cl|c}
Space & Basis elements & Range of indices & \#  \\
\hline\hline
$\OmegaG^k(\diamond)$
& 
$
\tfrac{dx^{i_1}}{\xvecjap} \wedge \cdots\wedge \tfrac{dx^{i_k}}{\xvecjap}$
\Tstrut\Bstrut
&
$1\le i_1<\cdots<i_k \le 3$
& $\ngG^{\Omega}_k$\\
\hline
$\OmegaG^k(\diamond)\otimesRR\boosts$
& 
$
\tfrac{dx^{i_1}}{\xvecjap} \wedge \cdots\wedge \tfrac{dx^{i_k}}{\xvecjap}
\otimes B^{\mu\nu}$
&
$\begin{aligned}
&1\le i_1<\cdots<i_k \le 3\\[-1mm]
&0\le \mu < \nu \le 3
\end{aligned}$
& $6\ngG^{\Omega}_k$\\
\hline
$\OmegaG^k(\diamond)\otimesRR\transl$
&
$\xvecjap
\tfrac{dx^{i_1}}{\xvecjap} \wedge \cdots\wedge \tfrac{dx^{i_k}}{\xvecjap}
\otimes T_\mu$
&
$\begin{aligned}
&1\le i_1<\cdots<i_k \le 3\\[-1mm]
&0\le \mu \le 3
\end{aligned}$
& $4\ngG^{\Omega}_k$\\
\hline
$\IG^2(\diamond)$
&
$\isoI{2}(\tfrac{1}{\xvecjap^2} h_\ell \oplus 0)$, 
$\isoI{2}(0\oplus \tfrac{1}{\xvecjap^2} h_\ell)$
&
$1\le \ell \le 5$
& $\ngG^{\I}_2$\\
$\IG^3(\diamond)$
&
$\isoI{3}(\tfrac{1}{\xvecjap^3} e_\ell\oplus 0)$,
$\isoI{3}(0\oplus \tfrac{1}{\xvecjap^3} e_\ell)$
&
$1\le \ell \le 3$
& $\ngG^{\I}_3$\\
$\IG^4(\diamond)$
& - &&
$0$\\
\hline
$\gxG^k(\diamond)$
&
\multicolumn{2}{c|}{combine basis elements in 
previous three rows, using \eqref{eq:gGsum}}
&
$\ngG_k$\\
\hline\hline
\multirow{2}{*}{$\Omega^k(\diamond)$}
&
\multicolumn{2}{c|}{basis elements of $\OmegaG^k(\diamond)$,}
& \multirow{2}{*}{$\ng^{\Omega}_k$} \\[-1mm]
&  
\multicolumn{2}{c|}{$dx^0/\xvecjap$ times
basis elements of $\OmegaG^{k-1}(\diamond)$}
&\\
\hline
%
\multirow{2}{*}{$\Omega^k(\diamond)\otimesRR\boosts$}
&
\multicolumn{2}{c|}{basis elements of $\OmegaG^k(\diamond)\otimesRR\boosts$,}
& \multirow{2}{*}{$6\ng^{\Omega}_k$} \\[-1mm]
&  
\multicolumn{2}{c|}{$dx^0/\xvecjap$ times
basis elements of $\OmegaG^{k-1}(\diamond)\otimesRR\boosts$}
&\\
\hline
%
\multirow{2}{*}{$\Omega^k(\diamond)\otimesRR\transl$}
& 
\multicolumn{2}{c|}{basis elements of $\OmegaG^k(\diamond)\otimesRR\transl$,}
& \multirow{2}{*}{$4\ng^{\Omega}_k$} \\[-1mm]
& 
\multicolumn{2}{c|}{$dx^0/\xvecjap$ times
basis elements of $\OmegaG^{k-1}(\diamond)\otimesRR\transl$}
& \\
\hline
$\lx^k(\diamond)$
&
\multicolumn{2}{c|}{combine basis elements in previous two rows}
& $10\ng^{\Omega}_k$\\
\hline
%
\multirow{2}{*}{$\I^k(\diamond)$}
& 
\multicolumn{2}{c|}{basis elements of $\IG^k(\diamond)$,}
&\multirow{2}{*}{$\ng^{\I}_k$}\\[-1mm]
& 
\multicolumn{2}{c|}{$dx^0/\xvecjap$ times
basis elements of $\IG^{k-1}(\diamond)$}&\\
\hline
%
\multirow{2}{*}{$\gx^k(\diamond)$}
& 
\multicolumn{2}{c|}{basis elements of $\gxG^k(\diamond)$,}
&\multirow{2}{*}{$\ng_k$}\\[-1mm]
& 
\multicolumn{2}{c|}{$dx^0/\xvecjap$ times
basis elements of $\gxG^{k-1}(\diamond)$}
&
\end{tabular}}
\captionsetup{width=115mm}
\caption{
Here $k=0\dots4$.
For example, the first row means that
$dx^{i_1}/\xvecjap \wedge 
\cdots\wedge dx^{i_k}/\xvecjap$
with $1\le i_1<\cdots<i_k \le 3$
are a $C^\infty(\diamond)$-basis of $\OmegaG^k(\diamond)$, 
and that there are $\ngG^{\Omega}_k$ basis elements.
In the second part of the table the
multiplication \eqref{eq:modmult} is used.}
\label{tab:Dbases}
\end{table}
\begin{lemma}[Bases on $\diamond$]
\label{lem:basisg}
The elements in Table \ref{tab:Dbases} are $C^\infty(\diamond)$-bases.
\end{lemma}%
\begin{proof}
By inspection, using the definition of
$\I(\diamond)$ in \cite[Definition 2]{MinkowskiPaper}.
\qed
\end{proof}

The module $\gxG^k(\diamond)$ is the space of sections
of a (trivial) bundle over $\diamond$, 
that we denote by $\gxG^k$.
We denote the pullback bundle of \smash{$\gxG^k$}
under $\diamonddata\hookrightarrow\diamond$
by \smash{$\gxdataG^k$}.
Recall that \smash{$\gxdata^k$, $\lxdata^k$, $\Idata^k$} 
are the pullback bundles of $\gx^k$, $\lx^k$, $\I^k$ respectively.

\begin{lemma}[Bases on $\diamonddata$]\label{lem:basisgdata}
For the modules of sections over $\diamonddata$
we use the same bases that we used for the modules over $\diamond$,
but restricted to $\diamonddata$.
In this sense, the following are $C^\infty(\diamonddata)$-bases:
For 
$$
\gxdataG^k(\diamonddata),\ 
\lxdata^k(\diamonddata),\ 
\Idata^k(\diamonddata),\ 
\gxdata^k(\diamonddata)
$$
use the same basis elements as for 
$\gxG^k(\diamond)$ 
$\lx^k(\diamond)$, 
$\I^k(\diamond)$, 
$\gx^k(\diamond)$ 
in Table \ref{tab:Dbases}.
\end{lemma}
\begin{proof}
Clear.\qed
\end{proof}
Let $C^\infty_c(\diamonddata)$ be the smooth 
compactly supported functions on $\diamonddata=\R^3$.
\begin{definition}[Norms on $\diamonddata$]\label{def:normsdata}
For $\ss\in\Z_{\ge0}$, $\delta\in\R$ and 
$f\in C^\infty_c(\diamonddata)$ define
\begin{align}
\|f\|_{\Hb^{\ss,\delta}(\diamonddata)}
&=
\big(\tsum_{\substack{\alpha\in\N_0^3 \\ |\alpha|\le \ss}}
\int_{\diamonddata}
|\xvecjap^{\delta} (\xvecjap \p_{\vec{x}})^{\alpha} f|^2
\frac{dx^1dx^2dx^3}{\xvecjap^3}
\big)^{\frac12} 
\label{eq:Hbnorms} \\
\|f\|_{\Cb^{\ss,\delta}(\diamonddata)}
&=
\textstyle
\sum_{\substack{\alpha\in\N_0^3 \\ |\alpha|\le \ss}}
\sup_{p\in\diamonddata}
| \big(\xvecjap^{\delta} (\xvecjap \p_{\vec{x}})^{\alpha} f\big)(p)|
\label{eq:Cbnorms}
\end{align}
where 
$
(\xvecjap \p_{\vec{x}})^{\alpha} f
=
(\xvecjap \p_{x^1})^{\alpha_1} 
(\xvecjap \p_{x^2})^{\alpha_2} 
(\xvecjap \p_{x^3})^{\alpha_3} f
$.
Let $\Hb^{\ss,\delta}(\diamonddata)$ and 
$\Cb^{\ss,\delta}(\diamonddata)$ be the completion of 
$C^\infty_c(\diamonddata)$ with respect to these norms.
Let 
$\Hb^{\ss,\delta}(\diamonddata,\lxdata^k)$ and
$\Hb^{\ss,\delta}(\diamonddata,\Idata^k)$
be the spaces of sections of $\lxdata^k$ respectively $\Idata^k$
whose components relative to the bases in Lemma \ref{lem:basisgdata}
are in $\Hb^{\ss,\delta}(\diamonddata)$, 
with norm $\|{\cdot}\|_{\Hb^{\ss,\delta}(\diamonddata)}$
given by applying \eqref{eq:Hbnorms} componentwise
and then taking the $\ell^2$-sum.
For $\ss\in\Z_{\ge1}$ define
\begin{align}\label{eq:normsoffset}
\begin{aligned}
\Hb^{\ss,\delta}(\diamonddata,\gxdata^k)
&=
\Hb^{\ss,\delta}(\diamonddata,\lxdata^k)
\oplus 
\Hb^{\ss-1,\delta}(\diamonddata,\Idata^{k+1})
\\
\|\uodata\oplus\uIdata\|_{\Hb^{\ss,\delta}(\diamonddata)}
&=
\|\uodata\|_{\Hb^{\ss,\delta}(\diamonddata)}
+
\|\uIdata\|_{\Hb^{\ss-1,\delta}(\diamonddata)}
\end{aligned}
\end{align}
Define 
$\Cb^{\ss,\delta}(\diamonddata,\lxdata^k),
\Cb^{\ss,\delta}(\diamonddata,\Idata^k),
\Cb^{\ss,\delta}(\diamonddata,\gxdata^k)$
analogously.
\end{definition}
We note that the norms are defined so that they are monotone in 
$\ss,\delta$, in the sense that if 
$\ss_1\le\ss_2$ and $\delta_1\le\delta_2$ then 
(analogously for the \smash{$\Cb^{\ss,\delta}$}-norms)
\begin{align}\label{eq:monotone}
\|\udata\|_{\Hb^{\ss_1,\delta_1}(\diamonddata)}
\le 
\|\udata\|_{\Hb^{\ss_2,\delta_2}(\diamonddata)}
\end{align}

We recall a Sobolev estimate, see e.g.~\cite[Lemma 35]{MinkowskiPaper}.
For all $\ss\in\Z_{\ge1}$ and $\delta\ge0$
there exists a constant $\Csob{\ss,\delta}\ge1$ such that for 
all $\udata\in\gxdata(\diamonddata)$:
\begin{align}\label{eq:sobolev}
\|\udata\|_{\Cb^{\ss,\delta}(\diamonddata)} 
\le
\Csob{\ss,\delta}\|\udata\|_{\Hb^{\ss+2,\delta}(\diamonddata)}
\end{align}
and where the dependency of the constant
\smash{$\Csob{\ss,\delta}$} on $\ss,\delta$ is increasing.

Set
\begin{equation}\label{eq:quotients}
\mathcal{A}=\textstyle
\SPAN_{\R}\big\{ 1,\frac{x^\mu}{\xvecjap},\frac{x^\mu x^\nu}{\xvecjap^2}
\mid \mu,\nu=0\dots3\big\}
\end{equation}
\begin{lemma}\label{lem:dglabasisNEW}
Let $b\in\gx^k(\diamond),b'\in\gx^{k'}(\diamond)$, $\theta\in\Omega(\diamond)$
be basis elements in Lemma \ref{lem:basisg},
and let $f,f'\in C^\infty(\diamond)$.
Then: 
\begin{itemize}
\item 
$\theta b$ is in the $\R$-span of the 
basis elements in Lemma \ref{lem:basisg}.
\item 
$\dg(fb) = 
		(\xvecjap\p_{x^\mu}f) (\frac{dx^\mu}{\xvecjap} b) 
		+ \dg b $.\\
Further $\dg b$ is in the 
$\mathcal{A}$-span of the basis elements in Lemma \ref{lem:basisg}.
\item 
$[fb,f'b'] = f\anchorg(b)(f')b'-(-1)^{kk'}f'\anchorg(b')(f)b + ff'[b,b']$.\\
Further $[b,b']$ is in the 
$\mathcal{A}$-span of the basis elements in Lemma \ref{lem:basisg},
and 
\begin{equation}\label{eq:dglabasisanchorNEW}
\anchorg(b)(f)
=
\tsum_{\mu=0}^3\theta_{\mu}(\xvecjap \p_{x^\mu} f)
\end{equation}
where $\theta_\mu\in\Omega^k(\diamond)$
are in the $\mathcal{A}$-span of the basis elements in Lemma \ref{lem:basisg}.
\end{itemize}
\end{lemma}
\begin{proof}
\textit{First item:} Immediate from the definitions.
\textit{Second item:}
The formula follows from \eqref{eq:dgLeib}.
Consider $\dg b$.
By definition of $\dg$, see \cite[(55b)]{MinkowskiPaper},
for $b = b_0\oplus b_{\I}$ with $b_0\in\lx^k(\diamond)$ 
and $b_{\I}\in\I^{k+1}(\diamond)$ we have
$
\dg(b_0\oplus b_{\I}) 
	= (\dl b_0 - (-1)^{k+1}\sigma(b_{\I}))\oplus\dI b_{\I}$.
Thus it suffices to show separately that 
$\dl b_0$, $\sigma(b_{\I})$, $\dI b_{\I}$
are in the $\mathcal{A}$-span of the basis elements in Lemma \ref{lem:basisg}.
For $\dl$ this follows from \eqref{eq:dlLeib},
$\dl(\xvecjap^k b_0)=0$ when $b_0\in\Omega^k(\diamond)\otimes\boosts$
respectively 
$\dl(\xvecjap^{k-1} b_0)=0$ when $b_0\in\Omega^k(\diamond)\otimes\transl$,
and $\p_{x^{\mu}}\xvecjap\in\mathcal{A}$.
Similarly for $\dI$ using \eqref{eq:dILeib} and 
$\dI (\xvecjap^{k+1}b_{\I})=0$
by \cite[(51), Remark 8]{MinkowskiPaper}.
For $\sigma$ use \cite[(54)]{MinkowskiPaper}.
\textit{Third item.}
The first formula follows from \eqref{eq:anchorleib}.
Consider $[b,b']$. First assume
$b=(\omega\otimes B^{\mu\nu})\oplus0$, 
$b'=(\omega'\otimes B^{\alpha\beta})\oplus0$
with $\omega\in\Omega^k(\diamond)$, $\omega'\in\Omega^{k'}(\diamond)$.
Then $[b,b']=U_0\oplus0$, where
\begin{align*}
U_0 
= (\omega\wedge \omega'\otimes [B^{\mu\nu},B^{\alpha\beta}]
 +\omega\wedge(\Lie_{B^{\mu\nu}}\omega')\otimes B^{\alpha\beta}
 -(\Lie_{B^{\alpha\beta}}\omega)\wedge \omega'\otimes B^{\mu\nu})
\end{align*}
by \cite[(49c)]{MinkowskiPaper}.
Each of the three terms is in the $\mathcal{A}$-span of the basis elements
in Lemma \ref{lem:basisg}:
For the first this is immediate; for the second this uses
the fact that
$\Lie_{B^{\mu\nu}}(dx^\gamma/\xvecjap)$ is 
a linear combination of $dx^0/\xvecjap,\dots,dx^3/\xvecjap$
with coefficients in $\mathcal{A}$,
and here we need the quadratic terms;
analogously for the third term.
The statement is checked analogously for 
$b=(\omega\otimes B^{\mu\nu})\oplus0$, 
$b'=(\omega'\otimes T_\alpha)\oplus0$ 
using $[B^{\mu\nu},T_\alpha]\in\transl$,
and for 
$b=(\omega\otimes T_{\mu})\oplus0$, 
$b'=(\omega'\otimes T_\alpha)\oplus0$ using $[T_\mu,T_\alpha]=0$. 
Now assume that $b=\omega\otimes B^{\mu\nu}\oplus0$ and 
$b'=0\oplus (b_{+}'\oplus b_{-}')$ 
where $b_{\pm}'\in \I_{\pm}^{k'+1}(\diamond)$.
Then
$[b,b']
=
0\oplus 
\omega(
\Lie_{B^{\mu\nu}}b_{+}'
\oplus
\Lie_{B^{\mu\nu}}b_{-}'
)$ by \cite[Definition 5]{MinkowskiPaper}.
Now use the Leibniz rule
for the Lie derivative and the formulas for $\isoI{2}$, $\isoI{3}$.
\eqref{eq:dglabasisanchorNEW}:
Write $b = (\omega\otimes\zeta)\oplus b_{\I}$
with $\omega\in\Omega^k(\diamond)$, $\zeta\in\Kil$.
Then 
$\anchorg(b)(f)=\omega\zeta(f)$,
thus the statement follows from the formulas 
for the Killing fields \eqref{eq:BT}.
\qed
\end{proof}
\begin{remark}\label{rem:lesssim}
We use the standard $\lesssim$ notation from \cite[p.~xiv]{Tao}:
If $X$ and $Y$ are two quantities then the notation $X\lesssim Y$
means that there exists a constant $C>0$ such that $X\le CY$.
If, in addition, $a_1,\dots,a_k$ are parameters then 
$X\lesssim_{a_1,\dots,a_k}Y$ means that there exists a constant
$C(a_1,\dots,a_k)>0$, that depends only on the parameters $a_1,\dots,a_k$,
such that $X\le C(a_1,\dots,a_k)Y$.
\end{remark}
The estimates in the next lemma 
are adapted to the application in Section \ref{sec:multilinearestimates}.
The proof uses, in particular, the fact that the elements in 
$\mathcal{A}|_{\diamonddata}
= \SPAN_{\R}\big\{ 1,\frac{x^i}{\xvecjap},\frac{x^i x^j}{\xvecjap^2}
\mid i,j=0\dots3\big\}$ are bounded in all 
$\Cb^{\ss,0}(\diamonddata)$-norms.
\begin{lemma}\label{lem:estpdbrdata}
Let $\ss\in\Z_{\ge1}$, $\delta,\delta_1,\delta_2\ge0$
with $\delta=\delta_1+\delta_2$ and $\udata,\udata'\in\gx(\diamond)$.
Let $u,u'\in\gx^1(\diamond)$ be the unique elements 
that satisfy $u|_{\diamonddata}=\udata$, $u'|_{\diamonddata}=\udata'$
and whose components relative to the basis in Lemma \ref{lem:basisg}
are constant in $\p_{x^0}$. Then:
\begin{subequations}\label{eq:estpdbrdata}
\begin{align}
\|\dg (u)|_{\diamonddata}\|_{\Hb^{\ss,\delta}}
&\lesssim_{\ss,\delta}
\|\udata\|_{\Hb^{\ss+1,\delta}} \label{eq:ddataest}
\\
\|\dg (u)|_{\diamonddata}\|_{\Cb^{3,\delta}}
&\lesssim_{\delta}
\|\udata\|_{\Cb^{4,\delta}} \label{eq:d4dataest}
\\
\|[u,u']|_{\diamonddata}\|_{\Hb^{\ss,\delta}}
&\lesssim_{\ss,\delta}
\|\udata\|_{\Hb^{\max\{5,\ss\},\delta_1}}
\|\udata'\|_{\Hb^{\ss+1,\delta_2}} \nonumber \\
&\qquad\qquad+
\|\udata\|_{\Hb^{\ss+1,\delta_1}}
\|\udata'\|_{\Hb^{\max\{5,\ss\},\delta_2}} \label{eq:brdataest}
\\
\|[u,u']|_{\diamonddata}\|_{\Cb^{3,\delta}}
&\lesssim_{\delta}
\|\udata\|_{\Cb^{4,\delta_1}}
\|\udata'\|_{\Cb^{4,\delta_2}} \label{eq:br3est}
\\
\|[u,u']|_{\diamonddata}\|_{\Hb^{\ss,\delta}}
&\lesssim_{\ss,\delta}
\|\udata\|_{\Cb^{\ss+1,\smash{\delta_1}}}
\|\udata'\|_{\Hb^{\ss+1,\smash{\delta_2}}} \label{eq:br4est}
\end{align}
\end{subequations}
where we abbreviate 
\smash{$\|{\cdot}\|_{\Hb^{\ss,\delta}(\diamonddata)}
	= \|{\cdot}\|_{\Hb^{\ss,\delta}}$}
	and analogously for \smash{$\Cb^{\ss,\delta}$}.
\end{lemma}
\begin{proof}
\eqref{eq:ddataest}, \eqref{eq:d4dataest}:
This follows from the first two items in Lemma \ref{lem:dglabasisNEW},
the fact that elements in $\mathcal{A}|_{\diamonddata}$
are bounded in $\Cb^{\ss+1,0}(\diamonddata)$, and the fact that, 
relative to the decomposition
$\gx^k(\diamond)=\lx^k(\diamond)\oplus\I^{k+1}(\diamond)$,
multiplication with $\frac{dx^\mu}{\xvecjap}$ is block diagonal.
\eqref{eq:brdataest}, \eqref{eq:br3est}, \eqref{eq:br4est}: 
This is easily checked using the
first and third items in Lemma \ref{lem:dglabasisNEW},
again using boundedness of elements in $\mathcal{A}|_{\diamonddata}$
and the block diagonal structure of 
multiplication with $\frac{dx^\mu}{\xvecjap}$,
and the fact that $\anchorg(0\oplus b_{\I})=0$ 
for all $b_{\I}\in\I(\diamond)$.
For \eqref{eq:brdataest} we also use the Sobolev estimate \eqref{eq:sobolev}.
\qed
\end{proof}

\section{Chain homotopy for the linearized constraints}
This section contains the linear theory used to solve \eqref{eq:MCinfintro}.
In Section \ref{sec:cdef} we recall the definition of 
the complex $(\cx(\diamonddata),\dc)$ in \eqref{eq:cintro}
from \cite[Lemma 85]{Thesis}\footnote{
Beware that in \cite[Lemma 85]{Thesis} this complex is called
$(\gxdata,d_{\gxdata})$.
In the present paper the notation 
$\gxdata$ is used for the pullback bundle of 
$\gx$ along $\diamonddata\hookrightarrow\diamond$.}.
In Section \ref{sec:cbasis} we fix a basis and norms for $\cx(\diamonddata)$.
In Section \ref{sec:Gc} we use \cite{HomotopyPaper} to construct a  
chain homotopy for $(\cx(\diamonddata),\dc)$,
which in particular provides a right inverse (up to integrability
conditions) of the linear operator in \eqref{eq:MCinfintro}.

\subsection{A complex for the linearized constraints}
\label{sec:cdef}
We define the complex $(\cx(\diamonddata),\dc)$.
The definitions will be made on $\diamonddata=\R^3$, 
we however note that the bundle $\cx$ and the 
and differential operator $\dc$ are in fact the 
restrictions of a smooth bundle respectively 
differential operator on the three-sphere $S^3$. 
See \cite[Lemma 85]{Thesis} for definitions on $S^3$. 

Let $\Omega^k(\diamonddata)$ and $\OmegaC^k(\diamonddata)$ 
be the smooth real respectively complex differential 
$k$-forms on $\diamonddata$.
Recall that $\Kil$ is the 
ten-dimensional Lie algebra of Killing fields for the Minkowski
metric.
Let $\dens{s}(\diamonddata)$ be the 
module of sections of the $s$-density bundle on $\diamonddata$,
where we use the convention that on 
$\diamonddata$ one integrates $3$-densities.
We denote by $\mu^{s}_{\gminkdata}$
the $s$-density associated to $\gminkdata$, 
defined by $\mu_{\gminkdata}^s(\p_{x^1}\wedge\p_{x^2}\wedge\p_{x^3})=1$.
\begin{definition}\label{def:lxdef}
Define $\lxc(\diamonddata) = \oplus_{k=0}^3\lxc^k(\diamonddata)$
where 
$$\lxc^k(\diamonddata)=\Omega^k(\diamonddata)\otimesRR\Kil$$
Define $\dlxc:\lxc^k(\diamonddata)\to\lxc^{k+1}(\diamonddata)$
by 
$\dlxc(\omega\otimes\zeta) = \ddR\omega\otimes\zeta$
with $\ddR$ the de Rham differential.
\end{definition}
\begin{definition}\label{def:icdef}
Define the following $C^\infty(\diamonddata,\C)$-modules:\footnote{
In \cite{Thesis} we use the notation
$\underline{\I}_{\pm}^k$ instead of $\iC^k$,
see \cite[Definition 22]{Thesis}.}
\begin{itemize}
\item 
$\iC^2(\diamonddata)
  \subset \dens{-1}(\diamonddata)\otimesCinf S^2(\OmegaC^2(\diamonddata))$,
where $S^2$ denotes the symmetric tensor product over $C^\infty$,
is given by all 
traceless elements, that is all 
$u$ with
\[ 
\textstyle\sum_{a,b,i,j=1}^3 
\eta^{ab}\eta^{ij} u(\p_{x^a},\p_{x^i},\p_{x^b},\p_{x^j}) = 0
\]
\item 
$\iC^3(\diamonddata) 
  = \dens{-1}(\diamonddata)\otimesCinf \OmegaC^3(\diamonddata)\otimesCinf\OmegaC^2(\diamonddata)$.
\end{itemize}
For $k=2,3$ let
$ 
\ix^k(\diamonddata) = \big(\iC^k(\diamonddata)\oplus\iC^k(\diamonddata)\big)_{\R}
$
be given by all elements of the form 
$\up\oplus\um$ with $\upm\in\iC^k(\diamonddata)$ and $\um=\upbar$,
where the bar is complex conjugation.
\end{definition}
\begin{remark}[Sweedler's notation]\label{rem:Sweedler}
Every element in $u \in \iC^k(\diamonddata)$ can be written 
as a finite sum of product elements, that is,
$$u = \textstyle\sum_{i=1}^n \mu_i \otimes \omega_i\otimes\omega_i'$$
for some $n\in\N$ and elements
$\mu_i\in\dens{-1}(\diamonddata)$,
$\omega_i\in\OmegaC^k(\diamonddata)$, $\omega'_i\in\OmegaC^2(\diamonddata)$.
We will abbreviate this sum by $u=\mu\otimes\omega\otimes\omega'$,
which is known as Sweedler's notation.
\end{remark}
We will use the multiplication 
$\OmegaC^1(\diamonddata)\times\iC^2(\diamonddata)\to\iC^3(\diamonddata)$
given by 
\begin{equation}
\label{eq:OmegaIcmult}
\nu u = \mu \otimes(\nu\wedge\omega)\otimes\omega'
\end{equation}
where we write $u = \mu\otimes\omega\otimes\omega'$
using Sweedler's notation in Remark \ref{rem:Sweedler}.

\begin{definition}
Define the map $\diC:\iC^2(\diamonddata)\to\iC^3(\diamonddata)$ by\footnote{%
In \cite{Thesis} we use the notation $d_{\underline{\I}}$
instead of $\diC$, see \cite[(333)]{Thesis}.
} 
\begin{equation}\label{eq:diCdef}
\diC(u)
=
	\textstyle\sum_{i=1}^3 dx^i
	\Big( 
	\mu_{\gminkdata}^{-1}\otimes \nabla^{\gminkdata}_{\p_{x^i}}(\mu_{\gminkdata}^1\otimes u) \Big)
\end{equation}
where $\nabla^{\gminkdata}$ 
is the Levi-Civita connection of $\gminkdata$,
where $\mu_{\gminkdata}^1\otimes u
\in S^2(\OmegaC^2(\diamonddata))$ using 
the canonical isomorphism
$\dens{1}(\diamonddata)\otimesCinf\dens{-1}(\diamonddata)
\simeq C^\infty(\diamonddata)$,
and where we use the multiplication \eqref{eq:OmegaIcmult}.
Define $\di:\ix^2(\diamonddata)\to\ix^3(\diamonddata)$ by
$\di(\up\oplus\um) = \diC(\up)\oplus\diC(\um)$.
\end{definition}

We relate $\di$ to the divergence
on symmetric traceless matrices.
Let 
\begin{align*}
\STMat_0&\simeq C^\infty(\diamonddata,\R^5)
	&&
	\text{symmetric traceless 3-by-3 matrices}\\
\VEC_0&= C^\infty(\diamonddata,\R^3)
	&&
	\text{column vectors with three entries}
\end{align*}
Let $\divST:\STMat_0\to\VEC_0$
be the symmetric traceless divergence given by 
\[ 
\divST(m)_j = \textstyle\sum_{i=1}^3\p_{x^i}m_{ij}
\]
\begin{lemma}\label{lem:diandDIV}
Define 
$\isoDIV_2:\STMat_0\oplus\STMat_0 \to \ix^2(\diamonddata)$
and
$\isoDIV_3:\VEC_0\oplus \VEC_0 \to \ix^3(\diamonddata)$
by 
\begin{align*}
\isoDIV_2(m\oplus n) &=
(m_{jk}+in_{jk})\eps_{jab}\eps_{kpq} 
	\mu_{\gminkdata}^{-1}\otimes dx^a\wedge dx^b \otimes dx^p\wedge dx^q
	\oplus\text{cc}\\
\isoDIV_3(v\oplus w) &=
2(v_j + iw_j) \eps_{jpq} 
	\mu_{\gminkdata}^{-1}\otimes dx^1\wedge dx^2\wedge dx^3 \otimes dx^p\wedge dx^q
	\oplus\text{cc}
\end{align*}
where the notation $u\oplus\text{cc}$ means $u\oplus\overline{u}$;
repeated indices are implicitly summed over $1,2,3$;
and $\eps_{jab}$ is the Levi-Civita symbol with $\eps_{123}=1$.
Then $\isoDIV_2$, $\isoDIV_3$ are isomorphisms of vector spaces, 
and for all $m,n\in\smash{\STMat_0\oplus\STMat_0}$:
\begin{align}\label{eq:isointertw}
\di\big(\isoDIV_2(m\oplus n)\big)
	\;=\; \isoDIV_3 \big(\divST(m)\oplus \divST(n)\big)
\end{align}
\end{lemma}
\begin{proof}
The fact that $\isoDIV_k$ indeed map to $\ix^k(\diamonddata)$,
and are isomorphisms, is immediate from Definition \ref{def:icdef}.
The identity \eqref{eq:isointertw} holds by direct calculation.
\qed
\end{proof}
We define the maps $\uinj$ in \eqref{eq:cintro}.
Let 
\begin{align*}
\confKilS \;=\; \{ \underline{\zeta} \in \Gamma(T\diamonddata) \mid
	\exists f\in C^\infty(\diamonddata):\ 
	\Lie_{\underline{\zeta}}\gminkdata = f \gminkdata \}
\end{align*}
be the ten-dimensional Lie algebra of conformal Killing fields for $\gminkdata$.
A basis is
\begin{align}\label{eq:confKilS3}
\begin{aligned}
\underline{T}_a 
	&= \p_{x^a} \\
\underline{S}   
	&= \textstyle \sum_{i=1}^3 x^i \p_{x^i}\\
\underline{L}^a 
	&= x^b \p_{x^c} - x^c \p_{x^b}\\
\underline{K}^a 
	&= \textstyle 2 x^a \sum_{i=1}^3 x^i \p_{x^i} - |\vec{x}|^2\p_{x^a}
\end{aligned}
\end{align}
where $(abc)\in\{(123),(231),(312)\}$ is a cyclic index.
They correspond to translations, conformal scaling, 
rotations, and special conformal transformations.

\begin{definition}\label{def:Hodge}
Define the $C^\infty$-linear map 
$\ast_{\iC}:\iC^k(\diamonddata) \to \OmegaC^k(\diamonddata)\otimesCinf\OmegaC^1(\diamonddata)$,
\[ 
\ast_{\iC} u
=
\omega \otimes (\ast_{\gminkdata}\omega')
\]
where \smash{$u = \mu_{\gminkdata}^{-1}\otimes \omega \otimes\omega'$}
using Sweedler's notation in Remark \ref{rem:Sweedler},
and 
where \smash{$\ast_{\gminkdata}$} is the Hodge star operator\footnote{
While $\ast_{\gminkdata}$ depends on $\gminkdata$,
the map $\ast_{\iC}$ 
only depends on the conformal metric $[\gminkdata]$.
We will not use this here, it is however important when
extending the definitions to $S^3$.} 
associated to \smash{$\gminkdata$}.
\end{definition}
For $u\in\iC^k(\diamonddata)$ and a vector field $X$ on $\diamonddata$,
we will denote by 
\begin{equation} \label{eq:insertnotation}
(\ast_{\iC} u)(X) \in \OmegaC^k(\diamonddata)
\end{equation}
the form given by inserting $X$ into the second factor of
$\ast_{\iC} u\in \OmegaC^k(\diamonddata)\otimesCinf\OmegaC^1(\diamonddata)$.
\begin{remark}\label{rem:hodgebasis}
Let $m,n\in\STMat_0$ and $v,w\in\VEC_0$.
Denote
\begin{align*}
\isoDIV_2(m\oplus n) =\up\oplus\um && 
\isoDIV_3(v\oplus w) =\up'\oplus\um'
\end{align*}
with $\upm\in\iC^2(\diamonddata)$ and $\upm'\in\iC^3(\diamonddata)$.
For every vector field $X$ on $\diamonddata$:
\begin{align*}
\begin{aligned}
(\ast_{\iC}\up)(X) 
&=
2\big((m\cdot X)_j+i(n\cdot X)_{j}\big)
	\eps_{jab}dx^a\wedge dx^b  \\
(\ast_{\iC}\up')(X) 
&=
4\big((v\cdot X) +i(w\cdot X\big))dx^1 \wedge dx^2\wedge dx^3 
\end{aligned}
\end{align*}
where we sum over $j,a,b=1,2,3$,
on the left hand sides we use the notation \eqref{eq:insertnotation},
and on the right hand sides we view $X=X^\ell \p_{x^\ell}$ as a column vector.
\end{remark}
Let $\Kil_\C=\Kil\otimesRR\C$ be the complexification of $\Kil$.
Let $\lxc_\C(\diamonddata)=\Omega_\C(\diamonddata)\otimes_{\C}\Kil_\C$.
\begin{definition}\label{def:uinj}
Define $\uinj_{\pm}: \iC^k(\diamonddata) \to
\lxc_{\C}^k(\diamonddata)$
by 
\[
\uinj_{\pm}(u) 
	= 
	\pm i (\ast_{\iC} u)(\underline{S})\otimes T_0
	+
	(\ast_{\iC} u)(\underline{L}^a)\otimes T_a
	\mp i
	(\ast_{\iC} u)(\underline{T}_a)\otimes B_{\pm}^a
\]
where we sum over $a=1,2,3$, use the notation \eqref{eq:insertnotation},
and the boosts and translations \eqref{eq:BT}
with $B^{a}_{\pm} = B^{0a} \pm i B^{bc}$
for $(abc)\in\{(123),(231),(312)\}$. 

Define $\uinj: \ix^k(\diamonddata) \to \lxc^k(\diamonddata)$ by 
$\uinj(\up\oplus\um) = \uinj_{+}(\up)+\uinj_{-}(\um)$.
\end{definition}
The maps $\uinj_{\pm}$ in the previous definition
agree with those in \cite[Lemma 84]{Thesis},
by the first item in \cite[Lemma 89]{Thesis}.
In particular, $\uinj_{\pm}$ are chain maps,
i.e.~for all $u\in\iC^2(\diamonddata)$ one has
$\uinj_{\pm}(\diC u) = (\ddR\otimes\one)(\uinj_{\pm}u)$
with $\ddR$ the de Rham differential.

We can now define the complex \eqref{eq:cintro}.
\begin{definition}\label{eq:defc}
Define $\cx(\diamonddata) = \oplus_{k=0}^3\cx^k(\diamonddata)$ where
\begin{equation}
\cx^k(\diamonddata) 
	= \lxc^k(\diamonddata) \oplus \ix^{k+1}(\diamonddata)
	\label{eq:ckdirsum}
\end{equation}
Define the $\R$-linear map
	$\dc:\cx^k(\diamonddata)\to\cx^{k+1}(\diamonddata)$ by
\begin{equation*}
\dc(\co\oplus\ci)
	=
	\left(\dlxc(\co) -(-1)^{k+1}\uinj(\ci)\right)
	\oplus
	(\di\ci)
\end{equation*}
\end{definition}

Each of
$\lxc^k(\diamonddata)$,
$\ix^k(\diamonddata)$,
$\cx^k(\diamonddata)$
is the module of sections, over $\diamonddata$, 
of a trivial vector bundle which we denote by 
$\lxc^k$, $\ix^k$, $\cx^k$ respectively.

The map $\dc$ is a differential \cite[Lemma 85]{Thesis},
$$
\dc^2=0
$$

We note that one also has a module multiplication
$\Omega^q(\diamonddata) \times \cx^k(\diamonddata)
	\to \cx^{q+k}(\diamonddata)$, 
and $\dc$ satisfies a Leibniz rule relative to this multiplication.

We explain the relation between 
$(\cx(\diamonddata),\dc)$ and $(\gx(\diamond),\dg)$,
see also Lemma \ref{lem:pih}.
See \eqref{eq:DPdc} for the 
relation between $\dc$ and $D\Pconstraints$.
\begin{definition}\label{def:pdef}
Define the $\R$-linear map 
$\pSHS: \gx^k(\diamond) \to \cx^k(\diamonddata)$ 
given by 
\begin{align}\label{eq:pullbackg}
\pSHS((\omega\otimes\zeta)\oplus(\up\oplus\um)) 
	= 
	((p^*\omega)\otimes\zeta)\oplus((p^*\up)\oplus(p^*\um)) 
\end{align}
where $\omega\otimes\zeta\in\Omega^k(\diamond)\otimesRR\Kil$,
 $\upm\in\I^{k+1}_{\pm}(\diamond)$,
and $p^*$ denotes the pullback of forms and densities\footnote{
The pullback $\dens{s}(\diamond)\to\dens{s}(\diamonddata)$
is given by $f \mu_{\gmink}^{s} \mapsto (p^*f) \mu_{\gminkdata}^{s}$,
see also \cite[Remark 8]{Thesis}.} along $\diamonddata\hookrightarrow\diamond$.
One indeed has $p^*\upm\in\iC^{k+1}(\diamonddata)$
by \cite[Lem.~83]{Thesis}.
\end{definition}
By \cite[Lemma 83, 86]{Thesis}
the map $\pSHS$ is surjective and a chain map,
\begin{equation}\label{eq:dppd}
\pSHS \dg = \dc \pSHS
\end{equation}
Since $\pSHS$ is given by pullback, it only depends
on the restriction of the input to $\diamonddata$.
Thus it restricts to a surjective map defined $\diamonddata$,
that we denote by 
\begin{equation}\label{eq:pdatadef}
\pdata:\; \gxdata^k(\diamonddata)\to\cx^k(\diamonddata)
\end{equation}
Recall 
$\gxG^k(\diamond)\subset\gx^k(\diamond)$ in \eqref{eq:gGsum}.
The map $\pdata$ restricts to a (fiberwise) isomorphism
$\gxdataG^k(\diamonddata)\xrightarrow{\sim}\cx^k(\diamonddata)$,
we denote its (fiberwise) inverse
by 
\begin{equation}\label{eq:idatadef}
\idata:\; \cx^k(\diamonddata)\xrightarrow{\sim}\gxdataG^k(\diamonddata)
\end{equation}

\subsection{Basis and norms}\label{sec:cbasis}
We fix a basis and norms for the space $\cx(\diamonddata)$.
\begin{table}[t]%
\def\arraystretch{1.3}
\centering
\footnotesize{
\begin{tabular}{c|cl|c}
Space & Basis elements & Range of indices & \smash{\#} \\
\hline
$\Omega^k(\diamonddata)\otimesRR\boosts$
& 
$
\tfrac{dx^{i_1}}{\xvecjap} \wedge \cdots\wedge \tfrac{dx^{i_k}}{\xvecjap}
\otimes B^{\mu\nu}$ 
&
$\begin{aligned}
&1\le i_1<\cdots<i_k \le 3\\[-1mm]
&0\le \mu < \nu \le 3
\end{aligned}$
& $6\ngG^{\Omega}_k$
\\
\hline
$\Omega^k(\diamonddata)\otimesRR\transl$
&
$\xvecjap
\tfrac{dx^{i_1}}{\xvecjap} \wedge \cdots\wedge \tfrac{dx^{i_k}}{\xvecjap}
\otimes T_\mu$
&
$\begin{aligned}
&1\le i_1<\cdots<i_k \le 3\\[-1mm]
&0\le \mu \le 3
\end{aligned}$
& $4\ngG^{\Omega}_k$\\
\hline
$\lxc^k(\diamonddata)$
&
\multicolumn{2}{c|}{combine basis elements in previous two rows}
& $10\ngG^{\Omega}_k$
\\
\hline
$\ix^2(\diamonddata)$
&
$\isoDIV_2(\tfrac{1}{\xvecjap^2} h_\ell \oplus 0)$, 
$\isoDIV_2(0\oplus \tfrac{1}{\xvecjap^2} h_\ell)$
&
$1\le \ell \le 5$
& 
$\ngG^{\I}_2$\\
$\ix^3(\diamonddata)$
&
$\isoDIV_3(\tfrac{1}{\xvecjap^3} e_\ell\oplus 0)$,
$\isoDIV_3(0\oplus \tfrac{1}{\xvecjap^3} e_\ell)$
&
$1\le \ell \le 3$
&
$\ngG^{\I}_3$\\
\hline
$\cx^k(\diamonddata)$
&
\multicolumn{2}{c|}{combine basis elements in previous two rows,
using \eqref{eq:ckdirsum}}
& $\ngG_k$
\end{tabular}
}
\captionsetup{width=115mm}
\caption{
Here $k=0\dots3$.
The maps $\isoDIV_2$, $\isoDIV_3$ are defined in Lemma \ref{lem:diandDIV};
$h_\ell$, $e_\ell$ are defined in \eqref{eq:hdef};
and the numbers in the last column 
are defined in \eqref{eq:nm}.}
\label{tab:cbas}
\end{table}%
\begin{lemma}\label{lem:cbasis}
The elements in Table \ref{tab:cbas} are $C^\infty(\diamonddata)$-bases.
The basis elements for $\cx^k(\diamonddata)$
are equal to the 
elements obtained by applying $\pSHS$
to the basis elements of $\gxG(\diamond)$ in Lemma \ref{lem:basisg}.
Further, relative to the basis in Lemma \ref{lem:basisgdata} and the
basis in this lemma, the map $\pdata:
\gxdata^k(\diamonddata)\to\cx^k(\diamonddata)$ in \eqref{eq:pdatadef}
is given by the matrix
\begin{equation}\label{eq:pmatrix}
\begin{pmatrix}
\one_{\ngG_k\times \ngG_k} & 0_{\ngG_k \times \ngG_{k-1}}
\end{pmatrix}
\end{equation}
and the map $\idata$ in \eqref{eq:idatadef} is given by the identity matrix.
\end{lemma}%
\begin{proof}
The fact that the given elements are indeed bases 
is immediate from the definitions and Lemma \ref{lem:diandDIV}.
The maps $\isoI{k}$ (see \eqref{eq:isoIdef}) and $\isoDIV_k$
are defined such that the pullback of $\isoI{2}(m\oplus n)$
along $\diamonddata\hookrightarrow\diamond$ equals 
$\isoDIV_2(m|_{\diamonddata}\oplus n|_{\diamonddata})$,
and analogously for $k=3$.
Using this, the remaining statements easily follow.
\qed
\end{proof}
\begin{definition}\label{def:cnorms}
For $\ss\in\Z_{\ge0}$, $\delta\in\R$ we denote by 
$\Hb^{\ss,\delta}(\diamonddata,\lxc^k)$ and 
$\Hb^{\ss,\delta}(\diamonddata,\ix^k)$
the spaces of sections of $\lxc^k$ 
respectively $\ix^k$ whose components
relative to the bases in Lemma \ref{lem:cbasis}
are in $\Hb^{\ss,\delta}(\diamonddata)$, see Definition \ref{def:normsdata},
with norm \smash{$\|{\cdot}\|_{\Hb^{\ss,\delta}(\diamonddata)}$}
given by applying \eqref{eq:Hbnorms} componentwise
and then taking the $\ell^2$-sum.
Moreover, for $\ss\in\Z_{\ge1}$ define 
\smash{$
\Hb^{\ss,\delta}(\diamonddata,\cx^k)
=
\Hb^{\ss,\delta}(\diamonddata,\lxc^k)
\oplus 
\Hb^{\ss-1,\delta}(\diamonddata,\ix^{k+1})
$},
where the norm is given by
$$\|c\|_{\Hb^{\ss,\delta}(\diamonddata)}
=
\|\co\|_{\Hb^{\ss,\delta}(\diamonddata)}
+
\|\ci\|_{\Hb^{\ss-1,\delta}(\diamonddata)}
\qquad\quad
c=\co\oplus\ci
$$
We make analogous definitions for 
$\Cb^{\ss,\delta}(\diamonddata,\lxc^k),
\Cb^{\ss,\delta}(\diamonddata,\ix^k),
\Cb^{\ss,\delta}(\diamonddata,\cx^k)$, using \eqref{eq:Cbnorms}.
\end{definition}
Recall $\mathcal{A}$ in \eqref{eq:quotients},
note that 
\smash{$\mathcal{A}|_{\diamonddata} 
=  	\SPAN_{\R}\big\{ 1,\frac{x^i}{\xvecjap},\frac{x^i x^j}{\xvecjap^2}
	\mid i,j=0\dots3\big\}$}.
\begin{lemma}\label{lem:diffbop}
Relative to the basis in Lemma \ref{lem:cbasis},
the map $\dc$ is a matrix differential operator whose entries
are in the $\mathcal{A}|_{\diamonddata}$-span of 
$1, \xvecjap\p_{x^1},\xvecjap\p_{x^2},\xvecjap\p_{x^3}$.
Moreover, for all $\ss\in\Z_{\ge1}$, $\delta\in\R$ 
and all $c\in\cx(\diamonddata)$:
\begin{align}\label{eq:dcest}
\|\dc c\|_{\Hb^{\ss,\delta}(\diamonddata)} 
\lesssim_{\ss,\delta} \|c\|_{\Hb^{\ss+1,\delta}(\diamonddata)}
\end{align}
\end{lemma}
\begin{proof}
Let $u\in\gx^1(\diamond)$ with
$u|_{\diamonddata}=\idata c$ and such that components
of $u$ relative to the basis in Lemma \ref{lem:basisg}
are constant in $\p_{x^0}$.
Then $\pSHS u = c$, thus \eqref{eq:dppd} yields 
$$
\dc c 
= \dc \pSHS u
= \pSHS \dg u
= \pdata ((\dg u)|_{\diamonddata})
$$
Relative to the bases in Lemma \ref{lem:basisgdata} and \ref{lem:cbasis},
the map $\pdata$ is given by \eqref{eq:pmatrix} 
and the map $\idata$ is the identity matrix. 
Thus the first part of the lemma holds by the 
first two items in Lemma \ref{lem:dglabasisNEW},
the estimate \eqref{eq:dcest} holds by \eqref{eq:ddataest}.
\qed
\end{proof}

\subsection{Chain homotopy}
\label{sec:Gc}

We construct a chain homotopy for the complex $(\cx(\diamonddata),\dc)$
that has optimal mapping properties relative to the weighted
b-Sobolev norms in Definition \ref{def:cnorms}.

We denote by 
\begin{equation}\label{eq:compactsec}
\lxc^k_{\cpt}(\diamonddata)
\qquad
\ix^k_{\cpt}(\diamonddata)
\qquad
\cx^k_{\cpt}(\diamonddata)
\end{equation}
the submodules of
$\lxc^k(\diamonddata),\ix^k(\diamonddata),\cx^k(\diamonddata)$
given by all sections that have compact support on $\diamonddata=\R^3$.
The maps $\dlxc,\di,\uinj,\dc$ are local, hence
they restrict to the spaces \eqref{eq:compactsec}.
Recall that \eqref{eq:compactsec} are dense in the b-Sobolev spaces
in Definition \ref{def:cnorms}.
\begin{definition}\label{def:MLCA}
Define the $\R$-linear map 
$\Charge : \cx^2_{\cpt}(\diamonddata)\to\R^{10}$
by 
\begin{align*}
\Charge(c) 
=
\big(
&\Re \big( \tint_{\diamonddata} (\ast_{\iC}\cp)(\underline{S}) \big),\ 
-\Im \big( \tint_{\diamonddata} (\ast_{\iC}\cp)(\underline{L}^j) \big)_{j=1,2,3}, \\
&\quad\tfrac12\Re \big( \tint_{\diamonddata} (\ast_{\iC}\cp)(\underline{K}^j)  \big)_{j=1,2,3},\ 
-\tfrac12\Im \big( \tint_{\diamonddata} (\ast_{\iC}\cp)(\underline{K}^j) \big)_{j=1,2,3}
\big)
\end{align*}
where we write $c=\co\oplus(\cp\oplus\cm)$
with $\co\in\lxc_{\cpt}^2(\diamonddata)$, 
$\cpm\in\iC^3_{\cpt}(\diamonddata)$,
and use the notation \eqref{eq:insertnotation},
the conformal Killing fields \eqref{eq:confKilS3},
the map $\ast_{\iC}$ in Definition \ref{def:Hodge}.
\end{definition}
\begin{lemma}\label{lem:ChargeEst}
For all $\delta>2$ and all $c\in\cx_{\cpt}^{2}(\diamonddata)$ one has
$|\Charge(c)| \lesssim_{\delta} \|c\|_{\Hb^{1,\delta}(\diamonddata)}$.%
\end{lemma}
\begin{proof}
This is a straightforward estimate using 
the basis in Lemma \ref{lem:cbasis},
Remark \ref{rem:hodgebasis}, 
and the fact that $|\underline\zeta|\lesssim \xvecjap^2$ 
for each conformal Killing field \eqref{eq:confKilS3}.\qed
\end{proof}
\newcommand{\Rhomotopy}{r_0}
\begin{prop}\label{prop:Gxhomotopy}
For every real number $\Rhomotopy>0$ there exists an $\R$-linear map
\[ 
\Gc:\;\; \cx_{\cpt}^k(\diamonddata) \;\to\;\cx_{\cpt}^{k-1}(\diamonddata)
\]
defined for every integer $k$, with the following properties:
\begin{enumerate}[label=\textnormal{({b\arabic*})}]
\item \label{item:homotopy}
The map $\Pic:\cx_{\cpt}^k(\diamonddata) \to\cx_{\cpt}^{k}(\diamonddata)$ defined by 
\begin{equation*}
\one - \Pic \;=\; \dc\Gc + \Gc\dc
\end{equation*}
satisfies $\Pic^2=\Pic$, $\Pic\dc=0$, $\dc\Pic=0$.
In particular, $\Pic$ is a projection onto a complement of the image
of $\dc$ in the kernel of $\dc$.
Furthermore: 
For $c\in\cx_{\cpt}^k(\diamonddata)$ with $k=0,3$ one has $\Pic(c)=0$;
the range of $\Pic:\cx_{\cpt}^1(\diamonddata)\to\cx_{\cpt}^1(\diamonddata)$ is infinite dimensional;
there exist $\mu_1,\dots,\mu_{10} \in \cx_{\cpt}^2(\diamonddata)$
such that 
$\mu_{\ell}|_{|\vec{x}|\ge \Rhomotopy}=0$ 
and 
$ \Charge(\mu_{\ell})_{\ell'}=\delta_{\ell\ell'}$ 
and for all 
$c\in\cx^2_{\cpt}(\diamonddata)$:
\begin{equation}\label{eq:Picexplicit}
\Pic(c) = \textstyle\sum_{\ell=1}^{10} \Charge(c)_{\ell}\mu_{\ell}
\end{equation}
\item \label{item:support}
For all $r\ge \Rhomotopy$ and $c\in\cx_{\cpt}(\diamonddata)$ one has:
$c|_{|\vec{x}|\ge r}=0  \Rightarrow  (\Gc c)|_{|\vec{x}|\ge r}=0$.
\item \label{item:estimates}
For all $k=0\dots3$, $i=1,2,3$, $\ss\in\Z_{\ge1}$, $\delta>2$
and all $c\in \cx^k_{\cpt}(\diamonddata)$:
\begin{subequations}
\begin{align}
\|\Gc c\|_{\Hb^{\ss+1,\delta}(\diamonddata)}
&\lesssim_{\ss,\delta,\Rhomotopy}
\|c\|_{\Hb^{\ss,\delta}(\diamonddata)} \label{eq:Gcest}\\
\|[\Gc,\xvecjap\p_{x^i}] c\|_{\Hb^{\ss+1,\delta}(\diamonddata)}
&\lesssim_{\ss,\delta,\Rhomotopy}
\|c\|_{\Hb^{\ss,\delta}(\diamonddata)} \label{eq:Gccomest}
\end{align}
\end{subequations}
where we use the norms in Definition \ref{def:cnorms},
and where in \eqref{eq:Gccomest} the derivatives are
taken componentwise with respect to the basis 
in Lemma \ref{lem:cbasis}.
\end{enumerate}
\end{prop}

As preparation for the proof we recall the notion of a contraction, 
see also \cite[Section 2.1]{HomotopyPaper}.
A complex $(C,\dd)$ is a 
$\Z$-graded vector space $C=\oplus_{k\in\Z}C^k$
with a differential $\dd$, which is a 
linear map $C\to C$ of degree one with $\dd^2=0$,
where degree one means that it maps $C^{k}\to C^{k+1}$
for every $k$.

A contraction from a complex $(C,\dd)$ to a complex
$(C',\dd')$ is a triple of maps
\begin{align*}
\big(
r: C \to C'\ ,\;
\ell: C' \to C\ ,\;
h: C \to C
\big)
\end{align*}
where $r,\ell$ have degree zero and $h$ has degree negative one,
such that:
\begin{subequations}\label{eq:contractionproperties}
\begin{align}
\dd \ell &= \ell \dd'&
r \dd &= \dd' r &
1-\ell r &= \dd h + h \dd& 
1-r \ell &= 0 
\end{align}
\end{subequations}
Note that $\ell$ and $r$ induce isomorphisms in 
homology that are inverse to each other, 
in particular $(C,\dd)$ and $(C',\dd')$ have isomorphic homologies. 
Further 
$(\ell r)^2=\ell r$ and $\dd'=r\dd\ell$.
Note that if $\dd'=0$ is the zero differential, then 
the map $\Pi$ defined by $\one-\Pi = \dd h + h\dd$
satisfies $\Pi=\ell r$, $\Pi^2=\Pi$, $\Pi\dc=0$, $\dc\Pi=0$.

A contraction $(r,\ell,h)$ from $(C,\dd)$ to $(C',\dd')$
that also satisfies
\begin{align}\label{eq:sidecond}
r h = 0&&
h \ell = 0 &&
h^2=0
\end{align}
is called a contraction with side conditions
(or special deformation retract \cite{crainic}).

We first give proof outline.
Define the following auxiliary spaces:
\begin{align*}
H_{\dR}&=\oplus_{k=0}^3H_{\dR}^k 
&&\text{with}&
H_{\dR}^{0}&=0,\ 
H_{\dR}^{1}=0,\ 
H_{\dR}^{2}=0,\ 
H_{\dR}^3=\R
\\
\Hibig&=\oplus_{k=2}^3\Hibig^k
&&\text{with}&
\Hibig^2 &= \ker(\di:\ix^2_{\cpt}(\diamonddata)\to\ix^3_{\cpt}(\diamonddata))
,\ 
\Hibig^3 = \R^{10}\oplus\R^{10}
\\
\tilde{H}_{\cx} &= \oplus_{k=0}^3 \tilde{H}_{\cx}^k
&&\text{with}&
\tilde{H}_{\cx}^k &= (H_{\dR}^{k}\otimesRR\Kil)\oplus \Hibig^{k+1}
\\
H_{\cx} &= \oplus_{k=0}^3 H_{\cx}^k
&&\text{with}&
H_{\cx}^k &= 0\oplus \Hinew^{k+1} \ \ 
\text{where}\ \ 
\Hinew^{2}=\Hibig^{2},\ \Hinew^{3}=\R^{10}
\end{align*}
We will see that $H_{\dR}$ is the homology of 
$(\Omega_{\cpt}(\diamonddata),\ddR)$;
$\Hibig$ is the homology of 
$(\ix_{\cpt}(\diamonddata),\di)$;
$\tilde{H}_{\cx}$ is the homology of 
$(\cx_{\cpt}(\diamonddata),\tilde{d}_{\cx})$ with differential 
$\tilde{d}_{\cx}$ defined by 
\begin{equation}\label{eq:dctildedef}
\tilde{d}_{\cx}(\co\oplus\ci) = \dlxc\co\oplus\di\ci
\end{equation}
and $H_{\cx}$ is the homology of  $(\cx_{\cpt}(\diamonddata),d_{\cx})$.
We will successively construct:
\begin{allowdisplaybreaks}
\begin{align}
&\text{A contraction 
$(P_{\dR},I_{\dR},G_{\dR})$
from $(\Omega_{\cpt}(\diamonddata),\ddR)$
to $(H_{\dR},0)$.} \label{eq:ContrdR}\\[-1mm]
&\text{This uses \cite[Theorem 2]{HomotopyPaper}.}\nonumber\\
&\text{A contraction 
$(P_{\ix},I_{\ix},G_{\ix})$
from $(\ix_{\cpt}(\diamonddata),\di)$
to $(H_{\ix},0)$.}\label{eq:Contri}\\[-1mm]
&\text{This uses \cite[Theorem 3]{HomotopyPaper}.}\nonumber\\
&\text{A contraction 
$(\tilde{P}_{\cx},\tilde{I}_{\cx},\tilde{G}_{\cx})$
from $(\cx_{\cpt}(\diamonddata),\tilde{d}_{\cx})$
to $(\tilde{H}_{\cx},0)$.}\label{eq:Contrdirectsum}\\[-1mm]
&\text{This is given by the direct sum of the contractions 
in \eqref{eq:ContrdR}, \eqref{eq:Contri}.}\nonumber\\
&\text{A contraction 
$(\slashed{P}_{\cx},\slashed{I}_{\cx},\slashed{G}_{\cx})$
from $(\cx_{\cpt}(\diamonddata),d_{\cx})$
to $(\tilde{H}_{\cx},\kappa=\slashed{P}_{\cx} \dc \slashed{I}_{\cx})$.}
\label{eq:ContrHPL}\\[-1mm]
&\text{This uses \eqref{eq:Contrdirectsum} 
and the homological perturbation lemma \cite{crainic}.}\nonumber\\
&\text{A contraction 
$(P_{0},I_{0},G_{0})$
from $(\tilde{H}_{\cx},\kappa)$ to $(H_{\cx},0)$.}
\label{eq:ContrLinalg}\\[-1mm]
&\text{This is a linear algebra argument.}\nonumber\\
&\text{A contraction 
$(P_{\cx},I_{\cx},G_{\cx})$ 
from $(\cx_{\cpt}(\diamonddata),d_{\cx})$
to $(H_{\cx},0)$.}
\label{eq:ContrFinal}\\[-1mm]
&\text{This is obtained by composing 
\eqref{eq:ContrHPL} and \eqref{eq:ContrLinalg}.}\nonumber
\end{align}
\end{allowdisplaybreaks}
The last contraction gives the homotopy $G_{\cx}$ in Proposition \ref{prop:Gxhomotopy}.
\begin{proof}
For a degree $q$ map of graded vector spaces
$f:\oplus_{k\in\Z}V^k \to \oplus_{k\in\Z}W^k$
we denote by $f_k$ the component of $f$ that maps $V^k\to W^{k+q}$,
e.g.~$\ddR_{k}$ is the component of the de Rham differential 
that maps $\Omega^k_{\cpt}(\diamonddata)\to \Omega^{k+1}_{\cpt}(\diamonddata)$.

\textit{Step 1: the contraction \eqref{eq:ContrdR}.}
Let 
$G_{\dR}: \Omega^k_{\cpt}(\diamonddata)\to\Omega^{k-1}_{\cpt}(\diamonddata)$,
$\Pi_{\dR}: \Omega^k_{\cpt}(\diamonddata)\to\Omega^{k}_{\cpt}(\diamonddata)$,
$\mu\in\Omega^3_{\cpt}(\diamonddata)$ be 
as in \cite[Theorem 2, Corollary 21]{HomotopyPaper}\footnote{%
In the reference these maps are defined on 
the relative differential forms,
but they restrict to compactly supported forms by 
\cite[(a2) of Theorem 2]{HomotopyPaper}} applied to $\Rhomotopy>0$, $n=3$.
Define 
$I_\dR:\Omega^k_{\cpt}(\diamonddata)\to H^k_{\dR}$,
$P_\dR:H_{\dR}^k\to\Omega^k_{\cpt}(\diamonddata)$
by $I_{\dR,3}(1)=\mu$ respectively $P_{\dR,3}(\omega) = \int_{\diamonddata}\omega$,
the other components of $I_{\dR}$ and $P_{\dR}$ are zero.
By \cite[(a1) of Theorem 2 and Corollary 21]{HomotopyPaper},
the triple $(P_\dR,I_\dR,G_\dR)$ is a contraction
with side conditions, see \eqref{eq:sidecond}, 
from $(\Omega_{\cpt}(\diamonddata),\ddR)$ to $(H_{\dR},0)$.

\textit{Step 2: the contraction \eqref{eq:Contri}.}
Recall the complex $(W,D)$ in \cite[(79)]{HomotopyPaper},
and let $G_D:W^k \to W^{k-1}$, $\Pi_D:W^k \to W^{k}$,
$\nu_1,\dots,\nu_{10}$ be 
as in \cite[Theorem 3]{HomotopyPaper}
applied to $\Rhomotopy>0$,
where by \cite[Remark 1, see also Corollary 21]{HomotopyPaper}
we can assume 
\begin{align}\label{eq:Dside}
\Pi_D G_D = 0
&&
G_D \Pi_D = 0
&&
G_D^2 = 0
\end{align}
Define 
$G_{\ix}:\ix^{k}_{\cpt}(\diamonddata)\to\ix^{k-1}_{\cpt}(\diamonddata)$
by, using $\isoDIV_2,\isoDIV_3$ from Lemma \ref{lem:diandDIV},
\begin{align}\label{eq:Giisoiso}
\Gicomp{3} = \isoDIV_2 \circ (\GDcomp{3}\oplus \GDcomp{3}) \circ (\isoDIV_3)^{-1}
\end{align} 
(only nontrivial component).
Define 
$P_{\ix}: \ix_{\cpt}^k(\diamonddata) \to \Hibig^k$,
$I_{\ix}: \Hibig^k \to \ix_{\cpt}^k(\diamonddata)$ as follows.
The map $I_{\ix,2}$ is the canonical inclusion
and, for $v\oplus v'\in\R^{10}\oplus\R^{10}$,
\begin{align}\label{eq:Iidef}
I_{\ix,3}(v\oplus v') 
&= 
\tsum_{\ell=1}^{10} v_{\ell}(\tilde{\nu}_\ell\oplus\tilde{\nu}_{\ell})
+
\tsum_{\ell=1}^{10} v_{\ell}'(i\tilde{\nu}_\ell\oplus(-i\tilde{\nu}_{\ell}))
\end{align}
where $\tilde{\nu}_{\ell}\in\iC_{\cpt}^3(\diamonddata)$
is defined by $\tilde{\nu}_\ell\oplus \tilde{\nu}_\ell=\isoDIV_3(\frac14\nu_{\ell}\oplus 0)$,
in particular $\overline{\tilde{\nu}_\ell}=\tilde{\nu}_\ell$.
Further $P_{\ix,2}=\one - G_{\ix,3}d_{\ix,2}$,
$P_{\ix,3}(\cp\oplus\cm) 
= \Re(\Chargei(\cp))\oplus\Im(\Chargei(\cp))$, defining
\begin{equation}\label{eq:Csldef}
\Chargei(c_+) = \big(
\tint_{\diamonddata} (\ast_{\iC}c_+)(\uzeta_1),
\dots,
\tint_{\diamonddata} (\ast_{\iC}c_+)(\uzeta_{10})
\big)\;\in\;\C^{10}
\end{equation}
with $\uzeta_{1},\dots,\uzeta_{10}$ 
the basis \eqref{eq:confKilS3} of $\confKilS$.
We note that
\begin{align} 
\Chargei(\tilde{\nu}_\ell)_{\ell'}
	&=\delta_{\ell\ell'}
	\label{eq:ortho}\\
I_{\ix,2}P_{\ix,2} 
	&= \isoDIV_2((\one - G_{D,3}D_2)\oplus(\one -G_{D,3}D_2))(\isoDIV_2)^{-1}
	\nonumber\\
	&\overset{(1)}{=} \isoDIV_2(D_1 G_{D,2}\oplus D_1 G_{D,2})(\isoDIV_2)^{-1}
	\label{eq:Pii2}\\
I_{\ix,3}P_{\ix,3} 
	&= \isoDIV_3 (\Pi_{D,3}\oplus\Pi_{D,3}) (\isoDIV_3)^{-1}
	\nonumber
\end{align}
The first holds by the orthogonality statement for the $\nu_\ell$
in \cite[(a1) of Theorem 3]{HomotopyPaper} and Remark \ref{rem:hodgebasis};
in $(1)$ we use \cite[(a1) of Theorem 3]{HomotopyPaper};
the last line uses Remark \ref{rem:hodgebasis}.
It is now a direct calculation to check that 
$(P_{\ix},I_{\ix},G_{\ix})$ is a contraction with side conditions
from $(\ix_{\cpt}(\diamonddata),\di)$ to $(\Hibig,0)$,
using \cite[(a1) of Theorem 3]{HomotopyPaper}, \eqref{eq:Dside}.

\textit{Step 3: the contraction \eqref{eq:Contrdirectsum}.}
Recall the definition of $\tilde{d}_{\cx}$ in \eqref{eq:dctildedef}.
Define the maps
$\tilde{P}_{\cx}:\cx^{k}_{\cpt}(\diamonddata)\to \tilde{H}_{\cx}^{k}$,
$\tilde{I}_{\cx}: \tilde{H}_{\cx}^{k}\to\cx^{k}_{\cpt}(\diamonddata)$,
$\tilde{G}_{\cx}:\cx^{k}_{\cpt}(\diamonddata)
\to\cx^{k-1}_{\cpt}(\diamonddata)$ by 
\begin{align*}
\tilde{P}_{\cx,k}
&=
\left(\begin{smallmatrix}
P_{\dR,k}\otimes\one &0 \\
0 & P_{\ix,k+1}
\end{smallmatrix}\right)&
\tilde{I}_{\cx,k}
&=
\left(\begin{smallmatrix}
I_{\dR,k}\otimes\one &0 \\
0 & I_{\ix,k+1}
\end{smallmatrix}\right)&
\tilde{G}_{\cx,k}
&= 
\left(\begin{smallmatrix}
G_{\dR,k}\otimes\one &0 \\
0 & G_{\ix,k+1}
\end{smallmatrix}\right)
\end{align*}
where we use the $2\times2$ block decomposition with respect to 
$\cx_{\cpt}^k(\diamonddata) 
= \lxc_{\cpt}^k(\diamonddata)\oplus\ix_{\cpt}^{k+1}(\diamonddata)$
and
$\tilde{H}_{\cx}^k = (H_{\dR}^k\otimesRR\Kil)\oplus \Hibig^{k+1}$.
By Step 1 and 2, 
the triple $(\tilde{P}_{\cx},\tilde{I}_{\cx},\tilde{G}_{\cx})$ is a contraction
with side conditions 
from \smash{$(\cx_{\cpt}(\diamonddata),\tilde{d}_{\cx})$ to $(\tilde{H}_{\cx},0)$}.

\textit{Step 4: the contraction \eqref{eq:ContrHPL}.}
Set $q=\dc-\tilde{d}_{\cx}$.
Then the triple 
\begin{align}\label{eq:contrto20}
\big(
\slashed{P}_{\cx}
=
\tilde{P}_{\cx}
\;,\;
\slashed{I}_{\cx}
=
(\one - \tilde{G}_{\cx}q)\tilde{I}_{\cx}
\;,\;
\slashed{G}_{\cx}
=
\tilde{G}_{\cx} (\one - q \tilde{G}_{\cx})\big)
\end{align}
is a contraction from $(\cx_{\cpt}(\diamonddata),\dc)$ to 
$(\tilde{H}_{\cx}, \kappa:=\slashed{P}_{\cx} \dc \slashed{I}_{\cx} )$.
This follows from the homological perturbation lemma
\cite{crainic}, see also \cite[Lemma 4]{HomotopyPaper},
which is applicable by $\dc^2=0$ and $(q\tilde{G}_{\cx})^2=0$
and the fact that the contraction in Step 3 satisfies
the side conditions. 
The formula for $\slashed{P}_{\cx}$ is as shown because 
$\tilde{P}_{\cx} (\one - q \tilde{G}_{\cx})=\tilde{P}_{\cx}$,
using 
$\tilde{P}_{\cx} q \tilde{G}_{\cx}=0$
(the image of $q \tilde{G}_{\cx}$ is only 
nontrivial in $\lxc_{\cpt}^2(\diamonddata)$, where $\tilde{P}_{\cx}$ vanishes).

We compute $\kappa=\slashed{P}_{\cx} \dc \slashed{I}_{\cx}:\tilde{H}_{\cx}^k\to\tilde{H}_{\cx}^{k+1}$.
First note that
\[ 
\smash{
\kappa 
\overset{(1)}{=}
\tilde{P}_{\cx} \tilde{d}_{\cx} \slashed{I}_{\cx} 
+
\tilde{P}_{\cx} q \slashed{I}_{\cx} 
\overset{(2)}{=}
\tilde{P}_{\cx} q \slashed{I}_{\cx} 
\overset{(3)}{=}
\tilde{P}_{\cx} q \tilde{I}_{\cx}
-
\tilde{P}_{\cx} q \tilde{G}_{\cx}q \tilde{I}_{\cx}
\overset{(4)}{=}
\tilde{P}_{\cx} q \tilde{I}_{\cx}}
\]
where 
(1) holds by definition of $\slashed{P}_{\cx}$ and $\dc=\tilde{d}_{\cx}+q$;
(2) holds because $\tilde{P}_{\cx} \tilde{d}_{\cx}=0$;
(3) holds by definition of $\slashed{I}_{\cx}$,
(4) holds because $q \tilde{G}_{\cx}q=0$.
Thus
\begin{align*}
\kappa_k
= 
\left(\begin{smallmatrix}
0 & \kappa_{k+1}'\\
0&0
\end{smallmatrix}\right)
\quad
\text{where}
\quad
\kappa_{k+1}'
:=
-(-1)^{k+1} (P_{\dR,k+1}\otimes\one)\circ \uinj_{k+1} \circ  I_{\ix,k+1}
\end{align*}
using the 2-by-2 block decomposition relative to
$\tilde{H}_{\cx}^k = (H_{\dR}^k\otimesRR\Kil)\oplus \Hibig^{k+1}$.
Note that $\kappa_{k+1}'$ maps $\Hibig^{k+1} \to H_{\dR}^{k+1}\otimesRR\Kil$.
We have $\kappa'_{2}=0$ because $P_{\dR,2}=0$.
Let $v\oplus v'\in \Hibig^3 = \R^{10}\oplus\R^{10}$.
By direct calculation, 
using the definition of $I_{\ix,3}$ in \eqref{eq:Iidef},
of $\uinj$ in Definition \ref{def:uinj},
and \smash{$P_{\dR,3}(\omega)=\int_{\diamonddata}\omega$}, 
\begin{align*}
\kappa'_3(v\oplus v')
&=
2\tsum_{\ell=1}^{10}
\int_{\diamonddata}
v_{\ell} 
\Big(
	+(\ast_{\iC} \tilde{\nu}_\ell)(\underline{L}_a)\otimes T_a
	+(\ast_{\iC} \tilde{\nu}_\ell)(\underline{T}_a)\otimes B^{bc}
	\Big)\\
&\quad
+2\tsum_{\ell=1}^{10}
\int_{\diamonddata}
v_{\ell}' 
\Big(
	- 
	(\ast_{\iC} \tilde{\nu}_\ell)(\underline{S})\otimes T_0
	+ 
	(\ast_{\iC} \tilde{\nu}_\ell)(\underline{T}_a)\otimes B^{0a}
\Big)
\intertext{
where $(abc)\in\{(123),(231),(312)\}$ is a cyclic index
and where we sum over $a=1,2,3$.
Now the orthogonality condition \eqref{eq:ortho} yields}
\kappa'_3(v\oplus v')
&=
2 \big(
v_1 \otimes B^{23}
+
v_2 \otimes B^{31}
+
v_3 \otimes B^{12}
+
v_5\otimes T_1
+
v_6\otimes T_2
+
v_7\otimes T_3\\
&
\qquad
+
v_1' \otimes B^{01}
+
v_2' \otimes B^{02}
+
v_3' \otimes B^{03}
-
v_4' \otimes T_0
\big)
\end{align*}
This is surjective, with kernel given by 
all elements 
\begin{equation}\label{eq:kernel}
\left(\begin{smallmatrix}
(0,0,0)^T\\
v_4\\
(0,0,0)^T\\
(v_8,v_9,v_{10})^T
\end{smallmatrix}\right)
\oplus
\left(\begin{smallmatrix}
(0,0,0)^T\\
0\\
(v_5',v_6',v_7')^T\\
(v_8',v_9',v_{10}')^T
\end{smallmatrix}\right)
\;\in\;
\Hibig^3=\R^{10}\oplus\R^{10}
\end{equation}

\textit{Step 5: the contraction \eqref{eq:ContrLinalg}.}
Let $\KER=\ker(\kappa'_3)$ and let 
$\KER^\bot$ be its orthogonal complement.
We identify  
\begin{align}\label{eq:kkbot}
\Hibig^3=\KER\oplus \KER^\bot
\qquad\qquad
\Hinew^3 = \KER
\end{align}
Define maps 
$I_{0}':\Hinew^k\to \Hibig^k$, 
$P_{0}':\Hibig^k\to \Hinew^k$,
$G_{0}':H_{\dR}^k\otimes \one \to \Hibig^k$ 
as follows, using the identification \eqref{eq:kkbot}:
$I_{0,2}'$ is the identity map;
$I_{0,3}'$ is the inclusion $\KER\to\KER\oplus \KER^{\bot}$;
$P_{0,2}'$ is the identity map, 
$P_{0,3}'$ is the projection $\KER\oplus \KER^{\bot}\to \KER$;
$G_{0,0}'=0$,
$G_{0,1}'=0$,
$G_{0,2}'=0$,
$G_{0,3}'=(\kappa'_3|_{\KER^{\bot}})^{-1}$.

Define 
$P_{0}:\tilde{H}_{\cx}^k\to H_{\cx}^k$,
$I_{0}:H_{\cx}^k \to \tilde{H}_{\cx}^k$,
$G_{0}:\tilde{H}_{\cx}^k \to \tilde{H}_{\cx}^{k-1}$
by 
\begin{align}\label{eq:fdcontr}
P_{0,k}
=
\left(\begin{smallmatrix}
0 & 0 \\
0 & P_{0,k+1}'
\end{smallmatrix}\right)
&&
I_{0,k}
=
\left(\begin{smallmatrix}
0 & 0 \\
0 & I_{0,k+1}'
\end{smallmatrix}\right)
&&
G_{0,k}
&=
\left(\begin{smallmatrix}
0 & 0 \\
G_{0,k}' & 0
\end{smallmatrix}\right)
\end{align}
By construction, the triple $(P_0,I_0,G_0)$
is a contraction from $(\tilde{H}_{\cx},\kappa)$ to $(H_{\cx},0)$.

\textit{Step 6: the contraction \eqref{eq:ContrFinal}.}
Compose the contraction \eqref{eq:contrto20}
from $(\cx_{\cpt}(\diamonddata),\dc)$ to $(\tilde{H}_{\cx}, \kappa )$
with the contraction \eqref{eq:fdcontr} 
from $(\tilde{H}_{\cx},\kappa)$ to $(H_{\cx},0)$,
using \cite[Lemma 3]{HomotopyPaper}.
This yields a contraction from $(\cx_{\cpt}(\diamonddata),\dc)$ to $(H_{\cx},0)$, given by
\begin{align}\label{eq:ccontr}
\big(
P_{\cx} = P_0 \tilde{P}_{\cx}\,,\,
I_{\cx} = \slashed{I}_{\cx}I_0\,,\,
G_{\cx} = \slashed{G}_{\cx} + \slashed{I}_{\cx} G_0\tilde{P}_{\cx}
\big)
\end{align}

We check that $G_{\cx}$ satisfies the properties in the proposition.

\ref{item:homotopy}:
By \eqref{eq:contractionproperties} for 
the contraction \eqref{eq:ccontr} we have
$\Pi_{\cx} = I_{\cx}P_{\cx}$ and 
$\Pi_{\cx}^2=\Pi_{\cx}$, $\Pi_{\cx}\dc=0$, $\dc\Pi_{\cx}=0$.
%
%
Further $\Pi_{\cx,0}=0$, $\Pi_{\cx,3}=0$
since $H_{\cx}^{0}$, $H_{\cx}^{3}$ are trivial.
Further 
$\Pi_{\cx}
	= I_{\cx}P_{\cx}
	= 
	(\one - \tilde{G}_{\cx}q)\tilde{I}_{\cx}I_0P_0 \tilde{P}_{\cx}$,
thus (also using $I_{0,2}'P_{0,2}'=\one$)
\begin{equation}\label{eq:pi12}
\Pi_{\cx,1}
=
	\left(\begin{smallmatrix}
	0 & - (G_{\dR,2}\otimes\one) \uinj_{2} I_{\ix,2}  P_{\ix,2}\\
	0 & I_{\ix,2}  P_{\ix,2}
	\end{smallmatrix}\right)
\qquad
\Pi_{\cx,2}
=
	\left(\begin{smallmatrix}
	0 &  (G_{\dR,3}\otimes\one) \uinj_{3} I_{\ix,3} I_{0,3}'P_{0,3}' P_{\ix,3}\\
	0 & I_{\ix,3} I_{0,3}'P_{0,3}' P_{\ix,3}
	\end{smallmatrix}\right)
\end{equation}

We check that the range of $\Pi_{\cx,1}$ is infinite dimensional.
For this it suffices to check that the range of $I_{\ix,2}  P_{\ix,2}$
is infinite dimensional, which is equal to
$\ker(\di:\ix^2_{\cpt}(\diamonddata)\to\ix^3_{\cpt}(\diamonddata))$,
because for $c$ in the kernel we have $I_{\ix,2}  P_{\ix,2}c=c$
by \eqref{eq:contractionproperties}.
Recall \eqref{eq:isointertw}. 
By \cite[(a1) in Theorem 3]{HomotopyPaper},
the kernel of $\smash{\divST}$ 
is equal to the range of $\CURL$ in \cite[(79)]{HomotopyPaper}, 
which is easily seen to be infinite-dimensional.

We compute $I_{\ix,3} I_{0,3}'P_{0,3}' P_{\ix,3}$.
Still using the identification \eqref{eq:kkbot}, 
the map $I_{0,3}'P_{0,3}': \Hibig^3\to \Hibig^3$ is given by 
$(k,k')\mapsto (k,0)$. 
Thus for $\cp\oplus\cm\in\ix_{\cpt}^3(\diamonddata)$,
\begin{align}
I_{\ix,3} I_{0,3}'P_{0,3}' P_{\ix,3}(\cp\oplus\cm)&=
\Re(\Chargei(\cp))_{4}
(\tilde{\nu}_4\oplus\tilde{\nu}_{4})+
\textstyle
\sum_{\ell=8}^{10}\Re(\Chargei(\cp))_{\ell}
(\tilde{\nu}_\ell\oplus\tilde{\nu}_\ell)\nonumber\\
&\quad\textstyle+
\sum_{\ell=5}^{10}\Im(\Chargei(\cp))_{\ell}
(i\tilde{\nu}_\ell\oplus(-i\tilde{\nu}_{\ell}))
\label{eq:chform}
\end{align}
using \eqref{eq:kernel},
the definition of $P_{\ix,3}$ before \eqref{eq:Csldef}
and of $I_{\ix,3}$ in \eqref{eq:Iidef}.
Set 
$\tilde{\mu}_{1}=\tilde{\nu}_4\oplus\tilde{\nu}_{4}$
and set 
$\tilde{\mu}_{2+i}=-(i\tilde{\nu}_{5+i}\oplus(-i\tilde{\nu}_{5+i}))$,
$\tilde{\mu}_{5+i}=2(\tilde{\nu}_{8+i}\oplus\tilde{\nu}_{8+i})$,
$\tilde{\mu}_{8+i}=-2(i\tilde{\nu}_{8+i}\oplus(-i\tilde{\nu}_{8+i}))$ 
for $i=0,1,2$.
With this definition, and using \eqref{eq:pi12}, \eqref{eq:chform}
we obtain \eqref{eq:Picexplicit} with 
$\mu_{\ell} = ((G_{\dR,3}\otimes\one) \uinj_{3}\tilde{\mu}_\ell )\oplus \tilde{\mu}_\ell$.
By \cite[(a1) of Theorem 3]{HomotopyPaper}, the 
$\nu_{\ell}$ and hence the
$\tilde{\nu}_\ell$ and $\tilde{\mu}_\ell$ vanish on $|\vec{x}|\ge \Rhomotopy$,
thus $\mu_\ell$ vanish on $|\vec{x}|\ge \Rhomotopy$ by 
\cite[(a2) of Theorem 2]{HomotopyPaper} and locality of $\uinj$.
One checks $ \Charge(\mu_{\ell})_{\ell'}=\delta_{\ell\ell'}$
using \eqref{eq:ortho} and the relation between 
the $\tilde{\nu}_\ell$ and $\mu_\ell$.
This concludes the proof of \ref{item:homotopy}.

\ref{item:support}:
By \eqref{eq:contrto20} and \eqref{eq:ccontr},
\begin{equation}\label{eq:Gxexp}
G_{\cx} 
= \tilde{G}_{\cx} 
- \tilde{G}_{\cx} q \tilde{G}_{\cx}
+ \tilde{I}_{\cx} G_0\tilde{P}_{\cx} 
- \tilde{G}_{\cx}q\tilde{I}_{\cx} G_0\tilde{P}_{\cx}
\end{equation}
The operator $\tilde{G}_{\cx}$ satisfies \ref{item:support}
by \cite[(a2) of Theorem 2, (a2) of Theorem 3]{HomotopyPaper}.
Since $q$ is local, also $\tilde{G}_{\cx} q \tilde{G}_{\cx}$
satisfies \ref{item:support}.
The last two terms satisfy \ref{item:support}
since the range of $G_0$ is only nonzero in $\ix^3_{\cpt}(\diamonddata)$,
and the $\tilde\nu_{\ell}$
in \eqref{eq:Iidef} vanish for $|\vec{x}|\ge \Rhomotopy$.

\ref{item:estimates}, \eqref{eq:Gcest}:
For $\ss\in\Z_{\ge1}$ and $\delta>2$
and all \smash{$\co\in\lxc^k_{\cpt}(\diamonddata)$, $\ci\in\ix^{k+1}_{\cpt}(\diamonddata)$}:
\begin{align*}
\|(G_{\dR}\otimes\one)\co\|_{\Hb^{\ss+1,\delta}(\diamonddata)}
&\lesssim_{\ss,\delta,\Rhomotopy} 
\|\co\|_{\Hb^{\ss,\delta}(\diamonddata)},&
\|\Gi \ci\|_{\Hb^{\ss,\delta}(\diamonddata)}
&\lesssim_{\ss,\delta,\Rhomotopy} 
\|\ci\|_{\Hb^{\ss-1,\delta}(\diamonddata)}
\end{align*}
by \cite[(a3) of Theorem 2]{HomotopyPaper},
respectively
by the $k=3$ instance of \cite[(a3) of Theorem 3]{HomotopyPaper},
where one must carefully change bases.
Beware that in \cite{HomotopyPaper}, the norms
are defined relative to the measure $dx^1dx^2dx^3$
while the norms in Definition \ref{def:cnorms} use $dx^1dx^2dx^3/\xvecjap^3$.
Thus \eqref{eq:Gcest} holds for $\tilde{G}_{\cx}$, 
and then by Lemma \ref{lem:diffbop} 
it holds for $\tilde{G}_{\cx} q \tilde{G}_{\cx}$.
For the last two terms in \eqref{eq:Gxexp} the estimates
follow from $|(P_{\dR}\otimes\one)\co| \lesssim_{\delta} \smash{\|\co\|_{\Hb^{0,\delta}(\diamonddata)}}$.
\eqref{eq:Gccomest}:
This is checked similarly to \eqref{eq:Gcest}, 
now using \cite[(a4) of Theorem 2, 3]{HomotopyPaper}.
We omit the details.
%
\qed
\end{proof}
\begin{remark}
Proposition \ref{prop:Gxhomotopy} is stated
for the subcomplex $\cx_{\cpt}(\diamonddata)\subset\cx(\diamonddata)$.
One can make an analogous statement for the subcomplex 
given by 
\[ 
\cx_{\rel}^k(\diamonddata)
=
\big(
\xvecjap^{-2}\Omega_{\rel}^k(\diamonddata) \otimesRR \boosts
\oplus 
\xvecjap^{-1}\Omega_{\rel}^k(\diamonddata) \otimesRR \transl
\big)
\oplus
(\iC_{\rel}^{k+1}(\diamonddata) \oplus \iC_{\rel}^{k+1}(\diamonddata))_{\R}
\]
where $\Omega_{\rel}^k(\diamonddata)$
are the relative differential forms 
on the radial compactification of $\diamonddata=\R^{3}$,
and where 
$\iC_{\rel}^k(\diamonddata)
=
\{u \in \iC^k(\diamonddata) \mid
\forall \underline{\zeta}\in\confKilS:\;
(\ast_{\iC} u)(\underline{\zeta}) \in \Omega_{\C,\rel}^k(\diamonddata)
\}$.
\end{remark}

\section{Constraints via homotopy transfer}
\label{sec:homotopytransfer}

We derive the constraint equations in the form \eqref{eq:MCinfintro}.
The first step, in Section \ref{sec:SHScontraction}, 
is to obtain a contraction $(\pSHS,\iSHS,\hSHS)$
from 
the complex $(\gx(\diamond),\dg)$ defined on $\diamond$
to the complex $(\cx(\diamonddata),\dc)$ defined on 
the initial hypersurface $\diamonddata$:
\begin{equation}
\label{eq:ContractionSHS}
\begin{tikzpicture}[baseline=(current  bounding  box.center)]
  \matrix (m) [matrix of math nodes, column sep = 28mm, minimum width = 4mm ]
  {
    (\gx(\diamond),\dg) & (\cx(\diamonddata),\dc) \\
  };
  \path[-stealth]
    (m-1-1) edge [transform canvas={yshift=1.5mm}] node [above] 
    {\footnotesize $\pSHS$ chain map} (m-1-2)
            edge [out=180+25,in=180-25,min distance=8mm] node
                 [left,xshift=-1mm] {
                 \footnotesize $\hSHS$ homotopy 
                 } (m-1-1)
    (m-1-2) edge [transform canvas={yshift=-1.5mm}] node [below] 
    {\footnotesize $\iSHS$ chain map} (m-1-1)
                 ;
\end{tikzpicture}
\end{equation}
The map $\pSHS$ was introduced in \eqref{eq:pullbackg}
and is given by pullback of forms and densities.
The map $\iSHS$ will be the solution operator of a linear symmetric hyperbolic system with initial data in $\cx(\diamonddata)$.
In Section \ref{sec:Bndef} we define the 
multilinear operator $\Bdata_n$ in \eqref{eq:MCinfintro}
as a sum of trees of the form indicated in Figure \ref{fig:AutobahnTree}.
The relation to the homotopy transfer theorem
\cite{berglund,huebschmannstasheff,kontsevich,ksoibelman,Vallette}
is discussed in Remark \ref{rem:HPT},
and in Remark \ref{rem:MCinLinf} 
we explain that \eqref{eq:MCinfintro} is the Maurer-Cartan
equation in an $L_\infty$ algebra \cite[Definition 4.1]{getzler}.
We derive estimates for the $\Bdata_n$ 
in Section \ref{sec:multilinearestimates},
and prove that \eqref{eq:MCinfintro} is equivalent to the constraint equations 
\eqref{eq:Constraints_Eq} in Section \ref{sec:MCequation}.

\begin{figure}
\centering
\begin{tikzpicture}[scale=0.54]
  \coordinate (a) at (4*\dx,0*\dy);
  \coordinate (b) at (3*\dx,0*\dy);
  \coordinate (ab) at (3.5*\dx,\dy);
  \coordinate (abc) at (2.5*\dx,2*\dy);
  \coordinate (c) at (2*\dx,\dy);
  \coordinate (d) at (1*\dx,2*\dy);
  \coordinate (abcd) at (1.5*\dx,3*\dy);
  \coordinate (e) at (0*\dx,3*\dy);
  \coordinate (abcde) at (0.5*\dx,4*\dy);
  \coordinate (out) at (.5*\dx,5*\dy);
\ir{a}{ab}{\footnotesize$\one-\hSHS\dg$}  
\il{b}{ab}{\footnotesize$\one-\hSHS\dg$}
\hr{ab}{abc}{\footnotesize$\hSHS$}
\il{c}{abc}{\footnotesize$\one-\hSHS\dg$}
\il{d}{abcd}{\footnotesize$\one-\hSHS\dg$}
\il{e}{abcde}{\footnotesize$\one-\hSHS\dg$}
\hr{abc}{abcd}{\footnotesize$\hSHS$}
\hr{abcd}{abcde}{\footnotesize$\hSHS$}
\hl{abcde}{out}{\footnotesize$\pSHS$}
\br{ab}{}{}
\br{abc}{}{}
\br{abcd}{}{}
\br{abcde}{}{}
\end{tikzpicture}
\captionsetup{width=115mm}
\caption{
A trivalent tree graph representing the map 
$\pSHS[(\one-\hSHS\dg)\cdot,\hSHS[(\one-\hSHS\dg)\cdot,
	\hSHS[(\one-\hSHS\dg)\cdot,
	\hSHS[(\one-\hSHS\dg)\cdot,(\one-\hSHS\dg)\cdot]]]]
$ from $\gx(\diamond)^{\otimes5}$ to $\cx(\diamonddata)$,
so each vertex stands for an application of 
the bracket.
The map $\Bdata_5$ is given by the total symmetrization (with signs) 
of the inputs of this tree,
and only depends on the restriction of the inputs to 
$\diamonddata$.
C.f.~\cite{Vallette}.
}
\label{fig:AutobahnTree}
\end{figure}

\subsection{A contraction 
via symmetric hyperbolic gauge fixing}
\label{sec:SHScontraction}

We construct the contraction \eqref{eq:ContractionSHS}
following \cite[Theorem 9]{RTgLa2}, 
see also \cite[Section 1.4]{RTgLa1} for a longer discussion.
The construction is based on linear symmetric hyperbolic gauge fixing
and uses the gauge subspaces 
$\gxG^k(\diamond)\subset \gx^{k}(\diamond)$
in \eqref{eq:gGsum}.
We note that they coincide with the gauge spaces in 
\cite[Section 5.4.1]{MinkowskiPaper} if instead of
the vector field and metric 
\cite[(265)]{MinkowskiPaper} one uses $\p_{x^0}$ and $\gmink$.

Recall the subspaces 
$\gxG^k(\diamond)\subset \gx^{k}(\diamond)$ in \eqref{eq:gGsum},
and the module multiplication \eqref{eq:modmult}.
By Lemma \ref{lem:basisg}, multiplication with $dx^0$ is an isomorphism
\begin{equation}\label{eq:gGdx0gGiso}
e_0:\; \gxG^{k}(\diamond) \;\xrightarrow{\;\sim\;}\; dx^0 \gxG^{k}(\diamond)
\end{equation}
and further 
\begin{equation}\label{eq:ggGgG}
\gx^{k+1}(\diamond) \;=\; \gxG^{k+1}(\diamond) \oplus dx^0 \gxG^{k}(\diamond)
\end{equation}
Let $\pi_{0}:\gx^{k+1}(\diamond)\to dx^0 \gxG^{k}(\diamond)$ be the projection onto the second summand.
\newcommand{\wRinv}{\mathsf{R}}
\begin{lemma}\label{lem:pih}
Define the $\R$-linear map 
\begin{equation}\label{eq:wdef}
\smash{w : \; 
	\gxG^k(\diamond) 
	\hookrightarrow 
	\gx^k(\diamond)
	\xrightarrow{\;\dg\;}
	\gx^{k+1}(\diamond)
	\xrightarrow{\;\pi_{0}\;}
	dx^0\gxG^{k}(\diamond)}
\end{equation}
The target space may be identified with $\gxG^k(\diamond)$,
using \eqref{eq:gGdx0gGiso}.
Then $w$ is a first order linear symmetric hyperbolic operator
for which the level sets of $x^0$ are spacelike.
In particular, $w$ is surjective, 
and there exists a unique right inverse 
\begin{equation*}
\smash{\wRinv :\;dx^0\gxG^{k}(\diamond)
	\to \gxG^k(\diamond) }
\end{equation*}
such that the elements in $\image(\wRinv)$ vanish along $\diamonddata$, 
i.e.~$\image(\wRinv)\subset x^0\gxG(\diamond)$.
Further, $\dg$ restricts to a map 
$\ker(w) \to \ker(w)$, 
hence $(\ker(w),\dg|_{\ker(w)})$ is a complex.
Further, the following composition is an isomorphism of complexes:
\begin{equation}\label{eq:kerwc}
\smash{\ker(w_k) 
\ \hookrightarrow\ 
\gx^k(\diamond)
\ \xrightarrow{\ \pSHS\ }\ 
\cx^k(\diamonddata)}
\end{equation}
where $w_k$ denotes the component of $w$ that maps
degree $k$ to degree $k+1$,
and $\pSHS$ is the map \eqref{eq:pullbackg}.
Define maps
$\iSHS: \cx^k(\diamonddata) \to \gx^k(\diamond)$
and
$\hSHS: \gx^{k+1}(\diamond) \to \gx^{k}(\diamond)$
by:
\begin{itemize}
\item 
$\hSHS:\gx^{k+1}(\diamond) 
	\xrightarrow{\;\pi_0\;} dx^0\gxG^{k}(\diamond)
	\xrightarrow{\;\wRinv\;}
	\gxG^k(\diamond)
	\hookrightarrow \gx^k(\diamond)$.
\item 
$\iSHS:\cx^k(\diamonddata) 
	\to
	\ker(w_k)
	\hookrightarrow \gx^k(\diamond)$
where the first arrow is the inverse of \eqref{eq:kerwc}.
\end{itemize}
Then the triple 
\begin{equation}\label{eq:pih}
(\pSHS,\iSHS,\hSHS)
\end{equation}
is a contraction \eqref{eq:contractionproperties}
with side conditions \eqref{eq:sidecond} from 
$(\gx(\diamond),\dg)$ to $(\cx(\diamonddata),\dc)$.
\end{lemma}
\begin{proof}
Symmetric hyperbolicity is shown in \cite[Lemma 42]{Thesis};
for concreteness we include a proof:
Denote the basis of $\gxG^k(\diamond)$ in Lemma \ref{lem:basisg}
by $b_1,\dots,b_{\ngG_k}$. 
For $u = f_i b_i$
(implicit sum over $i=1\dots\ngG_k$),
the second item in Lemma \ref{lem:dglabasisNEW} yields
\begin{align*}
w_k(u) 
&= \pi_0 \dg(f_i b_i)
=  (\xvecjap\p_{x^{\mu}}f_i) 
	\pi_0(\tfrac{dx^\mu}{\xvecjap}b_i) + (\pi_0 \dg b_i)f_i
\end{align*}
Only the first term is principal.
By direct calculation, with $\ell=1,2,3$,
\begin{equation}\label{eq:Amodmult}
\pi_0(\tfrac{dx^\mu}{\xvecjap}b_i) 
	= (A^{\mu}_k)_{ji} \tfrac{dx^0}{\xvecjap} b_j 
\quad
\text{where}
\quad
A^0_k=\one_{\ngG_k\times\ngG_k},\;
(A^\ell_k)^T=A^\ell_k\in\R^{\ngG_k\times\ngG_k}
\end{equation}
Concretely
$A^\ell_0=0$, 
$A^\ell_1=-a^\ell_1$, 
$A^\ell_2=-a^\ell_2$, 
$A^\ell_3=0$
with $a^\ell_1$, $a^\ell_2$ the matrices in \cite[(213), 
see also (199)]{MinkowskiPaper}.
Thus $w_k$ is linear symmetric hyperbolic.
The remaining statements follow from \cite[Lemma 76]{Thesis} 
applied to $M=\gx(\diamond)$, $M_{\textnormal{G}}=\gxG(\diamond)$
and \cite[Lemma 86]{Thesis}\footnote{%
In \cite{Thesis} the map $w$ is denoted $K$,
and instead of $\pi_0$ the canonical projection 
$\gx^{k+1}(\diamond)\to \gx^{k+1}(\diamond)/\gxG^{k+1}(\diamond)$ is used.};
the contraction $(p,i,h)$ in \cite[Lemma 76]{Thesis} 
is from $(\gx(\diamond),\dg)$ to $(\ker(w),\dg|_{\ker(w)})$,
to obtain \eqref{eq:pih} one must compose with the isomorphism
\eqref{eq:kerwc} and use \cite[(339)]{Thesis}
to conclude that $p$ composed with \eqref{eq:kerwc} is $\pSHS$.
\qed
\end{proof}
Recall the maps $\pdata$, $\idata$ 
from \eqref{eq:pdatadef}, \eqref{eq:idatadef}:
The map $\pdata$ is the restriction of 
$\pSHS$ to the initial hypersurface $\diamonddata$,
and $\idata$ is the inverse of 
$\gxdataG^k(\diamonddata)\to\cx^k(\diamonddata)$, 
$\udata\mapsto\pdata(\udata)$.
\smash{Thus for all 
$u\in\gx^k(\diamond)$ and 
all $\cc\in\cx^k(\diamonddata)$,}
\begin{align}\label{eq:idatai}
\pdata(u|_{\diamonddata}) = \pSHS(u) 
&&
\idata(\cc) = \iSHS(\cc)|_{\diamonddata}
\end{align}
\begin{remark}\label{rem:modmultfibinj}
We have stated that multiplication with 
$dx^0$ is a (fiberwise) injective map $\gxG^k(\diamond) \to \gx^k(\diamond)$.
More generally, for all $\omega\in\Omega^1(\diamond)$ 
that are timelike with respect to the Minkowski metric $\gmink$,
multiplication 
$\gxG^k(\diamond) \to \gx^{k+1}(\diamond)$, $u\mapsto \omega u$
is (fiberwise) injective.
Proof: 
It suffices to check that 
$u\mapsto \pi_{0}(\omega u)$ is injective.
This in turn follows from \eqref{eq:Amodmult}
and the fact that $\omega(\p_{x^\mu})A^\mu_k$ 
is positive definite 
when $\omega$ is future directed timelike,
c.f.~\cite[Remark 17]{MinkowskiPaper}.
\end{remark}

Define the $C^\infty(\diamond)$-linear map\footnote{
The analogous map in \cite[Section 5.4.1]{MinkowskiPaper}
was defined as the adjoint the multiplication map,
see \cite[(267)]{MinkowskiPaper}.
The definition here is equivalent.
}
\begin{equation}\label{eq:pitildedef}
\smash{\pitilde:\;
\gx^{k+1}(\diamond)
\xrightarrow{\pi_0} dx^0\gxG^k(\diamond)
\xrightarrow{e_0^{-1}} \gxG^k(\diamond)}
\end{equation}
where $\pi_0$ is defined after \eqref{eq:ggGgG} 
and 
$e_0^{-1}$ is the inverse of
$e_0$ in \eqref{eq:gGdx0gGiso}.
\begin{lemma}\label{lem:hfirstorder}
For all $u\in\gx(\diamond)$ and all $m\in\Z_{\ge0}$ one has
\begin{subequations}
\begin{align}
\hSHS((x^0)^m u) 
	&\;\in\; (x^0)^{m+1} \gx(\diamond)
	\label{eq:hx0}\\
(\hSHS- \tfrac{1}{m+1} x^0\pitilde)((x^0)^m u) 
	&\;\in\; (x^0)^{m+2} \gx(\diamond)
	\label{eq:hx0firstorder}
\end{align}
\end{subequations}
\end{lemma}
\begin{proof}
Denote $v=\pi_0u$, then $\hSHS((x^{0})^mu)=(x^{0})^m \wRinv v$.
\textit{Proof of \eqref{eq:hx0}.}
By definition of $\wRinv$ we have
$\wRinv((x^0)^m v)|_{\diamonddata}=0$, $w \wRinv((x^0)^m v)=(x^0)^m v$.
Thus $\wRinv((x^0)^m v)\in (x^0)^{m+1}\gx(\diamond)$, 
using the fact that $w$ is linear first order
symmetric hyperbolic with $\diamonddata$ spacelike, by Lemma \ref{lem:pih}.
\textit{Proof of \eqref{eq:hx0firstorder}.}
Set 
$$
A 
:= 
(\hSHS- \tfrac{1}{m+1} x^0\pitilde)((x^0)^m u) 
=
(\wRinv- \tfrac{1}{m+1} x^0 e_0^{-1})((x^0)^m v) 
$$
By symmetric hyperbolicity of $w$ it again suffices
to show that $A|_{\diamonddata}=0$ and 
$wA\in (x^0)^{m+1} \gx(\diamond)$.
The first is clear. We show the second:
Using \eqref{eq:dgLeib},
\begin{align*}
w\big(\tfrac{1}{m+1} (x^0)^{m+1} e_0^{-1}(v)\big)
&=
\pi_0 \dg\big(\tfrac{1}{m+1} (x^0)^{m+1} e_0^{-1}(v)\big)  \\
&=
(x^0)^{m}\pi_0( dx^0 e_0^{-1} v)
+
\tfrac{1}{m+1} (x^0)^{m+1} \pi_0 (\dg e_0^{-1}(v))
\end{align*}
The first term equals $(x^0)^m v=w\wRinv((x^0)^m v)$,
thus 
$wA\in(x^0)^{m+1}\gx(\diamond)$.
\qed
\end{proof}


\subsection{Definition of multilinear maps}
	\label{sec:Bndef}
	
We define the multilinear, 
first order differential operators $\Bdata_n$
in three steps:
in Definition \ref{def:tree}
we use the contraction $(\pSHS,\iSHS,\hSHS)$ in \eqref{eq:pih}
to define maps $\tree_n$ given by trees as
illustrated in Figure \ref{fig:AutobahnTree};
in Definition \ref{def:Bn} we 
define maps $\B_n$ given by the total graded symmetrization of the $\tree_n$;
in Proposition \ref{prop:Bproperties} we show that 
the $\B_n$ only depend on the restriction of the inputs to $\diamonddata$,
which defines $\Bdata_n$.


\begin{definition}[Trees]\label{def:tree}
Define the $\R$-bilinear map
\begin{align*}
\brh:\;\gx^k(\diamond)\times \gx^{k'}(\diamond)
	&\to \gx^{k+k'-1}(\diamond)
	&\brh(u,u') &=\brh_u(u')= [u,\hSHS u']
\end{align*}
For each $n\in\Z_{\ge2}$ define the $\R$-multilinear map 
$$\tree_{n}:\;\gx^{k_1}(\diamond)\times\dots\times\gx^{k_n}(\diamond)
	\to\cx^{k_1+\dots+k_n-(n-2)}(\diamonddata)$$ 
by
\begin{align}\label{eq:Tndef}
\begin{aligned}
\tree_2(u_1,u_2) 
	&= (-1)^{k_1} \pSHS[\tilde{u}_1,\tilde{u}_2]\\
\tree_{n\ge3}(u_1,\dots,u_n)
	&=
	(-1)^{k_1+\dots+k_{n-1}}
	\pSHS  \brh_{\tilde{u}_1} 
	\cdots \brh_{\tilde{u}_{n-2}}
	[\tilde{u}_{n-1},\tilde{u}_{n}]
\end{aligned}
\end{align}
where $\tilde{u}_i = (\one-\hSHS\dg)u_i$.
\end{definition}
The map $\tree_n$ may be represented as a
trivalent tree graph as in Figure \ref{fig:AutobahnTree}.
The next lemma will show that 
$\tree_n$ is a first order differential operator.

Recall the anchor map $\anchorg$ in \cite[Theorem 4]{MinkowskiPaper}:
For $(\omega\otimes\zeta)\oplus\uI\in\gx^k(\diamond)$ 
with $\omega\otimes\zeta\in\Omega^k(\diamond)\otimesRR\Kil$,
$\uI\in\I^{k+1}(\diamond)$ 
and for $\omega'\in\Omega(\diamond)$ one has
\[ 
\anchorg((\omega\otimes\zeta)\oplus\uI)(\omega') = \omega \wedge (\Lie_{\zeta} \omega')
\]
In particular, $\anchorg$ is $C^\infty$-linear in the first argument.
\begin{lemma}\label{lem:TreeSimp}
Define the $C^\infty$-bilinear map
$\rhoh:\gx^{k'}(\diamond) \times \gx^{k'}(\diamond)
\to \gx^{k+k'-1}(\diamond)$, 
\begin{align*}
	\rhoh(u,u')=\rhoh_u(u') = \anchorg(u)(x^0) \pitilde (u')
\end{align*}
where we use the multiplication \eqref{eq:modmult}
and the map $\pitilde$ in \eqref{eq:pitildedef}.
We further denote by 
$\rhohdata:\gxdata^{k}(\diamonddata)\times\gxdata^{k'}(\diamonddata)
\to \gxdata^{k+k'}(\diamonddata)$ the 
restriction of $\rhoh$ to $\diamonddata$.
Then 
\begin{equation}\label{eq:uup}
(\brh_u-\rhoh_u)(u') \in x^0\gx(\diamond)
\end{equation}
and 
for all $n\in\smash{\Z_{\ge2}}$
and $u_1\in\gx^{k_1}(\diamond),\dots,u_n\in\gx^{k_n}(\diamond)$:
\begin{align}\label{eq:TreeFormula}
\begin{aligned}
\tree_{n}(u_1,\dots,u_n)
	&=
	(-1)^{k_1+\dots+k_{n-1}}
	\pdata\,\rhohdata_{\udata_1} \cdots \rhohdata_{\udata_{n-2}}
	\Big([u_{n-1},u_n]|_{\diamonddata}\\
	&\qquad\qquad
	-
	\rhohdata_{\udata_{n-1}}(\dg u_n)|_{\diamonddata}
	+
	(-1)^{k_{n-1}k_n}\rhohdata_{\udata_n}(\dg u_{n-1})|_{\diamonddata}\Big)
\end{aligned}
\end{align}
where $\udata_i=u_i|_{\diamonddata}$.
\end{lemma}
\begin{proof}
\eqref{eq:uup}: By 
\eqref{eq:hx0firstorder} with $m=0$ and \eqref{eq:anchorleib}.
\eqref{eq:TreeFormula}:
Using \eqref{eq:uup}, $C^\infty$-bilinearity of $\rhoh$,
$\pSHS(x^0\gx(\diamond))=0$ and \eqref{eq:hx0} with $m=0$,
we obtain
\begin{align*}
\tree_{n}(u_1,\dots,u_n)
&=
	(-1)^{k_1+\dots+k_{n-1}}
	\pSHS \rhoh_{u_1} 
	\cdots \rhoh_{u_{n-2}} [\tilde{u}_{n-1},\tilde{u}_{n}]
\intertext{where $\tilde{u}_j = (1-\hSHS\dg)u_j$.
Also using the first of \eqref{eq:idatai}, we obtain}
\tree_{n}(u_1,\dots,u_n)
&=
	(-1)^{k_1+\dots+k_{n-1}}
	\pdata \rhohdata_{\udata_1} 
	\cdots \rhohdata_{\udata_{n-2}} ([\tilde{u}_{n-1},\tilde{u}_{n}]|_{\diamonddata})
\end{align*}
Using $\R$-bilinearity of the bracket and the definition of $\brh$,
we obtain
\begin{align*}
[\tilde{u}_{n-1},\tilde{u}_{n}]|_{\diamonddata}
=
[u_{n-1},u_n]|_{\diamonddata}
-
\rhohdata_{\udata_{n-1}}(\dg u_n)|_{\diamonddata}
+
(-1)^{k_{n-1}k_n}\rhohdata_{\udata_{n}}(\dg u_{n-1})|_{\diamonddata}
\end{align*}
where we also used \eqref{eq:uup} and $[\hSHS-,\hSHS-]|_{\diamonddata}=0$
by Lemma \ref{lem:hfirstorder} and \eqref{eq:anchorleib}.
\qed
\end{proof}
Let $S_n$ be the symmetric groups on $n$ letters.
\begin{definition}[Koszul signs]\label{def:Koszsigns}
For a permutation $\pi\in S_n$ 
and integers $k_1,\dots,k_n\in\Z$ define
$\chi(\pi;k_1,\dots,k_{n})\in\{\pm1\}$ as follows.
If $\pi=(j,j+1)$ for some $j=1\dots n-1$, 
i.e.~$\pi$ is a simple transposition, then 
\begin{align}\label{eq:chisimp}
\chi((j,j+1);k_1,\dots,k_{n}) &= -(-1)^{k_j k_{j+1}}
\end{align}
If $\pi_1\circ\pi_2$ is a composition of two permutations
$\pi_1,\pi_2\in S_n$ then 
\[ 
\chi(\pi_1\circ\pi_2;k_1,\dots,k_{n})
=
\chi(\pi_1;k_1,\dots,k_{n})
\chi(\pi_2;k_{\pi_1(1)},\dots,k_{\pi_1(n)})
\]
Define $\hat\chi(\pi;k_1,\dots,k_{n})\in\{\pm1\}$
analogously, where instead of \eqref{eq:chisimp} we use
\[ 
\hat\chi((j,j+1);k_1,\dots,k_{n}) = (-1)^{(k_j+1) (k_{j+1}+1)}
\]
\end{definition}
Observe that 
$\chi(\pi;1,\dots,1)=1$ and $\hat\chi(\pi;1,\dots,1)=1$
for all $\pi\in S_n$.
\begin{definition}\label{def:Bn}
For each $n\in\Z_{\ge1}$ define the $\R$-multilinear map
\begin{align*}
\B_n:\;\; \gx^{k_1}(\diamond)\times\dots\times\gx^{k_n}(\diamond)
	&\to\cx^{k_1+\dots+k_n-(n-2)}(\diamonddata)
\end{align*}
by 
\begin{align*}
\B_{1}(u_1)
	&=
	\pSHS \dg (\one-\hSHS\dg) u_1 
	=\pSHS \dg  u_1
	=\dc\pSHS u_1\\
\B_{n\ge2}(u_1,\dots,u_n)
	&=
	\tfrac12(-1)^{\lfloor\frac{n}{2}\rfloor+1}
	(-1)^{k_{n-1}+k_{n-3}+\dots}\\
	&\qquad\quad\tsum_{\pi\in S_n}
	\hat{\chi}(\pi;k_1,\dots,k_n)
	\tree_{n}\big(u_{\pi(1)},\dots,u_{\pi(n)}\big)
\end{align*}
\end{definition}

In the definition of $\B_1$, the second and third formulas follow from 
\eqref{eq:contractionproperties}, \eqref{eq:sidecond}.

When all $u_1,\dots,u_n$ are in degree one, 
i.e.~$k_1,\dots,k_n=1$, then the signs
that appear in the definition of $\B_{n\ge2}$
and of $\tree_{n\ge2}$ simplify to $(-1)^n$.
\begin{remark}[Relation to the homotopy transfer theorem]
\label{rem:HPT}
Composing the maps $\B_n$ with $\iSHS:\cx^k(\diamonddata)\to\gx^k(\diamond)$ in each input, 
we obtain maps
\begin{equation}\label{eq:Bni}
\B_n(\iSHS-,\cdots,\iSHS-):\; 
\cx^{k_1}(\diamonddata)\times\cdots\times\cx^{k_n}(\diamonddata)
\to \cx^{k_1+\cdots+k_n-(n-2)}(\diamonddata)
\end{equation}
Note that $(\one-\hSHS\dg)\iSHS = \iSHS$ by \eqref{eq:contractionproperties}, \eqref{eq:sidecond}.
Therefore \eqref{eq:Bni} coincide with the $L_\infty$
brackets obtained by transferring the Lie bracket $[\cdot,\cdot]$
along the contraction $(\pSHS,\iSHS,\hSHS)$ using the homotopy transfer theorem, 
see \cite[Theorem 10.3.5]{Vallette}.
The general formula for the $L_\infty$ brackets
obtained from homotopy transfer is given by a sum 
of trivalent tree graphs such as, for example,
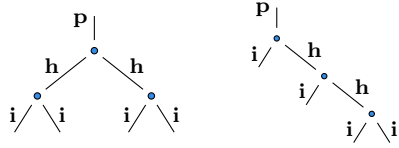
\[ 
\begin{aligned}
\begin{tikzpicture}[scale=0.6]
  \coordinate (a) at (0*\dx,0);
  \coordinate (b) at (\dx,0);
  \coordinate (c) at (2*\dx,0);
  \coordinate (d) at (3*\dx,0);
  \coordinate (ab) at (0.5*\dx,\dy);
  \coordinate (cd) at (2.5*\dx,\dy);
  \coordinate (abcd) at (1.5*\dx,2*\dy);
  \coordinate (out) at (1.5*\dx,3*\dy);
  \il{a}{ab}{\footnotesize$\iSHS$}
  \ir{b}{ab}{\footnotesize$\iSHS$}
  \il{c}{cd}{\footnotesize$\iSHS$}
  \ir{d}{cd}{\footnotesize$\iSHS$}
  \hl{ab}{abcd}{\footnotesize$\hSHS$}
  \hr{cd}{abcd}{\footnotesize$\hSHS$}
  \pl{abcd}{out}{\footnotesize$\pSHS$}
  \br{ab}{}{}
  \br{cd}{}{}
  \br{abcd}{}{}
\end{tikzpicture}
\end{aligned}
\qquad
\begin{aligned}
\begin{tikzpicture}[scale=0.5]
  \coordinate (a) at (4*\dx,0*\dy);
  \coordinate (b) at (3*\dx,0*\dy);
  \coordinate (ab) at (3.5*\dx,\dy);
  \coordinate (abc) at (2.5*\dx,2*\dy);
  \coordinate (c) at (2*\dx,\dy);
  \coordinate (d) at (1*\dx,2*\dy);
  \coordinate (abcd) at (1.5*\dx,3*\dy);
  \coordinate (out) at (1.5*\dx,4.1*\dy);
\ir{a}{ab}{\footnotesize$\iSHS$}  
\il{b}{ab}{\footnotesize$\iSHS$}
\hr{ab}{abc}{\footnotesize$\hSHS$}
\il{c}{abc}{\footnotesize$\iSHS$}
\il{d}{abcd}{\footnotesize$\iSHS$}
\hr{abc}{abcd}{\footnotesize$\hSHS$}
\pl{abcd}{out}{\footnotesize$\pSHS$}
\br{ab}{}{}
\br{abc}{}{}
\br{abcd}{}{}
\end{tikzpicture}
\end{aligned}
\]
The tree on the left vanishes identically 
because each $\hSHS$ maps to $x^0\gx(\diamond)$ by Lemma \ref{lem:hfirstorder}, 
and thus $[\hSHS-,\hSHS-]$ also maps to $x^0\gx(\diamond)$
by \eqref{eq:anchorleib}, hence $\pSHS[\hSHS-,\hSHS-]=0$.
By contrast, the tree on the right is in general nonzero,
and appears in \eqref{eq:Bni}, via Definition \ref{def:Bn}.
By this type of argument, the general formula 
for the $L_\infty$ brackets obtained from homotopy transfer
reduces to \eqref{eq:Bni}.

In particular, it follows from general theory that 
\eqref{eq:Bni} satisfy the $L_{\infty}$ Jacobi identities.
They also follow from \ref{item:Jacobi} below,
in fact replacing each input $u_j$ in \ref{item:Jacobi}
with an element in the image of $\iSHS$
yields the standard $L_{\infty}$ Jacobi identities.

In the literature on homotopy transfer and $L_\infty$
algebras, different sign conventions for the multilinear maps are used,
which affects the signs in the $L_{\infty}$ Jacobi identities.
We use the convention in \cite[Definition 4.1]{getzler}.
The $n$-multilinear maps in \cite[Section 13.2.9]{Vallette}
differ from those in \cite{getzler} by 
a multiplicative sign $(-1)^{{n+1\choose2}}$,
and thus satisfy $L_{\infty}$ Jacobi identities
where $(-1)^j$ in \ref{item:Jacobi} is replaced by $(-1)^{j(n-j)}$.
This change of convention is explained in \cite[Section 4]{getzler}.
\end{remark}
\begin{prop}\label{prop:Bproperties}
For all $n\in\Z_{\ge1}$ and 
$u_1\in\gx^{k_1}(\diamond),\dots,u_n\in\gx^{k_n}(\diamond)$:
\begin{enumerate}[label=\textnormal{({c\arabic*})}]
\item \label{item:Bngrsym}
$\B_n$ is graded symmetric:
For all $\kappa\in S_n$ and using 
$\chi$ in Definition \ref{def:Koszsigns},
\begin{align*}
\B_n(u_1,\dots,u_n)
	=
	\chi(\kappa;k_1,\dots,k_n)
	\B_n(u_{\kappa(1)},\dots,u_{\kappa(n)})
\end{align*}
\item \label{item:Jacobi}
The following higher Jacobi identities hold:
\begin{align*}
0
&=\textstyle
\sum_{j=1}^n (-1)^j \sum_{\pi\in S_{n,j}}\\
&\qquad \chi(\pi;k_1,\dots,k_n)\B_{n-j+1}\big(
\iSHS\B_{j}\big(u_{\pi(1)},\dots,u_{\pi(j)}\big),
u_{\pi(j+1)},\dots,u_{\pi(n)}\big)
\end{align*}
with $S_{n,j}$ given by all $\pi\in S_n$
with $\pi(1)<\dots<\pi(j)$, 
$\pi(j+1)<\dots<\pi(n)$.
\item \label{item:Bzero}
For $j=1\dots n$ one has 
$
\B_n(u_1,\dots, x^0 u_j,\dots,u_n) = 0
$.
Thus we can define
\[ 
\Bdata_{n}:\gxdata^{k_1}(\diamonddata)\otimes\cdots\otimes\gxdata^{k_n}(\diamonddata) 
	\to 
	\cx^{k_1+\dots+k_n-(n-2)}(\diamonddata)
\]
by 
$$
\Bdata_{n}(\udata_1,\dots,\udata_n) = \B_{n}(u_1,\dots,u_n)
$$
where each $u_j\in\gx^{k_j}(\diamond)$
is any element that satisfies
$u_j|_{\diamonddata} = \udata_j$.
\end{enumerate}
\end{prop}
\begin{proof}
\textit{Proof of \ref{item:Bngrsym}.}
It suffices to check this for $\kappa = (j,j+1)$ with $j=1,\dots,n-1$.
Let $\pi\in S_n$, and let 
$P$ and $Q$ be the coefficients of 
$$\tree_n(u_{\pi(1)},\dots,u_{\pi(n)})$$
in $\B_n(u_1,\dots,u_n)$
respectively in $-(-1)^{k_jk_{j+1}}\B_n(u_1,\dots,u_{j+1},u_j,\dots,u_n)$,
using Definition \ref{def:Bn}.
We show $Q=P$:
Note that $Q/P$ is equal to 
\begin{align*}
-(-1)^{k_j k_{j+1}} (-1)^{k_j+k_{j+1}}
\hat{\chi}(\pi;k_1,\dots,k_n) 
\hat{\chi}((j,j+1)\circ \pi;k_1,\dots,k_{j+1},k_j,\dots,k_n)
\end{align*}
Together with Definition \ref{def:Koszsigns}
this implies $Q/P=1$.

\textit{Proof of \ref{item:Jacobi}.}
Define
$\Baux_n: \gx^{k_1}(\diamond)\times\dots\times\gx^{k_n}(\diamond)
	\to\cx^{k_1+\dots+k_n-(n-2)}(\diamonddata)$
just like $\B_n$ in Definition \ref{def:Bn},
but replacing, 
in the definition of $\tree_n$ in \eqref{eq:Tndef}, 
each $\tilde{u}_i=(\one-\hSHS\dg)u_i$ by $u_i$.
Then 
\begin{equation}\label{eq:B'B''}
\B_{n}(u_1,\dots,u_n) = \Baux_n((\one-\hSHS\dg)u_1,\dots,(\one-\hSHS\dg)u_n)
\end{equation}
It now suffices to show that for all 
$v_1\in\gx^{k_1}(\diamond),\dots,v_n\in\gx^{k_n}(\diamond)$:
\begin{align}
(-1)^n \dc \Baux_{n}(v_1,\dots,v_n)
&=
\textstyle
\sum_{\pi\in S_{n,1}} \chi(\pi;k_1,\dots,k_n)
\Baux_{n}\big(
\dg v_{\pi(1)}, v_{\pi(2)},\dots,v_{\pi(n)}\big)\nonumber\\
&\textstyle \;-
\sum_{j=2}^{n-1} (-1)^j \sum_{\pi\in S_{n,j}}
\chi(\pi;k_1,\dots,k_n)\nonumber\\
&\;\; \Baux_{n-j+1}\big(
\iSHS\Baux_{j}\big(v_{\pi(1)},\dots,v_{\pi(j)}\big),
v_{\pi(j+1)},\dots,v_{\pi(n)}\big)\label{eq:Bauxjac}
\end{align}
This indeed implies \ref{item:Jacobi}
by setting $v_j = (\one-\hSHS\dg)u_j$,
using \eqref{eq:B'B''},
and using $\iSHS = (\one-\hSHS\dg)\iSHS$
and $\dg(\one-\hSHS\dg)=(\one-\hSHS\dg) \iSHS \B_1$
and $\dc=\B_1\iSHS$.

To prove \eqref{eq:Bauxjac}, first define
$\B_n'': \gx^{k_1}(\diamond)\times\dots\times\gx^{k_n}(\diamond)
	\to\cx^{k_1+\dots+k_n-(n-2)}(\diamonddata)$
to be given by the sum (with appropriate signs) 
over all trivalent tree graphs with with $n+1$ labeled 
leaves, where the first $n$ are inputs and the last is the output.
The output leaf stands for the map $\pSHS$, each vertex
stands for the bracket, each internal line for $\hSHS$,
and each input leaf for the identity map.
One can show that $\B_n''=\Baux_n$, since
all trees defining $\Baux_n$ appear in $\B_n''$,
and the remaining trees in the definition of $\B_n''$
vanish individually, similar to Remark \ref{rem:HPT}.
It thus suffices to prove \eqref{eq:Bauxjac}
with all $\Baux$ replaced by $\B''$.
The proof of \eqref{eq:Bauxjac} with $\Baux$ replaced by $\B''$ 
is then similar to the proof of the homotopy transfer theorem
in \cite[Section 9.4.1,9.4.3 and 10.3.1,10.3.2]{Vallette}.
Briefly, start with 
the term
$\dc \B_{n}''(u_1,\dots,u_n)$.
Use the definition of $\B_{n}''$ as a sum of 
trees, use $\dc \pSHS = \pSHS \dg$ once for each tree, 
and then repeatedly use the Leibniz rule to 
replace $\dg[-,-]$ by $ [\dg -,-]\pm[-,\dg -]$,
and \eqref{eq:contractionproperties} to
replace $\dg \hSHS$ by $\one-\iSHS\pSHS-\hSHS\dg$.
This process generates a sum of three types of trees:
\begin{itemize}
\item 
Tree graphs with $\dg$ at one input leave.
Their sum gives the first term on the right hand side
of \eqref{eq:Bauxjac}.
\item 
Tree graphs where one internal line is labeled by $\iSHS\pSHS$.
Their sum gives the second term on the right hand side
of \eqref{eq:Bauxjac}.
\item 
Tree graphs where one internal line is labeled by $\one$.
Their sum vanishes by the graded Jacobi identity 
for the bracket \cite[(56j)]{MinkowskiPaper}.
\end{itemize}

\textit{Proof of \ref{item:Bzero}.}
By \ref{item:Bngrsym} it suffices to check
$\B_n(x^0 u_1,u_2,\dots,u_n) = 0$.
Using \eqref{eq:dgLeib}:
$\B_1(x^0 u_1) = \pSHS \dg(x^0 u_1) =  \pSHS( dx^0 u_1 +x^0 \dg u_1)=0$.
For $n\ge2$
abbreviate 
$$
\tilde{u}_1=\hSHS (x^0 u_1)
\qquad
\tilde{u}_j=(\one-\hSHS\dg)u_j\;\; \text{for $j=2,\dots,n$}
$$
Using \eqref{eq:B'B''} and 
$(\one-\hSHS\dg)x^0 u_1 
	= (\iSHS\pSHS+\dg\hSHS)x^0 u_1 
	= \dg \tilde{u}_1$, we have 
\begin{align}
\B_n(x^0 u_1,u_2,\dots,u_n)
&=
\Baux_n(\dg \tilde{u}_1,\tilde{u}_2,\dots,\tilde{u}_n)
\label{eq:Bntilde}
\end{align}
By \eqref{eq:hx0} we have $\tilde{u}_1\in (x^0)^2\gx(\diamond)$.
Thus \eqref{eq:Bntilde} vanishes using \eqref{eq:Bauxjac} with 
$v_1,\dots,v_n$ given by $\tilde{u}_1,\dots,\tilde{u}_n$,
and the fact that each $\Baux$ vanishes
when at least one input is in $(x^0)^2 \gx(\diamond)$,
by a calculation analogous to \eqref{eq:TreeFormula}.
\qed
\end{proof}
\begin{remark}[%
Relation of 
$\Bdata_1$, $\Bdata_2$ 
to \eqref{eq:linconstraints}, \eqref{eq:CPuu}]
	\label{rem:B2charge}
Recall the linearized constraint operator $\linPconstraints$ 
in \eqref{eq:linconstraints}. Relative to the 
decomposition \eqref{eq:ggGgG} it is a
strictly lower triangular 2-by-2 block matrix differential operator.
The nonzero block maps 
$\gxdataG^1(\diamonddata)\to dx^0\gxdataG^2(\diamonddata)$.
The map $\idata$ in \eqref{eq:idatadef} identifies 
$\cx^1(\diamonddata)$ with the input space, 
and $\pdata\pitilde$ from \eqref{eq:pdatadef} and \eqref{eq:pitildedef} 
identifies the output space with $\cx^2(\diamonddata)$,
and then
\begin{equation}\label{eq:DPdc}
(\pdata \pitilde) \, \linPconstraints \  \idata \;=\; \dc
\end{equation}
(viewing the 
$C^\infty$-linear map $\pitilde$ as a map on $\diamonddata$).
Then, using $\linPconstraints\,\idata\,\pdata=\linPconstraints$, 
\begin{equation}\label{eq:DPB1}
(\pdata \pitilde) \linPconstraints  
\;=\; \dc \pdata 
\;=\; \Bdata_1
\end{equation}
Now suppose that $\udata \in \gxdata^1(\diamonddata)$
solves $\linPconstraints(\udata)=0$.
One obtains a solution $u\in\gx^1(\diamond)$
of the linearized Einstein equations $\dg u=0$ 
with initial data $u|_{\diamond}=\udata$
by choosing any $v\in\gx^1(\diamond)$ with $v|_{\diamonddata}=\udata$
and then setting
\begin{equation}\label{eq:d(1-hd)}
u = (\one-\hSHS\dg) v
\end{equation}
by \eqref{eq:contractionproperties}, \eqref{eq:DPB1}.
Now let $\udata_1,\udata_2\in\gxdata^1(\diamonddata)$ with 
$\linPconstraints(\udata_i)=0$ for $i=1,2$
and choose any two $u_1,u_2\in\gx^1(\diamond)$
with $\dg u_i=0$, $u_i|_{\diamonddata}=\udata_i$.
Then $(\one-\hSHS\dg)u_i = u_i$. Thus
$$ 
\pSHS [u_1,u_2]
=
\pSHS [(\one-\hSHS\dg)u_1,(\one-\hSHS\dg)u_2] 
=
\Bdata_2(\udata_1,\udata_2) 
$$
In particular, the expression \eqref{eq:CPuu}
is equivalently given by $\Charge(-\frac12\Bdata_2(\udata,\udata))$.
\end{remark}
\begin{remark}[An L$_{\infty}$ algebra on the space of initial data]
\label{rem:MCinLinf}
Recall from \eqref{eq:idatadef} that $\idata$ identifies
$\cx^k(\diamonddata)\simeq\gxdataG^k(\diamonddata)$.
Thus $\idata\Bdata_n$ is a map 
with input and output in $\gxdata(\diamonddata)$.
The identities \ref{item:Jacobi}
(after applying $\idata$ from the left)
are then the standard
$L_\infty$ Jacobi identities for the maps $\idata\Bdata_1,\idata\Bdata_2,\dots$, 
which thus define an 
$L_\infty$ algebra \cite[Definition 4.1]{getzler} on $\gxdata(\diamonddata)$.
The equation \eqref{eq:MCinfintro} is the
Maurer-Cartan equation in this L$_{\infty}$ algebra,
the equivalence with \eqref{eq:Constraints_Eq} is in 
Lemma \ref{lem:MCinfproperties} below.
\end{remark}

\subsection{Multilinear estimates}
	\label{sec:multilinearestimates}
We derive estimates for the first order differential
operators $\Bdata_n$ (Proposition \ref{prop:uBestimates}),
and for the commutators of $\Bdata_n$ with the vector fields
$\xvecjap\p_{x^i}$ (Lemma \ref{lem:uCBestimates}).

The main point of the next proposition is that 
the constant $\CB$, which enters the radius of 
convergence in \eqref{eq:MCinfintro}, 
is independent of $\ss,\delta$.

We abbreviate
\smash{$\|{\cdot}\|_{\Hb^{\ss,\delta}(\diamonddata)}=\|{\cdot}\|_{\Hb^{\ss,\delta}}$}
and analogously for \smash{$\Cb^{\ss,\delta}$}.
\begin{prop}\label{prop:uBestimates}
There exists $\CB\ge1$ such that for all 
$n\in\Z_{\ge1}$, 
$\ss\in\Z_{\ge6}$, 
$\delta,\delta_1,\dots,\delta_n\ge0$ with
$\delta_1+\dots+\delta_n=\delta$,
$Y\subset\{1,\dots,n\}$ 
and all $\udata_1,\dots,\udata_n\in\gxdata(\diamonddata)$:
\begin{align}\label{eq:uBestimates}
&\|\tfrac{1}{n!}\Bdata_{n}(\udata_1,\dots,\udata_n)\|_{\Hb^{\ss,\delta}}\\
&\quad
\lesssim_{\ss,\delta,|Y|}
	\CB^n\sum_{\substack{X\subset\{1,\dots,n\}\\ 
	|X| = \min\{n,\ss-3\}
	}}
	\Big(\prod_{i\in X\cup Y} 
	\|\udata_{i}\|_{\Hb^{\ss+1,\smash{\delta_{i}}}}\Big)
	\Big(\prod_{j\in X^c\setminus Y} 
	\|\udata_{j}\|_{\Cb^{4,\smash{\delta_{j}}}}
	\Big)\nonumber
\end{align}
using the norms in Definition \ref{def:cnorms}
on the left hand side, and the norms in Definition \ref{def:normsdata}
on the right hand side.
Further, for $n=2$ and all $\ss\in\Z_{\ge1}$,
\begin{align}\label{eq:B2est}
\|\Bdata_2(\udata_1,\udata_2)\|_{\Hb^{\ss,\delta}}
\lesssim_{\ss,\delta}
\|\udata_{1}\|_{\Cb^{\ss+1,\delta_1}}
\|\udata_{2}\|_{\Hb^{\ss+1,\delta_2}}
\end{align}
\end{prop}

The following corollary will be useful.
\begin{corollary}\label{cor:bn3}
For $\ss\in\Z_{\ge6}$, $\delta\ge0$, $b\ge1$,
$n\in\Z_{\ge1}$, $\udata_1,\dots,\udata_n\in\gxdata(\diamonddata)$:
\begin{itemize}
\item 
If $n\ge3$ and
$\|\udata_i\|_{\Hb^{\ss+1,0}} \le b$,
$\|\udata_i\|_{\Cb^{4,0}} \le \tfrac{1}{4\CB}$ 
for $i=4,\dots,n$
then 
\begin{align*}
\smash{\|\tfrac{1}{n!}\Bdata_n(\udata_1,\dots,\udata_n)
\|_{\Hb^{\ss,\delta+1}}
\lesssim_{\ss,\delta,b}
(\tfrac14)^n 
\|\udata_1\|_{\Hb^{\ss+1,\delta}}
\|\udata_2\|_{\Hb^{\ss+1,\frac12}}
\|\udata_3\|_{\Hb^{\ss+1,\frac12}}}
\end{align*}
\item 
If
$\|\udata_i\|_{\Hb^{\ss+1,0}} \le b$,
$\|\udata_i\|_{\Cb^{4,0}} \le \tfrac{1}{4\CB}$ 
for $i=2,\dots,n$ 
then
\begin{align*}
\|\tfrac{1}{n!}\Bdata_n(\udata_1,\dots,\udata_n)\|_{\Hb^{\ss,0}}
\lesssim_{\ss,b}
(\tfrac14)^n \|\udata_1\|_{\Hb^{\ss+1,0}}
\end{align*}
\end{itemize}
\end{corollary}
\begin{proof}
\textit{First item.}
We apply Proposition \ref{prop:uBestimates} with 
$Y=\{1,2,3\}$, 
$\delta$ there given by $\delta+1$ here,
and 
$\delta_1=\delta,
\delta_2=\frac12,\delta_3=\frac12$,
$\delta_{i\ge4}=0$.
In each summand on the right of \eqref{eq:uBestimates},
the first product is bounded by
$\|\udata_1\|_{\Hb^{\ss+1,\delta}}
\|\udata_2\|_{\Hb^{\ss+1,\smash{\frac12}}}
\|\udata_3\|_{\Hb^{\ss+1,\smash{\frac12}}} b^{|X\setminus Y|}$,
where \smash{$b^{|X\setminus Y|} \le b^{\ss} \lesssim_{\ss,b}1$};
the second product is bounded by
$\smash{(4\CB)^{-|X^c\setminus Y|}}=\smash{(4\CB)^{-n+|X\cup Y|}}
\lesssim_{\ss} \smash{(4\CB)^{-n}}$.
There are less than $2^{\ss}$ summands, thus the claim follows.
\textit{Second item.} 
Use Proposition \ref{prop:uBestimates} with 
$Y=\{1\}$ and all $\delta,\delta_i=0$.
\qed
\end{proof}
In the following we prove Proposition \ref{prop:uBestimates}.
Recall the map $\rhohdata$ in Lemma \ref{lem:TreeSimp}.
%
\newcommand{\CY}{\tc{cyan}{C}_0}
\newcommand{\CYYY}{\tc{parisgreen}{C}_1}
\begin{lemma}\label{lem:Yestimates}
There exists $\CY\ge1$ such that 
for all $\ss\in\Z_{\ge6}$ 
and $\delta\ge0$
there exists $C_{\ss,\delta}\ge1$ 
such that for all $\delta_1,\delta_2\ge0$
with $\delta_1+\delta_2=\delta$
and all $\udata,\udata'\in\gxdata(\diamonddata)$:%
\begin{subequations}\label{eq:Yestimates}
\begin{align}
\|\rhohdata(\udata, \udata')\|_{\Hb^{\ss,\delta}}
&\le
	\CY\big(\|\udata\|_{\Hb^{\ss,\delta_1}}
	\|\udata'\|_{\Cb^{3,\delta_2}}
	+
	\|\udata\|_{\Cb^{3,\delta_1}}
	\|\udata'\|_{\Hb^{\ss,\delta_2}}\big)
	\nonumber\\
&\quad+
	C_{\ss,\delta}\|\udata\|_{\Hb^{\ss-1,\delta_1}}
	\|\udata'\|_{\Hb^{\ss-1,\delta_2}}
	\label{eq:Biluv}\\
\|\rhohdata(\udata, \udata')\|_{\Hb^{5,\delta}}
	&\le
	\CY (
	\|\udata\|_{\Hb^{5,\delta_1}}
	\|\udata'\|_{\Cb^{3,\delta_2}}
	+
	\|\udata\|_{\Cb^{3,\delta_1}}
	\|\udata'\|_{\Hb^{5,\delta_2}}
	)
\label{eq:H5est}
\\
\|\rhohdata(\udata, \udata')\|_{\Cb^{3,\delta}}
	&\le
	\CY\|\udata\|_{\Cb^{3,\delta_1}}
	\| \udata'\|_{\Cb^{3,\delta_2}}
	\label{eq:C3est}
\end{align}
\end{subequations}
The dependency of
$C_{\ss,\delta}$ on $\ss$ and $\delta$ can be chosen
to be increasing.
Further,
\begin{align}\label{eq:YCHest}
\|\rhohdata(\udata, \udata')\|_{\Hb^{\ss,\delta}}
	\lesssim_{\ss,\delta} \|\udata\|_{\Cb^{\ss,\delta_1}}\|\udata'\|_{\Hb^{\ss,\delta_2}}
\end{align}
\end{lemma}
\begin{proof}
Relative to the basis in Lemma \ref{lem:basisgdata},
the map $\rhohdata$ is given by an array
whose entries are in $\mathcal{A}|_{\diamonddata}$ (see \eqref{eq:quotients}).
This follows from Lemma \ref{lem:dglabasisNEW} and the definition of $\pitilde$ in \eqref{eq:pitildedef}.
In particular, the entries are bounded in all $\Cb^{\ss,0}$-norms.

We conclude \eqref{eq:Biluv}. 
Write $\udata=\uodata\oplus \uIdata$, $\udata'=\uodata'\oplus\uIdata'$.
Using $\anchorg(0\oplus\uIdata)=0$,
\begin{align*}
\rhohdata(\udata,\udata') 
= 
\rhohdata(\uodata\oplus0)(\uodata'\oplus0) + \rhohdata(\uodata\oplus0)(0\oplus\uIdata')
\end{align*}
The two terms on the right are of the
form $\Ydata_0\oplus 0$ respectively $0\oplus \Ydata_{\I}$
with 
$\Ydata_0\in\lxdata(\diamonddata)$, 
$\Ydata_{\I}\in\Idata(\diamonddata)$.
Then 
$\|\rhohdata(\udata,\udata')\|_{\Hb^{\ss,\delta}}
=
\|\Ydata_0\|_{\Hb^{\ss,\delta}} + \|\Ydata_{\I}\|_{\Hb^{\ss-1,\delta}}$
by \eqref{eq:normsoffset}.
Since the entries of $\rhohdata$ 
are bounded in all \smash{$\Cb^{\ss,0}$}-norms,
it is easy to check that 
\begin{align*}
\|\Ydata_0\|_{\Hb^{\ss,\delta}} 
&\lesssim
\|\uodata\|_{\Cb^{0,\delta_1}}
\|\uodata'\|_{\Hb^{\ss,\delta_2}}
+
\|\uodata\|_{\Hb^{\ss,\delta_1}}
\|\uodata'\|_{\Cb^{0,\delta_2}}\\
&\qquad+
C_{\ss,\delta}
 \|\uodata\|_{\Hb^{\ss-1,\delta_1}} 
 \|\uodata'\|_{\Hb^{\ss-1,\delta_2}} \\
\|\Ydata_{\I}\|_{\Hb^{\ss-1,\delta}} 
&\lesssim
\|\uodata\|_{\Cb^{0,\delta_1}}
\|\uIdata'\|_{\Hb^{\ss-1,\delta_2}}
+
\|\uodata\|_{\Hb^{\ss-1,\delta_1}}
\|\uIdata'\|_{\Cb^{0,\delta_2}}\\
&\qquad+
C_{\ss,\delta}
\|\uodata\|_{\Hb^{\ss-2,\delta_1}} 
\|\uIdata'\|_{\Hb^{\ss-2,\delta_2}} 
\end{align*}
where $C_{\ss,\delta}>0$ only depends on $\ss,\delta$.
To obtain this estimate we also used \eqref{eq:sobolev} and 
$\lfloor\frac{\ss}{2}\rfloor+2\le \ss-1$, 
$\lfloor\frac{\ss-1}{2}\rfloor+2\le \ss-2$ 
requiring $\ss\ge6$.
This implies \eqref{eq:Biluv}.
The other estimates in the lemma are checked similarly.
It is easy to see that the dependency of $C_{\ss,\delta}$
on $\ss,\delta$ can be chosen to be increasing,
using the fact that $\ss$ runs over the integers,
and that the constant in \eqref{eq:sobolev}
is increasing.\qed
%
%
%
\end{proof}
\begin{lemma}\label{lem:MultilinEst}
There exists $\CYYY\ge1$ such that for 
all $n\in\Z_{\ge2}$, 
all $\ss\in\Z_{\ge6}$, 
all $\delta,\delta_1,\dots,\delta_n\ge0$ with $\delta_1+\dots+\delta_n=\delta$, 
and all $\vdata_1,\dots,\vdata_n\in\gxdata(\diamonddata)$: 
\begin{align}\label{eq:YYY}
\hspace{-1mm}
\|\rhohdata_{\vdata_1}\cdots\rhohdata_{\vdata_{n-1}} \vdata_n\|_{\Hb^{\ss,\delta}}
	\lesssim_{\ss,\delta}\CYYY^n
	\sum_{\substack{X\subset\{1,\dots,n\}\\ 
	|X| \le \ss-4}}
	\big(\prod_{i\in X} 
	\|\vdata_{i}\|_{\Hb^{\ss,\smash{\delta_{i}}}}\big)
	\big(\prod_{j\in X^c}
	\|\vdata_{j}\|_{\Cb^{3,\smash{\delta_{j}}}}\big)
\end{align}
\end{lemma}
\begin{proof}
Let $\CY$, $C_{\ss,\delta}$
be the constants in Lemma \ref{lem:Yestimates},
where the dependency of $C_{\ss,\delta}$ on $\ss,\delta$ is increasing.
It is instructive to consider first the example $n=3$, $\ss=7$.
We apply \eqref{eq:Yestimates} recursively:
First apply \eqref{eq:Biluv}, which yields 
\begin{align*}
\|\rhohdata_{\vdata_1} \rhohdata_{\vdata_{2}}\vdata_3\|_{\Hb^{7,\delta}}
&\le
\CY(\|\vdata_1\|_{\Hb^{7,\delta_1}}
\|\rhohdata_{\vdata_{2}}\vdata_3\|_{\Cb^{3,\delta_2+\delta_3}}
+
\|\vdata_1\|_{\Cb^{3,\delta_1}}
\|\rhohdata_{\vdata_{2}}\vdata_3\|_{\Hb^{7,\delta_2+\delta_3}})\\
&\quad+
C_{7,\delta}\|\vdata_1\|_{\Hb^{6,\delta_1}}
   \|\rhohdata_{\vdata_{2}}\vdata_3\|_{\Hb^{6,\delta_2+\delta_3}}
\intertext{
Now apply \eqref{eq:C3est} in the first term,
apply \eqref{eq:Biluv} in the second term
and use $C_{7,\delta_2+\delta_3}\le C_{7,\delta}$,
apply \eqref{eq:Biluv} in the third term and use
$C_{6,\delta_2+\delta_3}\le C_{7,\delta}$.
Then}
\|\rhohdata_{\vdata_1} \rhohdata_{\vdata_{2}}\vdata_3\|_{\Hb^{7,\delta}}
&\le
\CY^2\|\vdata_1\|_{\Hb^{7,\delta_1}}
\|\vdata_2\|_{\Cb^{3,\delta_2}}
\|\vdata_3\|_{\Cb^{3,\delta_3}}\\
&\hspace{-10mm}+
\CY^2(\|\vdata_1\|_{\Cb^{3,\delta_1}}
\|\vdata_2\|_{\Hb^{7,\delta_2}}
\|\vdata_3\|_{\Cb^{3,\delta_3}}
+
\|\vdata_1\|_{\Cb^{3,\delta_1}}
\|\vdata_2\|_{\Cb^{3,\delta_2}}
\|\vdata_3\|_{\Hb^{7,\delta_3}})\\
&\hspace{-10mm}+
\CY C_{7,\delta}\|\vdata_1\|_{\Cb^{3,\delta_1}}
\|\vdata_2\|_{\Hb^{6,\delta_2}}
\|\vdata_3\|_{\Hb^{6,\delta_3}}\\
&\hspace{-10mm}+
\CY C_{7,\delta}(\|\vdata_1\|_{\Hb^{6,\delta_1}}
\|\vdata_2\|_{\Hb^{6,\delta_2}}
\|\vdata_3\|_{\Cb^{3,\delta_3}}
+
\|\vdata_1\|_{\Hb^{6,\delta_1}}
\|\vdata_2\|_{\Cb^{3,\delta_2}}
\|\vdata_3\|_{\Hb^{6,\delta_3}})\\
&\hspace{-10mm}+
C_{7,\delta}^2\|\vdata_1\|_{\Hb^{6,\delta_1}}
\|\vdata_2\|_{\Hb^{5,\delta_2}}
\|\vdata_3\|_{\Hb^{5,\delta_3}}
\end{align*}
Now apply 
$\|\vdata_i\|_{\Hb^{\ell,\delta_i}}
\le\|\vdata_i\|_{\Hb^{7,\delta_i}}$
for $\ell\le 7$ by \eqref{eq:monotone},
and observe that the resulting term is bounded by the right
hand side of \eqref{eq:YYY},
using 
terms with $|X|=1$ for lines 1-2;
terms with $|X|=2$ for lines 3-4;
terms with $|X|=3$ for the last line.

Consider general $n$, $\ss$.
Similarly to the example,
we first apply \eqref{eq:Yestimates} recursively,
and then estimate \smash{$\|\vdata_i\|_{\Hb^{\ell,\delta_i}}
\le\|\vdata_i\|_{\Hb^{\ss,\delta_i}}$}
for $\ell\le \ss$.
In this way one obtains a sum of at most $3^{n-1}$ terms,
where each term is of the form 
\begin{equation}\label{eq:Xaux}
\textstyle
\CY^{|X^c|}
C_{\ss,\delta}^{|X|-1}
\Big(\prod_{i\in X} 
	\|\vdata_{i}\|_{\Hb^{\ss,\smash{\delta_{i}}}}\Big)
	\Big(\prod_{j\in X^c} 
	\|\vdata_{j}\|_{\Cb^{3,\smash{\delta_{j}}}}\Big)
\end{equation}
for some $X\subset\{1,\dots,n\}$.
In each term one has $|X|\le \ss-4$, 
because $|X|-1$ corresponds to the number of times
the second line of \eqref{eq:Biluv} is used,
and on the other hand 
the second line of \eqref{eq:Biluv}
is used at most $\ss-5$ times, since it 
reduces the number of derivatives in the norm,
hence $|X|-1\le \ss-5$.
Thus the left hand side of \eqref{eq:YYY} is bounded by 
$3^{n-1}$ times the sum of \eqref{eq:Xaux}
over the subsets $X\subset\{1,\dots,n\}$ with $|X|\le \ss-4$.
Setting $\CYYY=3\CY$ the claim follows.
\qed
\end{proof}
\begin{corollary}\label{cor:MultilinEst}
Lemma \ref{lem:MultilinEst} holds verbatim if in 
\eqref{eq:YYY} one replaces  
$$
\tsum_{\substack{X\subset\{1,\dots,n\}\\ 
	|X| \le \ss-4}}
	\qquad\text{by}\qquad
\tsum_{\substack{X\subset\{1,\dots,n\}\\ 
	|X| = \min\{\ss-4,n\}}}$$
\end{corollary}
\begin{proof}
The summation over $|X|\le \ss-4$
is the same as the summation over $|X|\le\min\{\ss-4,n\}$.
For each summand in \eqref{eq:YYY}
decompose $X^c = Y_1\cup Y_2$, where $|Y_1| = \min\{\ss-4,n\}-|X|$,
and where each element in $Y_1$ is smaller than each element in $Y_2$.
For each $j\in Y_1$ we apply \eqref{eq:sobolev}, 
which contributes a constant that only depends on $\ss,\delta$,
not on $n$.
In the right hand side of \eqref{eq:YYY} 
there are less than $2^{\ss}$ summands, 
thus the claim follows.\qed
\end{proof}
\begin{proof}[of Proposition \ref{prop:uBestimates}]
It suffices to check the proposition for $Y=\{\}$,
the statement for general $Y$ then follows using \eqref{eq:sobolev}.
$n=1$: By \eqref{eq:DPB1} and Lemma \ref{lem:cbasis}.
$n\ge2$: 
We have
\begin{align}\label{eq:BsumofTproof}
\|\tfrac{1}{n!}
\Bdata_{n}(\udata_1,\dots,\udata_n)\|_{\Hb^{\ss,\delta}}
\le
\tfrac{1}{n!}\tsum_{\pi\in S_n}
\|\tree_{n}(u_{\pi(1)},\dots,u_{\pi(n)})
\|_{\Hb^{\ss,\delta}}
\end{align}
where we choose $u_i\in\gx^{k_i}(\diamond)$ such that
$u_i|_{\diamonddata}=\udata_i$ and such that its components
relative to the basis in Lemma \ref{lem:basisg} are constant in $\p_{x^0}$.
Write each $\tree_n$ as a sum of three terms using \eqref{eq:TreeFormula}.
It then suffices to show that each such term is bounded 
by the right hand side of \eqref{eq:uBestimates},
which follows from Corollary \ref{cor:MultilinEst}, 
\eqref{eq:estpdbrdata} and \eqref{eq:sobolev}.
To check \eqref{eq:B2est} again use \eqref{eq:BsumofTproof}
and \eqref{eq:TreeFormula}, but now estimate 
$\pdata[u_{\pi(1)},u_{\pi(2)}]|_{\diamonddata}$ using \eqref{eq:br4est}
and estimate $\pdata \rhohdata_{\udata_{\pi(1)}}(\dg u_{\pi(2)})|_{\diamonddata}$
using \eqref{eq:YCHest} and \eqref{eq:ddataest}.
\qed
\end{proof}
We now consider the commutator of $\Bdata_n$ with 
the vector fields $\xvecjap\p_{x^i}$. 
\begin{definition}\label{def:Bcommutator}
For $n\in\Z_{\ge1}$ and $\alpha\in\N_0^3\setminus\{0\}$ define the commutator
\begin{align*}
\CBdata_{n,\alpha}(\udata_1,\dots,\udata_n)
&=
(\xvecjap\p_{\vec{x}})^{\alpha}\Bdata_n(\udata_1,\dots,\udata_n)
-
\sum_{i=1}^{n} 
\Bdata_n(\udata_1,\dots,(\xvecjap\p_{\vec{x}})^{\alpha}\udata_i,\dots,\udata_n)
\end{align*}
where we differentiate componentwise relative 
to the basis in Lemma \ref{lem:basisgdata}.
\end{definition}
\newcommand{\CBcom}{\tilde\CB}
\begin{lemma}\label{lem:uCBestimates}
There exists $\CBcom\ge1$ such that for all 
$n\in\Z_{\ge1}$, 
$\ss\in\Z_{\ge6}$,
$\delta,\delta_1,\dots,\delta_n\ge0$ with 
$\delta_1+\dots+\delta_n=\delta$, 
$Y\subset \{1,\dots,n\}$,
$\alpha\in\N_0^3\setminus\{0\}$,
and all $\udata_1,\dots,\udata_n\in\gxdata(\diamonddata)$:
\begin{align}
\label{eq:uCBestimates}
\begin{aligned}
&\|\tfrac{1}{n!}
\CBdata_{n,\alpha}(\udata_1,\dots,\udata_n)\|_{\Hb^{\ss,\delta}}
\lesssim_{\ss,\delta,\alpha,|Y|}
	\CBcom^n\tsum_{\substack{X\subset\{1,\dots,n\}\\ 
	|X| = \min\{n,\ss+|\alpha|-2\}
	}}
\\
&\textstyle
	\qquad\qquad\qquad\qquad\qquad\Big(\prod_{i\in X\cup Y} 
	\|\udata_{i}\|_{\Hb^{\ss+|\alpha|,\smash{\delta_{i}}}}\Big)
	\Big(\prod_{j\in X^c \setminus Y} 
	\|\udata_{j}\|_{\Cb^{4,\smash{\delta_{j}}}}
	\Big)
\end{aligned}
\end{align}
Furthermore, for $n=2$ and all $\ss\in\Z_{\ge1}$,
\begin{align}\label{eq:CB12est}
\|\CBdata_{n,\alpha}(\udata_1,\udata_2)\|_{\Hb^{\ss,\delta}}
\lesssim_{\ss,\delta,\alpha}
\|\udata_1\|_{\Cb^{\ss+|\alpha|,\delta_1}}\|\udata_2\|_{\Hb^{\ss+|\alpha|,\delta_2}}
\end{align}
\end{lemma}
\begin{proof}
This is similar to the proof of Proposition \ref{prop:uBestimates},
hence we are brief. 
It suffices to check the case $Y=\{\}$,
the general case then follows from \eqref{eq:sobolev}.
Write 
$
\CBdata_{n,\alpha}(\udata_1,\dots,\udata_n)
=
\tsum_{\pi\in S_n}
\pm\Ctree_{n,\alpha}(u_{\pi(1)},\dots,u_{\pi(n)})
$, using the same signs as in Definition \ref{def:Bn}, 
and where
\begin{align*}
\Ctree_{n,\alpha}(u_1,\dots,u_n)
&=
(\xvecjap\p_{\vec{x}})^{\alpha}\tree_{n}(u_1,\dots,u_n)\\
&\qquad\textstyle
-
\sum_{j=1}^n
\tree_{n}(
u_1,\dots,(\xvecjap\p_{\vec{x}})^{\alpha}u_j,\dots,u_n)
\end{align*}
The $u_i$ are chosen such that 
$u_i|_{\diamonddata}=\udata_i$ and 
such that the components are constant in $\p_{x^0}$.
Each $\smash{\Ctree}_{n,\alpha}(u_1,\dots,u_n)$ can be 
written as a sum of terms of the form \eqref{eq:TreeFormula},
where either one of the $\rhohdata$ is replaced by $\Crhohdata_{\alpha}$
defined by 
\begin{align*}
\Crhohdata_{\alpha}(\udata,\udata')
&=
(\xvecjap\p_{\vec{x}})^\alpha \rhohdata(\udata,\udata')
-
\rhohdata((\xvecjap\p_{\vec{x}})^\alpha \udata,\udata')
-
\rhohdata(\udata,(\xvecjap\p_{\vec{x}})^\alpha \udata')
\end{align*}
or the terms $[\cdot,\cdot]$ and $\dg$ are replaced
by their respective commutators with $(\xvecjap\p_{\vec{x}})^{\alpha}$.
These commutators can be estimated similarly to 
Lemma \ref{lem:Yestimates} respectively Lemma \ref{lem:estpdbrdata}
and they gain one derivative compared to the estimates without a commutator.
For example
$\|\Crhohdata_{\alpha}(\udata,\udata')\|_{\Hb^{\ss,\delta}}
\lesssim_{\ss,\delta,\alpha}
\|\udata\|_{\Hb^{\ss+|\alpha|-1,\smash{\delta_1'}}}
\|\udata'\|_{\Hb^{\ss+|\alpha|-1,\smash{\delta_2'}}}$
with \smash{$\delta_1'+\delta_2'=\delta$},
which gains one derivative compared to \eqref{eq:Biluv}.
One then derives an estimate similar to Lemma \ref{lem:MultilinEst}
with one of the $\rhohdata$ replaced by 
$\smash{\Crhohdata_{\alpha}}$,
from which Lemma \ref{lem:uCBestimates} follows.
\qed
\end{proof}

\subsection{Constraints as Maurer-Cartan equation}
\label{sec:MCequation}

We introduce \eqref{eq:MCinfintro} and show 
that it is equivalent to \eqref{eq:Constraints_Eq}.

Fix $\rhoB\in(0,1]$ such that:
\begin{itemize}
\item 
$\rhoB \le 1/(4 \max\{\CB,C_1,\CBcom\})$
where $\CB,C_1,\CBcom\ge1$ are the constants 
in Proposition \ref{prop:uBestimates},
Lemma \ref{lem:MultilinEst}, Lemma \ref{lem:uCBestimates} respectively.
\item 
For all $\udata=\udata_0\oplus\udata_{\I}\in \gxdata^1(\diamonddata)$,
if \smash{$\|\udata_0\|_{\Cb^{0,0}(\diamonddata)}\le \rhoB$} then
$\omega:=dx^0 + \anchorg(\udata)(x^0)$ is timelike
relative to $\gmink$ at every point on $\diamonddata$.
Such an $\rhoB$ exists
because $|(-1)-\gmink^{-1}(\omega,\omega)|
\lesssim\|\udata_0\|_{\Cb^{0,0}(\diamonddata)}(1+\|\udata_0\|_{\Cb^{0,0}(\diamonddata)})$
on $\diamonddata$, using \eqref{eq:dglabasisanchorNEW}.
\end{itemize}
\begin{definition}\label{def:MCinf}
For $\ss\in\Z_{\ge6}$ define the space
$\Wconv^{\ss} = \{\udata\in\Hb^{\ss,0}(\diamonddata,\gxdata^1)
\mid \|\udata\|_{\Cb^{4,0}(\diamonddata)} \le \rhoB
\}$
and define the map 
\begin{subequations}
\begin{equation}\label{eq:MCinfmap}
\MCinf:\; 
\Wconv^{\ss+1}
\;\to\; \Hb^{\ss,0}(\diamonddata,\cx^2)
\end{equation}
given by the absolutely converging sum (in $\Hb^{\ss,0}$,
by Corollary \ref{cor:bn3} and $\rhoB \le \tfrac{1}{4\CB}$)
\begin{equation}\label{eq:MCinf}
\MCinf(\udata) \;=\; 
\textstyle \sum_{n\ge1} \frac{1}{n!} \Bdata_n(\udata,\dots,\udata)
\end{equation}
\end{subequations}
\end{definition}
The expression \eqref{eq:MCinf} is the so-called curvature in the 
$L_\infty$ algebra described in Remark \ref{rem:MCinLinf},
the equation $\MCinf(\udata)=0$ is the Maurer-Cartan equation,
see e.g.~\cite{getzler}.
\begin{lemma}\label{rem:MCinfY}
Let $\ss\in\Z_{\ge6}$ and 
$\udata\in \Wconv^{\ss+1}$.
Then 
\begin{equation}\label{eq:MCinfnice}
\MCinf(\udata)
=
\pdata \left(\big(\tsum_{n\ge0}
	(-\rhohdata_{\udata})^{n}\big) 
\big(\dg u + \tfrac12[u,u]\big)|_{\diamonddata}\right)
\end{equation}
where $u\in\gx^1(\diamond)$ is any element 
with $u|_{\diamonddata}=\udata$.
Here 
$\wdata \mapsto \tsum_{n\ge0}(-\rhohdata_{\udata})^{n} \wdata$
is a bounded linear map $\Hb^{\ss,\delta}(\diamonddata,\gxdata^k) \to \Hb^{\ss,\delta}(\diamonddata,\gxdata^k)$,
absolutely converging in $\Hb^{\ss,\delta}$.
\end{lemma}
\begin{proof}
The statement about the map 
$\wdata \mapsto \tsum_{n\ge0}(-\rhohdata_{\udata})^{n} \wdata$
follows from Lemma \ref{lem:MultilinEst} and \eqref{eq:sobolev}.
Writing out \eqref{eq:MCinf} using Definition \ref{def:Bn}
and \eqref{eq:TreeFormula}, we obtain
\begin{align*}
\MCinf(\udata)
=
\pdata\left(
(\dg u)|_{\diamonddata}
+
\tsum_{n\ge2}
\tfrac12 (-1)^n \rhohdata_{\udata}^{n-2}
	\big([u,u]|_{\diamonddata}
	-
	2\rhohdata_{\udata}(\dg u)|_{\diamonddata}\big)\right)
\end{align*}
which is equal to \eqref{eq:MCinfnice}.\qed
\end{proof}
\begin{lemma}\label{lem:MCinfproperties}
The map $\MCinf$ satisfies:
\begin{itemize}
\item 
For all $\udata\in \smash{\Wconv^{7}}$:
$$
(dx^0 + \anchorg(\udata)(x^0))\,\idata\,\MCinf(\udata) =
\Pconstraints(\udata)
$$
where we use the multiplication \eqref{eq:modmult} restricted to $\diamonddata$.
The map $\idata$ and multiplication with $dx^0 + \anchorg(\udata)(x^0)$
are fiberwise injective.
In particular
\begin{equation}\label{eq:equivMCP}
\MCinf(\udata)|_q=0
\;\;\Leftrightarrow\;\;
\Pconstraints(\udata)|_q=0
\qquad\qquad
\forall q\in \diamonddata
\end{equation}
\item 
For all $\udata\in \smash{\Wconv^{8}}$
(the sum converges absolutely in $\Hb^{6,0}$):
\begin{equation}\label{eq:MCid}
\textstyle
\sum_{m\ge1} (-1)^{m} \frac{1}{(m-1)!} \Bdata_m(\idata\MCinf(\udata),\udata,\dots,\udata)
=0
\end{equation}
\end{itemize}
\end{lemma}
\begin{proof}
\textit{First item:}
Fix $u\in\gx^1(\diamond)$ with $u|_{\diamonddata}=\udata$.
Set $A=(\dg u + \tfrac12[u,u])|_{\diamonddata}$.
Note that $\idata \pdata = \one - \pi_0$ (see e.g.~\eqref{eq:pmatrix}), 
where $\pi_0$ is the projection onto the
second summand of \eqref{eq:ggGgG},
here viewed as a map on $\diamonddata$.
Together with \eqref{eq:MCinfnice} we obtain
\begin{align*}
&(dx^0+\anchorg(\udata)(x^0))\idata\MCinf(\udata)
=
(dx^0+\anchorg(\udata)(x^0)) 
(\one - \pi_0) \big(\tsum_{n\ge0}
(-\rhohdata_{\udata})^{n}\big) A\\
&\qquad\qquad=
\Pconstraints(\udata)
+
(dx^0+\anchorg(\udata)(x^0))\big(
-\pi_0  A + (\one - \pi_0) \big(\tsum_{n\ge1}
(-\rhohdata_{\udata})^{n}\big) A\big)
\end{align*}
Parse the second line as $\Pconstraints(\udata)+Q$.
Since $dx^0\pi_0=0$
and $\anchorg(\udata)(x^0)\rhohdata_{\udata}=0$,
\begin{align}\label{eq:Q}
Q
=
dx^0\big(\tsum_{n\ge1} (-\rhohdata_{\udata})^{n}\big) A)
-
\anchorg(\udata)(x^0)\pi_0
\big(\tsum_{n\ge0} (-\rhohdata_{\udata})^{n}\big) A
\end{align}
We have $\pi_0 = dx^0 \pitilde$, 
thus $\anchorg(\udata)(x^0) \pi_0 =- dx^0 \rhohdata_{\udata}$,
thus the two terms in \eqref{eq:Q} cancel, i.e.~$Q=0$.
This proves the identity in the first item.
Clearly $\idata$ is injective.
The definition of $\rhoB$ (second item) ensures that the one-form
$dx^0 + \anchorg(\udata)(x^0)$ is timelike, 
thus multiplication with this one-form is injective 
by Remark \ref{rem:modmultfibinj}.

\textit{Second item:}
Absolute convergence follows from Corollary \ref{cor:bn3}.
The identity follows from \ref{item:Jacobi}
and a combinatorial argument identical to \cite[Lemma 4.5]{getzler}.
%
\qed
\end{proof}

For $n_0\in\Z_{\ge1}$ we denote
$\MC_{\ge n_0}(\udata)=
\textstyle \sum_{n\ge n_0} \frac{1}{n!} \Bdata_n(\udata,\dots,\udata)$.

\smash{We abbreviate 
$\|{\cdot}\|_{\Hb^{\ss,\delta}(\diamonddata)}=\|{\cdot}\|_{\Hb^{\ss,\delta}}$
and analogously for $\Cb^{\ss,\delta}$}.
\begin{lemma}\label{lem:MCinf}
For 
$\ss\in\Z_{\ge7}$,
$\delta\ge1$,
$b\ge1$,
$\rfix\ge1$
and for all
$\udata,\udata'\in\Hb^{\ss,\delta}(\diamonddata,\gxdata^1)$
and 
$\kappa,\kappa'\in\gxdata^1(\diamonddata)$,
if 
\begin{subequations}\label{eq:ukappaassp}
\begin{align}
\|\udata\|_{\Cb^{4,0}}  &\le  \tfrac23\rhoB
	&\|\udata\|_{\Hb^{\ss,\delta}} &\le b 
	\label{eq:uassp}\\
 \|\kappa\|_{\Cb^{4,0}} &\le \tfrac13 \rhoB
 	&\|\kappa\|_{\Cb^{\ss,1}} &\le b 
	& \MCinf(\kappa)|_{|\vec{x}|\ge\rfix} &= 0
	\label{eq:kappaassp}
\end{align}
\end{subequations}
then 
$\udata+\kappa\in \Wconv^{\ss}$
and
(the estimate \eqref{eq:MC3} does not need \eqref{eq:uassp})
\begin{subequations}\label{eq:MCestimates}
\begin{align}
\|\MC_{\ge2}(\udata+\kappa)-\MC_{\ge2}(\kappa)\|_{\Hb^{\ss-1,\delta+1}}
	&\lesssim_{\ss,\delta,b}
	\|\udata\|_{\Hb^{\ss,\delta}}(\|\udata\|_{\Hb^{\ss,1}} + \|\kappa\|_{\Cb^{\ss,1}})
	\label{eq:MC1}\\
\|\MC_{\ge3}(\udata+\kappa)-\MC_{\ge3}(\kappa)\|_{\Hb^{\ss-1,\delta+1}}
	&\lesssim_{\ss,\delta,b}
	\smash{\|\udata\|_{\Hb^{\ss,\delta}}
	(\|\udata\|_{\Hb^{\ss,\frac12}}+\|\kappa\|_{\Cb^{\ss,1}})^{2}}
	\label{eq:MC2}\\
\|\MCinf(\kappa)\|_{\Hb^{\ss-1,\delta+1}} 
	&\lesssim_{\ss,\delta,b,\rfix}
	\|\kappa\|_{\Cb^{\ss,1}}
	\label{eq:MC3}
\end{align}
Moreover, if 
$\udata'$, $\kappa'$ also satisfy \eqref{eq:ukappaassp}
then $\udata'+\kappa'\in \Wconv^{\ss}$ and
\begin{align}
&
\|\MCinf(\kappa)-\MCinf(\kappa')\|_{\Hb^{\ss-1,\delta+1}}
	\lesssim_{\ss,\delta,b,\rfix}
	\|\kappa-\kappa'\|_{\Cb^{\ss,1}}
	\label{eq:MC4}\\
&\|\big(
	\MC_{\ge2}(\udata+\kappa)
	-\MC_{\ge2}(\kappa)\big)
	-
	\big(
	\MC_{\ge2}(\udata'+\kappa')
	-\MC_{\ge2}(\kappa')
	\big)\|_{\Hb^{\ss-1,\delta+1}} \nonumber \\
	&
	\qquad\qquad\qquad\qquad\lesssim_{\ss,\delta,b,\rfix}
	\RFP_0\|\udata-\udata'\|_{\Hb^{\ss,\delta}}+
	\|\udata'\|_{\Hb^{\ss,\delta}}\|\kappa-\kappa'\|_{\Cb^{\ss,1}}
	\label{eq:MC5}
\end{align}
\end{subequations}
where
$\RFP_0=\max\{
	\|\udata\|_{\Hb^{\ss,1}},\|\udata'\|_{\Hb^{\ss,1}},
	\|\kappa\|_{\Cb^{\ss,1}},\|\kappa'\|_{\Cb^{\ss,1}}
	\}$.
\end{lemma}
\begin{proof}
We will use the notation
\begin{equation}\label{eq:tensorprnot}
\Bdata_n(\underbrace{\vdata,\dots,\vdata}_{\text{$j$ times}},
	\underbrace{\wdata,\dots,\wdata}_{\text{$n-j$ times}})
=
\Bdata_n(\vdata^{\otimes j}, \wdata^{\otimes n-j})
\end{equation}
From the definition of the norms it is easy to see that 
\begin{align} \label{eq:kappaH}
\|\kappa\|_{\Hb^{\ss,\frac12}} 
	\lesssim_{\ss} \|\kappa\|_{\Cb^{\ss,1}} \le b
\end{align}
also using \eqref{eq:kappaassp}.
Then the triangle inequality yields $\udata+\kappa\in \Wconv^{\ss}$.

\textit{Proof of \eqref{eq:MC1}, \eqref{eq:MC2}:} 
Using multilinearity of the $\Bdata_n$ and the fact that
they are totally symmetric when all inputs are in degree one, 
see \ref{item:Bngrsym},
we obtain
$\MC_{\ge2}(\udata+\kappa)-\MC_{\ge2}(\kappa)
=
\textstyle
\sum_{n\ge2}
A_n$ where
\begin{align}\label{eq:An}
A_n:=\tsum_{j=1}^{n}
{n\choose j}
\tfrac{1}{n!}\Bdata_n (\udata^{\otimes j}, \kappa^{\otimes n-j})
\end{align}
Since $j\ge1$, 
there is always at least one decaying term $\udata$.
By Proposition \ref{prop:uBestimates}, 
\begin{align*}
\|A_2\|_{\Hb^{\ss-1,\delta+1}}
\lesssim_{\ss,\delta}
(
\|\udata\|_{\Hb^{\ss,1}}
+
\|\kappa\|_{\Cb^{\ss,1}})
\|\udata\|_{\Hb^{\ss,\delta}}
\end{align*}
where we estimate $\Bdata_2(\udata,\udata)$
using \eqref{eq:uBestimates} and
$\Bdata_2(\udata,\kappa)$ using \eqref{eq:B2est}.
Set
\begin{align}\label{eq:qdef}
q 
&:= 
\smash{\max\{ \|\udata\|_{\Hb^{\ss,\frac12}} , \|\kappa\|_{\Hb^{\ss,\frac12}}\}
\lesssim_{\ss}
\max\{ \|\udata\|_{\Hb^{\ss,\frac12}} , \|\kappa\|_{\Cb^{\ss,1}}\}
\le b}
\end{align}
where we use 
\eqref{eq:kappaH}, 
\eqref{eq:ukappaassp}, 
\eqref{eq:monotone},
$\delta\ge1$.
By Corollary \ref{cor:bn3} (first item),
for $n\ge3$:
\begin{align*}
\smash{\|\tfrac{1}{n!}\Bdata_n (\udata^{\otimes j}, \kappa^{\otimes n-j})\|_{\Hb^{\ss-1,\delta+1}}
\lesssim_{\ss,\delta,b}
\|\udata\|_{\Hb^{\ss,\delta}} q^2 (\tfrac14)^n}
\end{align*}
The sum in \eqref{eq:An} contributes an additional factor $2^n$
to $A_{n\ge3}$.
Thus $\sum_{n\ge2}A_n$ converges absolutely in $\Hb^{\ss-1,\delta+1}$.
Using \eqref{eq:qdef} this proves \eqref{eq:MC1}, \eqref{eq:MC2}.

\textit{Proof of \eqref{eq:MC3}:}
The term $\MCinf(\kappa)$ vanishes on $|\vec{x}|\ge\rfix$
by \eqref{eq:kappaassp}, 
thus
\begin{align}\label{eq:deltato0}
\|\MCinf(\kappa)\|_{\Hb^{\ss-1,\delta+1}}
&\lesssim_{\ss,\delta,\rfix}
\|\MCinf(\kappa)\|_{\Hb^{\ss-1,0}}
\end{align}
Now \eqref{eq:MC3} follows from 
Corollary \ref{cor:bn3} (second item) and 
\eqref{eq:kappaH}, \eqref{eq:kappaassp}.

\textit{Proof of \eqref{eq:MC4}:}
Like in \eqref{eq:deltato0}, the norm with weight 
$\delta+1$ is bounded by the norm with weight $0$,
up to a constant that depends on $\ss,\delta,\rfix$.
Adding and subtracting, and using the symmetry \ref{item:Bngrsym}, 
we obtain
\begin{align}\label{eq:MCMCkkp}
\MCinf(\kappa)-\MCinf(\kappa')
&= 
\textstyle
\sum_{n\ge1} \tsum_{j=0}^{n-1} 
\frac{1}{n!}\Bdata_n(\kappa-\kappa',
	\kappa^{\otimes j},{\kappa'}^{\otimes n-1-j})
\end{align}
By Corollary \ref{cor:bn3} (second item)
and using \eqref{eq:kappaH}, \eqref{eq:kappaassp}, \eqref{eq:monotone},
\[ 
\textstyle
\|\frac{1}{n!}\Bdata_n(\kappa-\kappa',
	\kappa^{\otimes j},{\kappa'}^{\otimes n-1-j})\|_{\Hb^{\ss-1,0}}
\lesssim_{\ss,b}
(\frac14)^{n}\|\kappa-\kappa'\|_{\Hb^{\ss,0}}
\lesssim_{\ss}
(\frac14)^{n}\|\kappa-\kappa'\|_{\Cb^{\ss,1}}
\]
The sum over $j$ in \eqref{eq:MCMCkkp} contributes 
an additional factor $n$.
Thus \eqref{eq:MC4} follows.

\textit{Proof of \eqref{eq:MC5}:}
The left side of \eqref{eq:MC5} equals
\smash{$\|Z^{(1)}+Z^{(2)}\|_{\Hb^{\ss-1,\delta+1}}$}
with
\begin{align*}
Z^{(1)}
&=
\MC_{\ge2}(\udata+\kappa)-\MC_{\ge2}(\udata'+\kappa)\\
Z^{(2)}
&=
\big(\MC_{\ge2}(\udata'+\kappa)-\MC_{\ge2}(\kappa)\big)
-
\big(\MC_{\ge2}(\udata'+\kappa')-\MC_{\ge2}(\kappa')\big)
\end{align*}
We bound the two terms separately.
Define $q'$ analogously to $q$ in \eqref{eq:qdef}
with $\udata$, $\kappa$ replaced by $\udata'$, $\kappa'$,
and set $\tilde{q}=\max\{q,q'\}$. Using  \eqref{eq:kappaH}, \eqref{eq:ukappaassp}, \eqref{eq:monotone}, $\delta\ge1$:
\begin{align}\label{eq:RRQQNEW}
\tilde{q} \lesssim_{\ss} \RFP_0 \le b
\end{align}
Using \ref{item:Bngrsym} we obtain
$\smash{Z^{(1)}} =\sum_{n\ge2} Z^{(1)}_{n}(\udata-\udata')$
where $\smash{Z^{(1)}_{n}}$ is the linear operator
\begin{align}
\textstyle
Z^{(1)}_{n}(W)
=
	\sum_{i=1}^{n}\sum_{a=0}^{i-1}
	{n\choose i} 
	\frac{1}{n!}\Bdata_n(W, \udata^{\otimes a},{\udata'}^{\otimes i-1-a},\kappa^{\otimes n-i})
	\label{eq:Z1n}
\end{align}
We bound this
using \eqref{eq:uBestimates}, \eqref{eq:B2est} in the case $n=2$
and using Corollary \ref{cor:bn3} (first item)
and \eqref{eq:RRQQNEW}, \eqref{eq:ukappaassp} in the case $n\ge3$.
This yields
\begin{align}\label{eq:Z12n}
\|Z^{(1)}_{n}(W)\|_{\Hb^{\ss-1,\delta+1}}
&\lesssim_{\ss,\delta,b}
n (\tfrac12)^n \RFP_0\|W\|_{\Hb^{\ss,\delta}} 
\end{align}
Thus
$\|Z^{(1)}\|_{\Hb^{\ss-1,\delta+1}}$
is bounded by the right hand side of \eqref{eq:MC5}.

Similarly, \smash{$Z^{(2)}=\tsum_{n\ge2}Z^{(2)}_{n}(\kappa-\kappa')$} 
where \smash{$Z^{(2)}_{n}$} is the linear operator
\begin{align*}
\smash{Z^{(2)}_{n}(W)}
&=
\textstyle
\sum_{i=1}^{n-1} 
\sum_{a=0}^{i-1}
{n\choose i}\frac{1}{n!}
\Bdata_n\big( 
W,\udata',
\kappa^{\otimes a},
{\kappa'}^{\otimes i-1-a},
{\udata'}^{\otimes n-1-i}
\big)
\end{align*}
%
We bound this using 
\eqref{eq:B2est} in the case $n=2$
and using Corollary \ref{cor:bn3} (first item)
in the case $n\ge3$, which yields
$
\|Z^{(2)}_{n}(W)\|_{\Hb^{\ss-1,\delta+1}}
\lesssim_{\ss,\delta,b}
n(\tfrac12)^n
\|\udata'\|_{\Hb^{\ss,\delta}}\|W\|_{\Cb^{\ss,1}} 
$.
Thus \smash{$\|Z^{(2)}\|_{\Hb^{\ss-1,\delta+1}}$}
is bounded by the right hand side of \eqref{eq:MC5}.\qed
\end{proof}
\section{Construction of solutions of the constraints}
We solve the constraint equations in the form \eqref{eq:MCinfintro},
see Theorem \ref{thm:mainL-infinity}.
This theorem implies Theorem \ref{thm:P=0} in the 
introduction, shown in Section \ref{sec:ProofThm1}. 

\subsection{Main theorem}

Recall the cone $\Ccone_{r_1,r_2}\subset\R^{10}$ in \eqref{eq:Ccone}
and $\Wconv^{\ss}$ and $\MCinf$ in Definition \ref{def:MCinf}.
Recall that $\MCinf(\udata)=0$ and $\Pconstraints(\udata)=0$
are equivalent, see Lemma \ref{lem:MCinfproperties}.
We use the norms in Definition \ref{def:normsdata} and \ref{def:cnorms}, 
and abbreviate
\smash{$\|{\cdot}\|_{\Hb^{\ss,\delta}(\diamonddata)}
	=\|{\cdot}\|_{\Hb^{\ss,\delta}}$},
analogously for $\Cb^{\ss,\delta}$.
\newcommand{\Cmain}{\tc{cyan(process)}{C}}
\begin{theorem}\label{thm:mainL-infinity}
For all 
$N\in\Z_{\ge8}$, 
$\delta>1$, 
$b\ge1$, 
$\rfix\ge1$
there exists 
$\eps_0\in (0,1]$ and $\Cmain\ge1$
such that for all $\eps\in(0,\eps_0]$ and all 
\begin{equation}\label{eq:u0Kmain}
\ulinsol\in \Hb^{N,\delta}(\diamonddata,\gxdata^1)
\qquad\qquad
\kerrdata:\; \Ccone_{\eps,\conenb{2^{-6}}} \;\to\; \gxdata^1(\diamonddata)
\end{equation}
the following holds.
If
\begin{enumerate}[label=\textnormal{({e\arabic*})}]
\item \label{item:thmu0sol}
	$\Bdata_1(\ulinsol) = 0$
\item \label{item:thmu0small}
\smash{$\|\ulinsol\|_{\Hb^{N,\delta}} \le b$
and $\|\ulinsol\|_{\Cb^{4,0}} \le \frac13\rhoB$}
\item \label{item:thmz*}
$z_* := \Charge(-\frac12\Bdata_2(\ulinsol,\ulinsol))$
satisfies
$z_* \in \Ccone_{\frac\eps2,\conenb{2^{-7}}}$ 
and
\smash{$\|\ulinsol\|_{\Hb^{N,\delta}} \le b |z_*|^{\frac12}$}
%
\item \label{item:thmkerr}
\smash{For all $z_1,z_2\in \Ccone_{\eps,\conenb{2^{-6}}}$:}
\begin{subequations}\label{eq:kerr12345}
\begin{align}
\|\kerrdata(z_1)\|_{\Cb^{4,0}}
	&\le \tfrac13\rhoB
		\label{eq:thmkerr1}\\
\|\kerrdata(z_1)\|_{\Cb^{N,1}}
	&\le b |z_1| 
		\label{eq:thmkerr2}\\
\|\kerrdata(z_1)-\kerrdata(z_2)\|_{\Cb^{N,1}}
	&\le b |z_1-z_2| 
		\label{eq:thmkerr3}\\
\MCinf(\kerrdata(z_1))|_{|\vec{x}|\ge\rfix}
	&= 0 \label{eq:thmkerrsol}\\
\Charge(\MCinf(\kerrdata(z_1))) 
	&= z_1
		\label{eq:thmkerr4}
\end{align}
\end{subequations}
\end{enumerate}
Then there exist
$\cdata\in \Hb^{N,\delta+1}(\diamonddata,\gxdata^1)$,
$z\in \Ccone_{\eps,\conenb{2^{-6}}}$
such that 
$\ulinsol + \kerrdata(z) + \cdata \in \Wconv^{N}$,
\begin{equation}\label{eq:MCu0Kv}
\MCinf\big( \ulinsol + \kerrdata(z) + \cdata \big) = 0
\end{equation}
and
$\|\cdata\|_{\Cb^{N-2,1}}
+ \|\kerrdata(z)\|_{\Cb^{N-2,1}} 
\le
\tfrac12 \|\ulinsol\|_{\Cb^{2,1}}$,
and such that:
\begin{itemize}
\item 
\textbf{Part 1.}
$\cdata$ and $\kerrdata(z)$ are quadratically small in $\ulinsol$,
more precisely:
\begin{subequations}\label{eq:thmvkz}
\begin{align}
\|\cdata\|_{\Hb^{N,\delta+1}}&\le \Cmain|z_*|
&
\|\cdata\|_{\Hb^{N,\delta+1}}
&\le \Cmain \|\ulinsol\|_{\Cb^{2,1}}\|\ulinsol\|_{\Hb^{2,\delta}}
\label{eq:thmvH}\\
|z-z_*|&\le \Cmain |z_*|^{\frac32}
&
\|\kerrdata(z)\|_{\Cb^{N,1}}
&\le \Cmain \|\ulinsol\|_{\Cb^{2,1}}\|\ulinsol\|_{\Hb^{2,\delta}}
\label{eq:thmz}
\end{align}
\end{subequations}
\item 
\textbf{Part 2.}
For all $N'\in\Z_{\ge N}$ and $\delta'\ge\delta$ and $b'\ge 1$,
if
\begin{enumerate}[label=\textnormal{({e\arabic*})},resume]
\item \label{item:ubound'}
$\|\ulinsol\|_{\Hb^{N',\delta'}} \le b'$
\item \label{item:kerrbound'}
For all \smash{$z_1,z_2\in\smash{\Ccone_{\eps,\conenb{2^{-6}}}}$}:
\begin{align}
\smash{\|\kerrdata(z_1)\|_{\Cb^{N',1}}
	 \le b'|z_1| \, ,
	 \quad
\|\kerrdata(z_1)-\kerrdata(z_2)\|_{\Cb^{N',1}}
	 \le b'|z_1-z_2|}\label{eq:thmkerr12'}
\end{align}
\end{enumerate}
then 
$\|\cdata\|_{\Hb^{N',\delta'+1}}
	\lesssim_{N',\delta',b,b',\rfix} 
	\|\ulinsol\|_{\Hb^{N',1}}
	\|\ulinsol\|_{\Hb^{N',\delta'}}$.
\item 
\textbf{Part 3.}
For all $r\ge\rfix$, if $\ulinsol|_{|\vec{x}|\ge r}=0$
then $\cdata|_{|\vec{x}|\ge r}=0$.
\end{itemize}
\end{theorem}
We note that the element $\kerrdata$ is not required to be
the initial data $\KSdata$ of the Kerr-Schild element \eqref{eq:KS};
only to conclude Theorem \ref{thm:P=0} we will choose $\kerrdata=\KSdata$.

The element $\cdata$ will be constructed in the gauge subspace
$\gxdataG^1(\diamonddata)$, thus 
$\cdata=\idata\cc$ with $\cc\in\cx^1(\diamonddata)$.
To prove Theorem \ref{thm:mainL-infinity} we use
Proposition \ref{prop:Gxhomotopy} to formulate \eqref{eq:MCu0Kv}
as a fixed point equation for $\cc$ and $\tilde{z}=z-z_*$, of the form
\[ 
\cc = \FP_{\cx}(\cc,\tilde{z})
\qquad\qquad
\tilde{z} = \FP_{\Charge}(\cc,\tilde{z})
\]
where $\FP_{\cx}$ and $\FP_{\Charge}$ are defined in the next lemma.
We will use \eqref{eq:MCid}
(based on the higher Jacobi identities \ref{item:Jacobi})
to show that the fixed point solves \eqref{eq:MCu0Kv}.

\begin{lemma}\label{lem:phi}
For all
$\ss\in\Z_{\ge7}$, 
$\delta>1$, 
$b\ge1$, 
$\rfix\ge1$
there exists a constant $\CFP\ge1$ such that 
for all $\rr\in (0,1]$ and 
$z_* \in \Ccone_{\frac{\rr}{2},\conenb{2^{-7}}}$ and
$$
\smash{\ulinsol\in 
\Hb^{\ss,\delta}(\diamonddata,\gxdata^1)
\qquad\qquad
\kerrdata:\;\; 
	\Ccone_{\rr,\conenb{2^{-6}}} 
	\;\to\;
	\gxdata^1(\diamonddata)}
$$
if these parameters satisfy  
\ref{item:thmu0sol};
\ref{item:thmu0small} with $N$ replaced by $\ss$;
and \eqref{eq:thmkerr1}, \eqref{eq:thmkerr2}, 
\eqref{eq:thmkerr3}, \eqref{eq:thmkerrsol} of \ref{item:thmkerr}
with $N$ replaced by $\ss$,
then the following holds:
Define
\begin{subequations}\label{eq:FPspacedef}
\begin{align}
\FPspace^{\ss,\delta}
= 
\smash{\big\{}
	(\cc,\tilde{z})
	\in \Hb^{\ss,\delta}(\diamonddata,\cx^1) \times \R^{10}  
	\mid 
	&\ \|\cc\|_{\Cb^{4,0}}\le \tfrac13\rhoB,\  
	\|\cc\|_{\Hb^{\ss,\delta}} \le b , \label{eq:FPspacedefc} \\
	&\;\; |\tilde{z}|\le \conenb{2^{-8}}|z_*| 
	\smash{\big\}}\label{eq:FPspacedefz}
\end{align}
\end{subequations}
\begin{itemize}
\item 
For all $(\cc,\tilde{z})\in \FPspace^{\ss,\delta}$ one has
$z_*+\tilde{z}\in\Ccone_{\rr,\conenb{2^{-6}}}$, 
$\ulinsol+\kerrdata(z_*+\tilde{z})+\idata \cc \in \Wconv^{\ss}$ and
\begin{align}
\MCinf\big(\ulinsol+\kerrdata(z_*+\tilde{z})+\idata \cc\big) - \dc \cc
	&\;\in\;
\Hb^{\ss-1,\delta+1}(\diamonddata,\cx^2)
	\label{eq:MCukv}
\end{align}
\item 
Define \smash{$\FP_{\cx}
	: 
	\FPspace^{\ss,\delta}
	\to 
	\Hb^{\ss,\delta+1}(\diamonddata,\cx^1)$}
by
\begin{align*}
\FP_{\cx}(\cc,\tilde{z})
	=
	-\Gc\left(\MCinf\big(\ulinsol+\kerrdata(z_*+\tilde{z})+\idata \cc\big) - \dc \cc \right)
\end{align*}
using (the continuous extension of)
$\Gc$ in Proposition \ref{prop:Gxhomotopy} with $\Rhomotopy=1$;
the map $\Gc$ can indeed be applied by \eqref{eq:MCukv}.
Then
\begin{align}
&\|\FP_{\cx}(\cc,\tilde{z})\|_{\Hb^{\ss,\delta+1}}
\le
	\CFP|z_*|  \nonumber\\
	&\qquad\qquad\quad
	+\CFP
	\big(\|\ulinsol\|_{\Hb^{\ss,1}}
		+
		\|\cc\|_{\Hb^{\ss,1}}
		+
		|z_*|\big)
	\big(\|\ulinsol\|_{\Hb^{\ss,\delta}}+\|\cc\|_{\Hb^{\ss,\delta}}\big)
		\label{eq:phicest}\\
&\|\FP_{\cx}(\cc,\tilde{z})-\FP_{\cx}(\cc',\tilde{z}')\|_{\Hb^{\ss,\delta+1}}
\le
	\CFP
	\big(|\tilde{z}-\tilde{z}'|
	+
	\smash{\RFP\|\cc-\cc'\|_{\Hb^{\ss,\delta}}}\big)%
	\label{eq:phicestcontraction}
\end{align}
where 
\smash{$\RFP=\max\big\{
	\|\ulinsol\|_{\Hb^{\ss,1}},
	\|\cc\|_{\Hb^{\ss,1}},
	\|\cc'\|_{\Hb^{\ss,1}},
	|z_*|
	\big\}$}.

\item 
Define $\FP_{\Charge}:\FPspace^{\ss,\delta}\to \R^{10}$ by 
\begin{align*}
	\FP_{\Charge}(\cc,\tilde{z})
	&=
	-\Charge\smash{\Big(}
	\MCinf(\ulinsol+\kerrdata(z_*+\tilde{z})+\idata \cc)-\dc \cc
	\\
	&\qquad\quad
	-\MCinf(\kerrdata(z_*+\tilde{z}))
	-\tfrac12\Bdata_2(\ulinsol,\ulinsol)\smash{\Big)}
\end{align*}
using (the continuous extension of)
$\Charge$ in Definition \ref{def:MLCA}.
Then
\begin{align}
&|\FP_{\Charge}(\cc,\tilde{z})|
\le
\CFP (\|\ulinsol\|_{\Hb^{\ss,\frac12}}^2+\|\cc\|_{\Hb^{\ss,1}}+|z_*|)
(\|\ulinsol\|_{\Hb^{\ss,\delta}}+\|\cc\|_{\Hb^{\ss,\delta}})
	\label{eq:phiChest}\\
&|\FP_{\Charge}(\cc,\tilde{z})-\FP_{\Charge}(\cc',\tilde{z}')|
\le
\CFP\smash{\big(}\RFP\|\cc-\cc'\|_{\Hb^{\ss,\delta}}\nonumber\\
	&\hspace{4.3cm}+ 
	(\|\ulinsol\|_{\Hb^{\ss,\delta}}+\|\cc'\|_{\Hb^{\ss,\delta}})
	|\tilde{z}-\tilde{z}'|\smash{\big)}
	\label{eq:phiChdiff}
\end{align}
\end{itemize}
\end{lemma}
\begin{proof}
When referring to 
\ref{item:thmu0sol},
\ref{item:thmu0small}, \eqref{eq:thmkerr1}, 
\eqref{eq:thmkerr2}, \eqref{eq:thmkerr3}, 
\eqref{eq:thmkerrsol} it will be understood that
$N$ is replaced by $\ss$.
Instead of specifying $\CFP$ up front, we will make 
finitely many admissible largeness assumptions on $\CFP$,
where admissible means that they depend only on 
$\ss,\delta,b,\rfix$.
We will use the monotonicity \eqref{eq:monotone} without further reference.
Let $(\cc,\tilde{z}),(\cc',\tilde{z}')\in\FPspace^{\ss,\delta}$.
Abbreviate
\begin{align*}
\udata &= \ulinsol + \idata \cc 
&
\udata' &= \ulinsol + \idata \cc' 
&
z &= z_* + \tilde{z}
&
z' &= z_* + \tilde{z}'
\end{align*}
By Lemma \ref{lem:cbasis} the map $\idata$ is given,
relative to the bases in Lemma \ref{lem:cbasis} and \ref{lem:basisgdata},
by the identity matrix.
Thus the norms of $\idata \cc$ and $\cc$ 
respectively of $\idata \cc'$ and $\cc'$ are equal,
and this will be used without further reference.
By \ref{item:thmu0sol}, \eqref{eq:DPB1}
and by $\pdata\idata=\one$, 
\begin{equation}\label{eq:dvBv}
\Bdata_1(\udata)  = \Bdata_1(\idata \cc)  = \dc \cc
\end{equation}
and analogously for $\udata',\cc'$. 
Using $z_*\in\smash{\Ccone_{\frac{\rr}{2},\conenb{2^{-7}}}}$
and \eqref{eq:FPspacedefz}, one checks that
\begin{equation}\label{eq:zz'small}
|z|,|z'| \le \smash{\tfrac32}|z_*|\le 1
\qquad
z,z' \in \Ccone_{\rr,\conenb{2^{-6}}}
\end{equation}

We will repeatedly use Lemma \ref{lem:MCinf}
with parameters $\ss,\delta,b,\rfix$
and $\udata,\kappa$ and $\udata',\kappa'$ there
given by 
$\ss,\delta,2b,\rfix$ and 
$\udata,\kerrdata(z)$ and $\udata',\kerrdata(z')$ here
(use $\ss\ge7$).
The assumptions \eqref{eq:ukappaassp} are satisfied
both for $\udata,\kerrdata(z)$ and for $\udata',\kerrdata(z')$:
for \eqref{eq:uassp} use the triangle inequality, 
\ref{item:thmu0small}, \eqref{eq:FPspacedefc};
for \eqref{eq:kappaassp} 
use 
\eqref{eq:zz'small}, \eqref{eq:thmkerr1}, \eqref{eq:thmkerr2}, \eqref{eq:thmkerrsol}.

\textit{Proof of first item.}
Use \eqref{eq:zz'small}
and $\udata+\kerrdata(z)\in \Wconv^{\ss}$
by Lemma \ref{lem:MCinf}.
Further
\begin{align}
\MCinf(\udata+\kerrdata(z)) - \dc \cc 
=
\big(\MC_{\ge2}(\udata+\kerrdata(z)) - \MC_{\ge2}(\kerrdata(z))\big)
+\MCinf(\kerrdata(z))
\label{eq:MCinfu+kerr}
\end{align}
using \eqref{eq:dvBv},
thus \eqref{eq:MCukv} follows from \eqref{eq:MC1} and \eqref{eq:MC3}.

\textit{Proof of second item.}
\eqref{eq:phicest}: 
By Proposition \ref{prop:Gxhomotopy} and then using \eqref{eq:MCinfu+kerr},
\begin{align*}
\|\FP_{\cx}(\cc,\tilde{z})\|_{\Hb^{\ss,\delta+1}}
&\lesssim_{\ss,\delta}
\|\MCinf(\udata+\kerrdata(z))-\dc \cc\|_{\Hb^{\ss-1,\delta+1}}\\
&\le
\|\MC_{\ge2}(\udata+\kerrdata(z)) 
- \MC_{\ge2}(\kerrdata(z))\|_{\Hb^{\ss-1,\delta+1}}
+
\|\MCinf(\kerrdata(z))\|_{\Hb^{\ss-1,\delta+1}}
\end{align*}
We bound the two terms using \eqref{eq:MC1} 
respectively \eqref{eq:MC3} and \eqref{eq:thmkerr2},
obtaining
\begin{align*}
\|\FP_{\cx}(\cc,\tilde{z})\|_{\Hb^{\ss,\delta+1}}
&\lesssim_{\ss,\delta,b,\rfix}
\|\udata\|_{\Hb^{\ss,\delta}}(\|\udata\|_{\Hb^{\ss,1}}+b|z|)
+b|z|
\end{align*}
This implies \eqref{eq:phicest} by 
the triangle inequality for $\udata=\ulinsol+\idata\cc$, \eqref{eq:zz'small},
an admissible largeness assumption on $\CFP$.
\eqref{eq:phicestcontraction}:
By linearity of $\Gc$ and Proposition \ref{prop:Gxhomotopy},%
\begin{align*}
&\|\FP_{\cx}(\cc,\tilde{z})-
	\FP_{\cx}(\cc',\tilde{z}')\|_{\Hb^{\ss,\delta+1}}
	\\
&\quad\lesssim_{\ss,\delta}
\|\big(\MCinf(\udata+\kerrdata(z)) - \dc \cc\big)
-
\big(\MCinf(\udata'+\kerrdata(z')) - \dc \cc'\big)\|_{\Hb^{\ss-1,\delta+1}}
\end{align*}
We estimate this using \eqref{eq:MC4} and \eqref{eq:MC5}.
For this purpose we add and subtract 
$\MCinf(\kerrdata(z))-\MCinf(\kerrdata(z'))$,
and use 
$\MCinf(\udata+\kerrdata(z))-\dc\cc-\MCinf(\kerrdata(z))
= \MC_{\ge2}(\udata+\kerrdata(z))-\MC_{\ge2}(\kerrdata(z))$
by \eqref{eq:dvBv} and analogously with $\udata',z'$.
Then
\begin{align*}
\|\FP_{\cx}(\cc,\tilde{z})-\FP_{\cx}(\cc',\tilde{z}')\|_{\Hb^{\ss,\delta+1}}
\lesssim_{\ss,\delta,b,\rfix}
	\RFP\|\cc-\cc'\|_{\Hb^{\ss,\delta}}
	+
	\|\kerrdata(z)-\kerrdata(z')\|_{\Cb^{\ss,1}}
\end{align*}
where we also use
the fact that $\RFP_0$ in \eqref{eq:MC5} satisfies 
$\RFP_{0}\lesssim_{b}\RFP$ by \eqref{eq:thmkerr2}, \eqref{eq:zz'small};
and use 
\smash{$\|\udata'\|_{\Hb^{\ss,\delta}}\le 2b$} 
by \ref{item:thmu0small}, \eqref{eq:FPspacedefc}.
Now \eqref{eq:thmkerr3}
and an admissible largeness assumption on $\CFP$
yield \eqref{eq:phicestcontraction}.

\textit{Proof of third item.}
\eqref{eq:phiChest}:
By Lemma \ref{lem:ChargeEst} and $\delta+1>2$ and $\ss-1\ge1$,
\begin{align*}
\smash{|\FP_{\Charge}(\cc,\tilde{z})|
\lesssim_{\delta}
\|\MCinf(\udata+\kerrdata(z)) - \dc \cc - \MCinf(\kerrdata(z)) 
- \tfrac12\Bdata_2(\ulinsol,\ulinsol)
\|_{\Hb^{\ss-1,\delta+1}}}
\end{align*}
The linear terms cancel using \eqref{eq:dvBv}.
Using the triangle inequality,
\begin{align*}
|\FP_{\Charge}(\cc,\tilde{z})|
	&\lesssim_{\delta}
	\|\Bdata_2(\idata \cc,\idata \cc)\|_{\Hb^{\ss-1,\delta+1}}
	+
	\|\Bdata_2(\ulinsol,\idata \cc)\|_{\Hb^{\ss-1,\delta+1}}
	+
	\|\Bdata_2(\udata,\kerrdata(z))\|_{\Hb^{\ss-1,\delta+1}} \\
	&
	\qquad+
	\smash{\|\MC_{\ge3}(\udata+\kerrdata(z))
	-\MC_{\ge3}(\kerrdata(z))\|_{\Hb^{\ss-1,\delta+1}}}
\end{align*}
We bound the quadratic terms using Proposition \ref{prop:uBestimates}
(use \eqref{eq:uBestimates} for the first two terms and
\eqref{eq:B2est} for the third), and the remaining term
using \eqref{eq:MC2}, obtaining
\begin{align*}
|\FP_{\Charge}(\cc,\tilde{z})|
	&\lesssim_{\ss,\delta,b}
	(\|\cc\|_{\Hb^{\ss,1}}+\|\kerrdata(z)\|_{\Cb^{\ss,1}})
	(\|\ulinsol\|_{\Hb^{\ss,\delta}}+\|\cc\|_{\Hb^{\ss,\delta}})
	\\
	&\qquad\quad+
	\smash{(\|\ulinsol\|_{\Hb^{\ss,\frac12}}
	+\|\cc\|_{\Hb^{\ss,\frac12}}
	+\|\kerrdata(z)\|_{\Cb^{\ss,1}})^{2}
	(\|\ulinsol\|_{\Hb^{\ss,\delta}}+\|\cc\|_{\Hb^{\ss,\delta}})}
\end{align*}
This implies \eqref{eq:phiChest}
using \eqref{eq:thmkerr2}, \eqref{eq:zz'small};
using 
\smash{$\|\ulinsol\|_{\Hb^{\ss,1/2}},
\|\cc\|_{\Hb^{\ss,1/2}},
|z_*|\lesssim b$} by \ref{item:thmu0small}, \eqref{eq:FPspacedefc}, \eqref{eq:zz'small}; 
and by an admissible largeness assumption on $\CFP$.
\eqref{eq:phiChdiff}:
By linearity of $\Charge$,
Lemma \ref{lem:ChargeEst} and $\delta+1>2$ and $\ss-1\ge1$, 
and also using \eqref{eq:dvBv},
\begin{align*}
|\FP_{\Charge}(\cc,\tilde{z})-\FP_{\Charge}(\cc',\tilde{z}')|
	&\lesssim_\delta
	\|\big(
	\MC_{\ge2}(\udata+\kerrdata(z))-\MC_{\ge2}(\kerrdata(z))
	\big)\\
	&\smash{\qquad-\big(
	\MC_{\ge2}(\udata'+\kerrdata(z'))-\MC_{\ge2}(\kerrdata(z'))
	\big)\|_{\Hb^{\ss-1,\delta+1}}}
\end{align*}
This is bounded by the right hand side of \eqref{eq:phiChdiff},
using \eqref{eq:MC5}
and then using \eqref{eq:thmkerr3}
and again $\RFP_{0}\lesssim_{b}\RFP$,
and by an admissible largeness assumption on $\CFP$.
\qed
\end{proof}

\renewcommand{\aa}{\tc{applegreen}{a}}

\begin{proof}[of Theorem \ref{thm:mainL-infinity}]
Instead of specifying $\eps_0$, $\Cmain$ up front, 
we will make finitely many admissible smallness assumptions
on $\eps_0$ and largeness assumptions on $\Cmain$, 
where admissible means that they only depend on 
$N$, $\delta$, $b$, $\rfix$.
Recall the monotonicity of the norms \eqref{eq:monotone},
we will use this without further reference.

\begin{table}
\centering
\footnotesize{
\begin{tabular}{c|c|ccc}
	Parameters 
	&
	\multicolumn{3}{c}{Parameters used to invoke Lemma \ref{lem:phi}}
	\\
	in Lemma \ref{lem:phi}
	&
	\textit{Existence}
	&
	\multicolumn{3}{c}{Part 2}
	\\
	&
	\textit{and Part 1}
	&
	\textit{regularity}
	&
	\textit{decay, base}
	&
	\textit{decay, step}
	\\
	\hline
	$\ss$, $\delta$
	&
	$N$, $\delta$
	&
	$N'$, $\delta$
	&
	$N'$, $\slashed\delta$
	&
	$N'$, $\delta_0$
	\\
	$b$, $\rfix$
	&
	$b$, $\rfix$
	&
	$\max\{b',b'_\ell\}$, $\rfix$
	&
	$\max\{b',b''\}$, $\rfix$
	&
	$\max\{b',b'''\}$, $\rfix$
	\\
	$\ulinsol$, $\kerrdata$, $\rr$, $z_*$
	&$\ulinsol$, $\kerrdata$, $\eps$, $z_*$ 
	&$\ulinsol$, $\kerrdata$, $\eps$, $z_*$ 
	&$\ulinsol$, $\kerrdata$, $\eps$, $z_*$ 
	&$\ulinsol$, $\kerrdata$, $\eps$, $z_*$ 
\end{tabular}
}
\captionsetup{width=115mm}
\caption{
The first column lists the input parameters of Lemma \ref{lem:phi}. 
The second column specifies the choice of input parameters used to invoke Lemma \ref{lem:phi} in the proof of existence and Part 1 of Theorem \ref{thm:mainL-infinity}. Analogously for the other columns.}
\label{tab:LemmaParameters}
\end{table}
%

\textit{Proof of existence and Part 1.}
Note that \ref{item:thmz*} implies
\begin{equation}\label{eq:z*eps}
|z_*| \le \eps\le\eps_0
\end{equation}

We use Lemma \ref{lem:phi} with parameters 
in the 'Existence and Part 1' column of Table \ref{tab:LemmaParameters};
clearly the assumptions of the lemma are satisfied.
Let $\CFPapply\ge1$ be the constant produced by Lemma \ref{lem:phi},
which depends only on $N,\delta,b,\rfix$,
in particular $\eps_0$ and $\Cmain$ are allowed to depend on $\CFPapply$.
Let $\FP_{\cx}$, $\FP_{\Charge}$, $\FPspace^{N,\delta}$ be as in the lemma.

Set $C_1 = \CFPapply\big(1+(b+1)(b+2)\big)$ and 
$C_2 = \CFPapply(b^2+C_1 +1) (b+1) $ and define
\[ 
\tildeFPspace
=
\big\{(\cc,\tilde{z}) \in \Hb^{N,\delta+1}(\diamonddata,\cx^1)\times\R^{10}
\mid
\|\cc\|_{\Hb^{N,\delta+1}}\le C_1 |z_*|,\  
	|\tilde{z}|\le C_2 |z_*|^{\frac32}\big\}
\]
On this space define the following metric, 
where $\aa = (4 \CFPapply b |z_*|^{\frac12})^{-1}$:
\begin{equation}\label{eq:dmetric}
d\big((\cc,\tilde{z}),(\cc',\tilde{z}')\big)
=
\|\cc-\cc'\|_{\Hb^{N,\delta+1}} + \aa |\tilde{z}-\tilde{z}'|
\end{equation}
We make admissible smallness assumptions on $\eps_0$:
$C_1\eps_0 \le \min\{(\Csob{4,0})^{-1}\tfrac13\rhoB,b\}$, 
$C_2 \eps_0^{1/2} \le \conenb{2^{-8}}$
with $\smash{\Csob{4,0}}\ge1$ the constant in \eqref{eq:sobolev}.
Then 
$
\tildeFPspace
\subset\FPspace^{N,\delta}$
by \eqref{eq:z*eps}, \eqref{eq:monotone}, \eqref{eq:sobolev}. 
We show that
for all $(\cc,\tilde{z}), (\cc',\tilde{z}')\in \smash{\tildeFPspace}$:
\begin{subequations}\label{eq:phicCall}
\begin{align}
(\FP_{\cx}(\cc,\tilde{z}),\FP_{\Charge}(\cc,\tilde{z}))
&\in \smash{\tildeFPspace} \label{eq:phicCFPest}\\
d\big(
\big(\FP_{\cx}(\cc,\tilde{z}),\FP_{\Charge}(\cc,\tilde{z})\big),
\big(\FP_{\cx}(\cc',\tilde{z}'),\FP_{\Charge}(\cc',\tilde{z}')\big)
\big)
&\le
\tfrac12d\big((\cc,\tilde{z}),(\cc',\tilde{z}')\big)
\label{eq:phicCdiff}
\end{align}
\end{subequations}
Proof of \eqref{eq:phicCFPest}:
By 
\eqref{eq:phicest}, \eqref{eq:phicestcontraction}, 
\ref{item:thmz*}, \eqref{eq:monotone} we have
\begin{align*}
\|\FP_{\cx}(\cc,\tilde{z})\|_{\Hb^{N,\delta+1}}
&\le
	|z_*|\CFPapply
	\big(1
	+
	(
	b
	+
	C_1 |z_*|^{\frac12})
	(
	b
	+
	 C_1 |z_*|^{\frac12}
	+
	|z_*|^{\frac12})\big)\\
|\FP_{\Charge}(\cc,\tilde{z})|
	&\le
	|z_*|^{\frac32}\CFPapply(b^2 +C_1+1)(b +C_1 |z_*|^{\frac12})
\intertext{
which implies \eqref{eq:phicCFPest} by  
\eqref{eq:z*eps} and by making the admissible 
smallness assumption $C_1\eps_0^{1/2} \le1$ on $\eps_0$.
Proof of \eqref{eq:phicCdiff}:
By 
\eqref{eq:phiChest}, \eqref{eq:phiChdiff}, \ref{item:thmz*}, \eqref{eq:monotone}
we have}
\|\FP_{\cx}(\cc,\tilde{z})-\FP_{\cx}(\cc',\tilde{z}')\|_{\Hb^{N,\delta+1}}
	&\le
	\CFPapply
	\big(
	\RFP\|\cc-\cc'\|_{\Hb^{N,\delta}}
	+
	\aa^{-1}\aa|\tilde{z}-\tilde{z}'|
	\big)\\
\aa|\FP_{\Charge}(\cc,\tilde{z})-\FP_{\Charge}(\cc',\tilde{z}')|
	&\le
	\CFPapply
	\big(
	\aa \RFP \|\cc-\cc'\|_{\Hb^{N,\delta}}
	+
	|z_*|^{\frac12}(b+C_1|z_*|^{\frac12}) \aa |\tilde{z}-\tilde{z}'|
	\big)
\end{align*}
with $\RFP$ from \eqref{eq:phiChdiff},
which satisfies
$\RFP\le |z_*|^{\frac12}
	\max\{b,C_1|z_*|^{\frac12},|z_*|^{\frac12}\}
	\le |z_*|^{\frac12}b$.
Thus $\CFPapply\aa\RFP \le \smash{\frac14}$.
Further 
$\CFPapply\max\{\RFP,\aa^{-1},|z_*|^{\frac12}(b+C_1|z_*|^{\frac12})\}
\le \smash{\frac14}$ using $b,\CFPapply\ge1$ and 
the admissible smallness assumption $4\CFPapply^2 b\eps_0^{1/2}\le\smash{\frac14}$.
So \eqref{eq:phicCdiff} holds, using \eqref{eq:monotone}.

For $\ell\in\Z_{\ge0}$ define $(\cc_\ell,\tilde{z}_\ell)$ by 
\begin{align}\label{eq:vlzldef}
(\cc_0,\tilde{z}_0) = (0,0) \in \smash{\tildeFPspace}
&&
(\cc_{\ell+1},\tilde{z}_{\ell+1})
=
\big(\FP_{\cx}(\cc_{\ell},\tilde{z}_{\ell}),
	\FP_{\Charge}(\cc_{\ell},\tilde{z}_{\ell})\big)
	\in \smash{\tildeFPspace}
\end{align}
using \eqref{eq:phicCFPest}.
By \eqref{eq:phicCdiff} the sequence 
$(\cc_\ell,\tilde{z}_\ell)$ is Cauchy with respect to the metric
\eqref{eq:dmetric}. Let
$(\cc,\tilde{z}) \in \smash{\tildeFPspace}$ be the limit. 
This satisfies the fixed point equation
\begin{equation} \label{eq:vzFP}
(\cc,\tilde{z}) = \big(\FP_{\cx}(\cc,\tilde{z}),\FP_{\Charge}(\cc,\tilde{z})\big)
\end{equation}

Set 
\[ 
\cdata = \idata \cc
\qquad
z=z_*+\tilde{z}
\]
We check that they satisfy the properties in the theorem.
The map $\idata$ is given,
relative to the bases in Lemma \ref{lem:cbasis} and \ref{lem:basisgdata},
by the identity matrix.
Thus the norms of $\cdata$ and $\cc$ are equal.
We have
$\cdata\in \Hb^{N,\delta+1}(\diamonddata,\gxdata^1)$,
$z\in \Ccone_{\eps,\conenb{2^{-6}}}$,
$\ulinsol + \kerrdata(z) + \cdata \in \Wconv^{N}$
as claimed, by 
$(\cc,\tilde{z})\in\tildeFPspace\subset\FPspace^{N,\delta}$
and by the first item of Lemma \ref{lem:phi}.

\eqref{eq:thmvH}, \eqref{eq:thmz}:
The estimates on the left follow from 
$(\cc,\tilde{z})\in\tildeFPspace$ and
the admissible largeness assumption $\Cmain\ge \max\{C_1,C_2\}$.
Note that 
\begin{align}\label{eq:z*est}
\smash{|z_*| 
=   |\Charge(-\tfrac12\Bdata_2(\ulinsol,\ulinsol))|
\lesssim_{\delta} \|\Bdata_2(\ulinsol,\ulinsol)\|_{\Hb^{1,\delta+1}}
\lesssim_{\delta} 
\|\ulinsol\|_{\Cb^{2,1}}
\|\ulinsol\|_{\Hb^{2,\delta}}}
\end{align}
by \ref{item:thmz*}, by Lemma \ref{lem:ChargeEst} and $\delta+1>2$, and by \eqref{eq:B2est}.
Thus the estimates on the right of \eqref{eq:thmvkz} follow from
\eqref{eq:thmvH} (left estimate), from
$\smash{\|\kerrdata(z)\|_{\Cb^{N,1}}} \le b |z_*+\tilde{z}| \le 2b|z_*|$
by \eqref{eq:thmkerr2}, 
and from an admissible largeness assumption on $\Cmain$.

Note that $\|\ulinsol\|_{\Hb^{2,\delta}}\le b\smash{\eps_0^{1/2}}$ 
by \ref{item:thmz*}, \eqref{eq:z*eps}.
Thus $\smash{\|\cdata\|_{\Cb^{N-2,1}}}
+ \smash{\|\kerrdata(z)\|_{\Cb^{N-2,1}}} 
\le
\smash{\tfrac12} \|\ulinsol\|_{\Cb^{2,1}}$
holds by \eqref{eq:sobolev}, \eqref{eq:thmvkz}
and an admissible smallness assumption on $\eps_0$.

\eqref{eq:MCu0Kv}:
We have $
\Charge(\MCinf(\kerrdata(z)))
+ \Charge(\frac12\Bdata_2(\ulinsol,\ulinsol)) = \tilde{z}$
by \ref{item:thmz*} and \eqref{eq:thmkerr4}. Thus%
\begin{subequations}
\begin{align}
\cc &\;=\; - \Gc\left(\MCinf\big(\ulinsol+\kerrdata(z)+\idata \cc\big) - \dc \cc \right)
	\label{eq:vfix1}\\
0 &\;=\; \phantom{+}\;\, \Charge\left(\MCinf\big(\ulinsol+\kerrdata(z)+\idata \cc\big)-\dc \cc\right)
	 \label{eq:zfix1}
\end{align}
\end{subequations}
by \eqref{eq:vzFP}.
Abbreviate $U = \ulinsol + \kerrdata(z) + \idata \cc$
and $E = \MCinf(U)$.
Our goal is to show that $E=0$.
Applying $\dc$ to \eqref{eq:vfix1}
and then using Proposition \ref{prop:Gxhomotopy} yields
\begin{align*}
\dc \cc 
	&= 
	- \dc \Gc(E - \dc \cc )
	=
	- (\one-\Pic-\Gc\dc)(E - \dc \cc )
	= 
	\dc \cc
	-E 
	+\Gc\dc E
\end{align*}
where the last step uses
$\Pic\left(E - \dc \cc \right)=0$ by \eqref{eq:zfix1}, \eqref{eq:Picexplicit} and $\dc^2=0$. 
Thus $E = \Gc \dc E = \Gc\Bdata_1(\idata E)$
where the second identity uses \eqref{eq:DPB1}, $\pdata\idata=\one$.
Note that $-\Bdata_1(\idata E)$ is the $m=1$ summand in \eqref{eq:MCid}
with $\udata=U$ (use $N\ge8$),
thus
\begin{align}
\label{eq:EFP}
E 
= 
\Gc \big(\tsum_{m\ge2} (-1)^{m} \frac{1}{(m-1)!} \Bdata_m(\idata E,U,\dots,U)\big)
\end{align}
By Proposition \ref{prop:uBestimates} with 
$Y = \{1,2\}$,
$\udata_1=E$, $\udata_{j\ge2}=U$,
$\delta_1 = \delta+1$, $\delta_{j\ge2}=0$:
\begin{align*}
\|\tfrac{1}{(m-1)!} \Bdata_m(\idata E,U,\dots,U)\|_{\Hb^{N-2,\delta+1}}
\lesssim_{N,\delta,b,\rfix}
|z_*|^{\frac12}\|E\|_{\Hb^{N-1,\delta+1}} m (\tfrac14)^m
\end{align*}
also using
$\|U\|_{\Cb^{4,0}} \le \rhoB$ and using
$\|U\|_{\Hb^{N-1,0}}
\lesssim_{N,\delta,b,\rfix} |z_*|^{\frac12} \le1$
by the triangle inequality, 
\ref{item:thmz*},
\eqref{eq:thmvH},
$\|\kerrdata(z)\|_{\Hb^{N-1,0}}
\lesssim_{N} \|\kerrdata(z)\|_{\Cb^{N-1,1}}\le 2b|z_*|$
by \eqref{eq:thmkerr2}.
Together with Proposition \ref{prop:Gxhomotopy} this yields
$
\|E\|_{\Hb^{N-1,\delta+1}} 
\lesssim_{N,\delta,b,\rfix}
|z_*|^{\frac12}\|E\|_{\Hb^{N-1,\delta+1}}
$.
By \eqref{eq:MCukv} the norm of $E$ is finite,
thus \eqref{eq:z*eps} and an admissible smallness assumption on $\eps_0$ 
yield $E=0$. 
This proves \eqref{eq:MCu0Kv}, and thus  
concludes the proof of existence and Part 1.

\textit{Proof of Part 2, for $\delta'=\delta$.}
We show inductively that $\cc_{\ell}\in \Hb^{N',\delta+1}(\diamonddata,\cx^1)$:
It holds for $\cc_0=0$.
If it holds for $\cc_{\ell}$ then
\smash{$\|\cc_{\ell}\|_{\Hb^{N',\delta+1}}\le b'_\ell$}
for some $b'_{\ell}\ge1$.
We now apply Lemma \ref{lem:phi} with the parameters in the 
'regularity column' of Table \ref{tab:LemmaParameters};
the assumptions of the lemma hold by
\ref{item:thmu0sol}, \ref{item:thmu0small},  \ref{item:ubound'}, 
\eqref{eq:thmkerr1}, \eqref{eq:thmkerrsol}, \ref{item:kerrbound'}.
We have $(\cc_{\ell},\tilde{z}_{\ell})\in\FPspace^{N',\delta}$,
thus the lemma implies 
$\cc_{\ell+1}
= \FP_{\cx}(\cc_{\ell},\tilde{z}_{\ell})
\in\Hb^{N',\delta+1}(\diamonddata,\cx^1)$.

For $\ss\in\{N,N+1,N+2,\dots,N'\}$
let $P_{\ss}$, $P_{\ss}'$ be the statements:
\begin{align*}
P_{\ss}: &&& 
\forall \ell\in\Z_{\ge0},
\forall \alpha\in\N_0^3\ \text{with}\ |\alpha|=\ss-N:\;\;\\
&&&		\|(\xvecjap\p_{\vec{x}})^{\alpha} \cc_{\ell}\|_{\Hb^{N,\delta+1}} 
		\lesssim_{N',\delta,b,b',\rfix} 
		\|\ulinsol\|_{\Hb^{\ss,1}}
		\|\ulinsol\|_{\Hb^{\ss,\delta}}
		\\
P_{\ss}': &&&
\begin{aligned}[t]
&\text{There exists a polynomial $p$ whose coefficients
only depend on}\\[-1mm]
&\text{$N',\delta,b,b',\rfix$, such that $\forall\ell\in\Z_{\ge1}$, 
$\forall \alpha\in\N_0^3$ with $|\alpha|=\ss-N$:}\\[-1mm]	
&\|(\xvecjap\p_{\vec{x}})^{\alpha} (\cc_{\ell+1}-\cc_{\ell})\|_{\Hb^{N,\delta+1}} 
		\le \tfrac12 
		\|(\xvecjap\p_{\vec{x}})^{\alpha} (\cc_{\ell}-\cc_{\ell-1})\|_{\Hb^{N,\delta}} 
		+ p(\ell) 2^{-\ell}
\end{aligned}
\end{align*}
We show by induction that all $P_{\ss}$, $P_{\ss}'$ hold.
We abbreviate 
$\lesssim_{N',\delta,b,b',\rfix}$ by $\lesssim_{*}$.
\textit{Induction base:}
We check $P_N$.
Using $(\cc_\ell,\tilde{z}_\ell)\in \tildeFPspace$
and then \eqref{eq:z*est}, \eqref{eq:monotone}, \eqref{eq:sobolev},
\begin{align}\label{eq:z*HH}
\|\cc_\ell\|_{\Hb^{N,\delta+1}}
\lesssim_{N,\delta,b,\rfix}
|z_*| 
\lesssim_{\delta}
\|\ulinsol\|_{\Hb^{N,1}}
\|\ulinsol\|_{\Hb^{N,\delta}}
\end{align}
We check $P_N'$. 
Using $(\cc_0,\tilde{z}_{0})=(0,0)$ and \eqref{eq:phicCdiff}
and $(\cc_1,\tilde{z}_1)\in\tildeFPspace$:
\begin{align}
\|\cc_{\ell+1}-\cc_{\ell}\|_{\Hb^{N,\delta+1}} 
+
\aa |\tilde{z}_{\ell+1}-\tilde{z}_{\ell}|
\le
2^{-\ell}(\|\cc_1\|_{\Hb^{N,\delta+1}}+\aa|\tilde{z}_1|)
\lesssim_{N,\delta,b,\rfix}
2^{-\ell}
\label{eq:dvNNEW}
\end{align}
also using $\aa |\tilde{z}_1| \lesssim_{N,\delta,b,\rfix} |z_*|\le1$.
\textit{Induction step:}
Let $\ss\in \{N+1,\dots,N'\}$ and assume that 
$P_{\ss'}, P_{\ss'}'$ hold for all $N\le\ss'\le\ss-1$.
For $\ell\in\Z_{\ge0}$ abbreviate
$$
\cdata_\ell = \idata \cc_\ell
\qquad
\udata_\ell = \ulinsol+\cdata_\ell
\qquad
z_\ell = z_* + \tilde{z}_\ell
\qquad
\kerrdata_\ell = \kerrdata(z_\ell)
$$
Recall that in the basis that we use, $\idata$ is the identity matrix.
We state some preliminary estimates.
Using $(\cc_\ell,\tilde{z}_\ell)\in\tildeFPspace$ 
and using \eqref{eq:z*HH} and \ref{item:thmu0small}:
\begin{align}
\|\udata_\ell\|_{\Cb^{4,0}} 
	&\le \tfrac23 \rhoB
&
\|\udata_{\ell}\|_{\Hb^{N,1}}
	&\lesssim_{N,\delta,b,\rfix}
	\|\ulinsol\|_{\Hb^{N,1}} 
	\le b 
	\label{eq:ulCbHb1}	
\intertext{Further
$z_\ell\in\Ccone_{\rr,\conenb{2^{-6}}}$
and $|z_\ell| \le 2|z_*|$, thus \eqref{eq:thmkerr1}
and \eqref{eq:thmkerr2}, \eqref{eq:z*HH}, \ref{item:thmu0small} imply}
\|\kerrdata_\ell\|_{\Cb^{4,0}} &\le \tfrac13 \rhoB
&
\|\kerrdata_\ell\|_{\Hb^{N,\frac12}}
&\lesssim_N
\|\kerrdata_\ell\|_{\Cb^{N,1}}
	\lesssim_{N,\delta,b,\rfix} 
	\|\ulinsol\|_{\Hb^{N,1}}
	\le b
	\label{eq:KCbHbN}
\\ 
&&\smash{\|\kerrdata_\ell\|_{\Hb^{\ss,\frac12}}}
	&\lesssim_{*}
	\|\kerrdata_\ell\|_{\Cb^{\ss,1}}
	\lesssim_{*}
	\|\ulinsol\|_{\Hb^{N,1}}\le b
\label{eq:kerrlh12NEW}
\end{align}
also using \eqref{eq:thmkerr12'} in the last line.
The induction hypothesis $P_{N},\dots,P_{\ss-1}$ implies 
\begin{align}
\|\cc_{\ell}\|_{\Hb^{\ss-1,\delta+1}} 
\lesssim_{*}
\|\ulinsol\|_{\Hb^{\ss-1,1}}
		\|\ulinsol\|_{\Hb^{\ss-1,\delta}}\label{eq:Ps-1X}
\end{align}
In particular, also using \ref{item:ubound'},
\begin{align}\label{eq:susimena_s1delta}
	\|\udata_{\ell}\|_{\Hb^{\ss-1,\delta}}
	&\lesssim_{*}
	\|\ulinsol\|_{\Hb^{\ss-1,\delta}} \le b'
	&
	\|\udata_{\ell}\|_{\Hb^{\ss-1,1}}
	&\lesssim_{*}
	\|\ulinsol\|_{\Hb^{\ss-1,1}} \le b'
\end{align}

We check $P_{\ss}$. 
Claim: Under an admissible smallness assumption on $\eps_0$,
\begin{align}\label{eq:GoalPs}
\|(\xvecjap\p_{\vec{x}})^{\alpha} \cc_{\ell+1}\|_{\Hb^{N,\delta+1}}
-
\tfrac12 \|(\xvecjap\p_{\vec{x}})^{\alpha} \cc_{\ell}\|_{\Hb^{N,\delta+1}}
\lesssim_*
\|\ulinsol\|_{\Hb^{\ss,1}}\|\ulinsol\|_{\Hb^{\ss,\delta}}
\end{align}
Proof of \eqref{eq:GoalPs}:
Decompose 
\begin{align*}
(\xvecjap\p_{\vec{x}})^{\alpha} \cc_{\ell+1}
&=
(\xvecjap\p_{\vec{x}})^{\alpha}\FP_{\cx}(\cc_{\ell},\tilde{z}_{\ell})
=
A+Y + Q
\end{align*}
where, 
using $\dc \cc_{\ell}=\Bdata_1(\udata_\ell)$,
\begin{align*}
A
&=
-\Gc(\xvecjap\p_{\vec{x}})^{\alpha}\big(\MC_{\ge2}(\udata_\ell+\kerrdata_\ell)
-\MC_{\ge2}(\kerrdata_\ell)\big) \\
Y
&=
-[(\xvecjap\p_{\vec{x}})^{\alpha},\Gc]
(\MC_{\ge2}(\udata_\ell+\kerrdata_\ell)
-\MC_{\ge2}(\kerrdata_\ell))\\
Q
&=
-(\xvecjap\p_{\vec{x}})^{\alpha}\Gc(\MCinf(\kerrdata_\ell) )
\end{align*}
We estimate the tree terms:
\begin{itemize}
\item 
$A$:
For $n\ge2$ define $A_n$ as in \eqref{eq:An},
but with $\udata,\kappa$ there replaced by 
$\udata_\ell,\kerrdata_\ell$ here.
Then by Proposition \ref{prop:Gxhomotopy},
$$\textstyle\|A\|_{\Hb^{N,\delta+1}} 
	\lesssim_{N,\delta} 
	\|(\xvecjap\p_{\vec{x}})^{\alpha}
	\sum_{n\ge2}A_n\|_{\Hb^{N-1,\delta+1}} $$
By Definition \ref{def:Bcommutator} and \ref{item:Bngrsym}
we have
\smash{$(\xvecjap\p_{\vec{x}})^{\alpha}A_n
= A_{n}^{(1)}+A_{n}^{(2)}+A_{n}^{(3)}$} where
\begin{align*}
A^{(1)}_n 
&:= 
\tsum_{j=1}^n {n \choose j} \tfrac{1}{n!}
j \Bdata_n((\xvecjap\p_{\vec{x}})^{\alpha}\udata_\ell, \udata_\ell^{\otimes j-1},\kerrdata_\ell^{\otimes n-j})\\
A^{(2)}_n 
&:=
\tsum_{j=1}^n {n \choose j} \tfrac{1}{n!}(n-j) 
\Bdata_n(
\udata_\ell,
(\xvecjap\p_{\vec{x}})^{\alpha}\kerrdata_\ell,
\udata_\ell^{\otimes j-1},
\kerrdata_\ell^{\otimes n-j-1})\\
A^{(3)}_n 
&:= 
\tsum_{j=1}^n {n \choose j} \tfrac{1}{n!}
\CBdata_{n,\alpha}(\udata_\ell^{\otimes j},\kerrdata_\ell^{\otimes n-j})
\end{align*}
using the notation \eqref{eq:tensorprnot}.
We estimate \smash{$A^{(1)}_n$}.
Using \eqref{eq:ulCbHb1}, \eqref{eq:KCbHbN}
together with \eqref{eq:uBestimates}, \eqref{eq:B2est}
for $n=2$ respectively Corollary \ref{cor:bn3} (first item)
for $n\ge3$,
\begin{align*}
\|\smash{A_{n}^{(1)}}\|_{\Hb^{N-1,\delta+1}}
&\lesssim_{N,\delta,b,\rfix}
\smash{\|\ulinsol\|_{\Hb^{N,1}}
\|(\xvecjap\p_{\vec{x}})^{\alpha} \udata_\ell\|_{\Hb^{N,\delta}} 
n(\tfrac12)^n }
\end{align*}
We use $\udata_\ell=\ulinsol+\cdata_\ell$ and the triangle inequality.
Then \smash{$\|\ulinsol\|_{\Hb^{N,1}}\lesssim_{b}\eps_0^{1/2}$}
(by \ref{item:thmz*}, \eqref{eq:z*eps})
and an admissible smallness assumption on $\eps_0$ yield
\begin{equation*}
\smash{\textstyle
\sum_{n\ge2} \|A_{n}^{(1)}\|_{\Hb^{N-1,\delta+1}}}
-\tfrac12 \|(\xvecjap\p_{\vec{x}})^{\alpha} \cc_\ell\|_{\Hb^{N,\delta}}
\le 
\|\ulinsol\|_{\Hb^{N,1}}\|\ulinsol\|_{\Hb^{\ss,\delta}}
\end{equation*}
We estimate \smash{$A^{(2)}_n$}:
Again by \eqref{eq:B2est} and Corollary \ref{cor:bn3} (first item),
\begin{align*}
\|\smash{A_{n}^{(2)}}\|_{\Hb^{N-1,\delta+1}}
	&\lesssim_{N,\delta,b,\rfix}
	\smash{\|\udata_\ell\|_{\Hb^{N,\delta}}
	\|(\xvecjap\p_{\vec{x}})^{\alpha}\kerrdata_\ell\|_{\Cb^{N,1}}
	n(\tfrac12)^n}\\
	&\lesssim_{*}
	\smash{
	\|\ulinsol\|_{\Hb^{N,\delta}}
	\|\ulinsol\|_{\Hb^{N,1}}
		n(\tfrac12)^n}
\end{align*}
using \eqref{eq:z*HH}, \ref{item:thmu0small} to bound the
$\udata_\ell$-term and $N+|\alpha|=\ss$, \eqref{eq:kerrlh12NEW}
for the $\kerrdata_\ell$-term.
Thus 
\smash{$\sum_{n\ge2}\|\smash{A_{n}^{(2)}}\|_{\Hb^{N-1,\delta+1}}$}
is bounded by the right side of \eqref{eq:GoalPs}.

We estimate \smash{$A^{(3)}_n$}:
For $n=2$ use \eqref{eq:uCBestimates} with $n=2$ and \eqref{eq:CB12est}.
For $n\ge3$ use Lemma \ref{lem:uCBestimates}
with $Y=\{1,2,3\}$,
$\delta_1=\delta,\delta_2=\frac12,\delta_3=\frac12,\delta_{j\ge4}=0$.
With the $\Cb^{4,0}$-bounds 
in \eqref{eq:ulCbHb1}, \eqref{eq:KCbHbN} and
$\rhoB\le 1/(4\CBcom)$, \eqref{eq:kerrlh12NEW},
\eqref{eq:susimena_s1delta},
this yields
\begin{align*}
\smash{\|A^{(3)}_2\|_{\Hb^{N-1,\delta+1}}}
&\lesssim_{*}
(\|\kerrdata_\ell\|_{\Cb^{\ss-1,1}}+\|\udata_\ell\|_{\Hb^{\ss-1,1}})
\|\udata_\ell\|_{\Hb^{\ss-1,\delta}}\\
\smash{\|A^{(3)}_{n\ge3}\|_{\Hb^{N-1,\delta+1}}}
&\lesssim_{*}
\smash{\|\udata_\ell\|_{\Hb^{\ss-1,\delta}} 
\|\ulinsol\|_{\Hb^{\ss-1,1}}^2 (\tfrac12)^n}
\end{align*}
With 
\eqref{eq:kerrlh12NEW},
\eqref{eq:susimena_s1delta}
this implies that 
\smash{$\sum_{n\ge2}\|\smash{A_{n}^{(3)}}\|_{\Hb^{N-1,\delta+1}}$}
is bounded by the right hand side of \eqref{eq:GoalPs}.
Collecting terms, we have shown that 
\smash{$\|A\|_{\Hb^{N,\delta+1}}
-
\tfrac12 \|(\xvecjap\p_{\vec{x}})^{\alpha} \cc_\ell\|_{\Hb^{N,\delta}}
$} is bounded by the right hand side of \eqref{eq:GoalPs}.
\item 
$Y$:
By Proposition \ref{prop:Gxhomotopy}, 
\begin{align}
\|Y\|_{\Hb^{N,\delta+1}}
&\lesssim_{*}
\|\MC_{\ge2}(\udata_\ell+\kerrdata_\ell)
- \MC_{\ge2}(\kerrdata_\ell)\|_{\Hb^{\ss-2,\delta+1}}\nonumber\\
&\lesssim_{*}
\smash{\|\udata_\ell\|_{\Hb^{\ss-1,\delta}}
(\|\udata_\ell\|_{\Hb^{\ss-1,1}} + \|\kerrdata_\ell\|_{\Cb^{\ss-1,1}})}
\label{eq:y}
\end{align}
where the second estimate holds by 
\eqref{eq:MC1} in Lemma \ref{lem:MCinf}
with $\ss$ there given by $\ss-1$ here;
the assumptions \eqref{eq:ukappaassp} hold
(for a suitable choice of $b$ there),
using 
the $\Cb^{4,0}$ estimates in \eqref{eq:ulCbHb1}, \eqref{eq:KCbHbN}
and 
\eqref{eq:kerrlh12NEW}, 
\eqref{eq:susimena_s1delta}, \eqref{eq:thmkerrsol}.
By 
\eqref{eq:kerrlh12NEW},
\eqref{eq:susimena_s1delta}
the term \eqref{eq:y} is bounded by the right side of \eqref{eq:GoalPs}.
\item 
$Q$:
Using Proposition \ref{prop:Gxhomotopy}, 
\eqref{eq:MC3} and \eqref{eq:thmkerr12'}
one obtains 
$\smash{\|Q\|_{\Hb^{N,\delta+1}}} \lesssim_{*}|z_*|$,
which is bounded by the right hand side of \eqref{eq:GoalPs}
by \eqref{eq:z*HH}.
\end{itemize}
This proves \eqref{eq:GoalPs}, which implies 
$P_{\ss}$ using $\cc_0=0$.

We check $P_{\ss}'$.
By 
\eqref{eq:kerrlh12NEW},
\eqref{eq:susimena_s1delta},
$P_{\ss}$ and \ref{item:ubound'}:
\begin{equation}\label{eq:b2}
\|\kerrdata_\ell\|_{\Cb^{\ss,1}},\ 
\|\udata_\ell\|_{\Hb^{\ss,\delta}} \lesssim_{*} 1
\end{equation}
For $\ell\ge1$ set $\Delta_\ell = \cdata_{\ell}-\cdata_{\ell-1}$.
Decompose
$(\xvecjap\p_{\vec{x}})^{\alpha}\Delta_{\ell+1} = Z+Z'+Q+Y$
with
\begin{align*}
Z&:=
-\Gc
(\xvecjap\p_{\vec{x}})^{\alpha}
\big(
F_{\ell,\ell}
-
F_{\ell-1,\ell}
\big)\\
Z'&:=
-[(\xvecjap\p_{\vec{x}})^{\alpha}, \Gc]
\big(
F_{\ell,\ell}
-
F_{\ell-1,\ell}
\big)\\
Q&:=
-(\xvecjap\p_{\vec{x}})^{\alpha}
\Gc\big(\big(
F_{\ell-1,\ell} 
- \MC_{\ge2}(\kerrdata_\ell)
\big)
-
\big(
F_{\ell-1,\ell-1}
 - \MC_{\ge2}(\kerrdata_{\ell-1})
\big)\big)\\
Y&:=
-(\xvecjap\p_{\vec{x}})^{\alpha}
\Gc\big(\MCinf(\kerrdata_\ell)-\MCinf(\kerrdata_{\ell-1})\big)
\end{align*}
where we abbreviate 
$F_{\ell,\ell'} = \MC_{\ge2}(\udata_\ell+\kerrdata_{\ell'})$.
We estimate the four terms:
\begin{itemize}
\item 
$Z$: 
For $n\ge2$ let $Z_{n}^{(1)}$ be the linear operator \eqref{eq:Z1n}
with $\udata,\udata',\kappa$ there replaced by 
$\udata_\ell,\udata_{\ell-1},\kerrdata_\ell$ here.
Then, using Proposition \ref{prop:Gxhomotopy}, 
\begin{align*}
\textstyle
\|Z\|_{\Hb^{N,\delta+1}}
\lesssim_{N,\delta}
\|(\xvecjap\p_{\vec{x}})^{\alpha}
\sum_{n\ge2} Z_n^{(1)}(\Delta_{\ell})
\|_{\Hb^{N-1,\delta+1}}
\end{align*}
By Definition \ref{def:Bcommutator}, \ref{item:Bngrsym}:
$(\xvecjap\p_{\vec{x}})^{\alpha} Z_n^{(1)}(\Delta_{\ell})
=Z_n^{(1)}((\xvecjap\p_{\vec{x}})^{\alpha}\Delta_{\ell})
+ \tsum_{i=2}^5 Z_{n}^{(i)}$,
\begin{align*}
	Z_n^{(2)}
	&:=
	\textstyle
	\sum_{i=1}^n \sum_{a=0}^{i-1} {n\choose i} 
	\tfrac{a}{n!}	\Bdata_n\big(
		\Delta_\ell,
		(\xvecjap\p_{\vec{x}})^{\alpha}\udata_\ell,
		\udata_\ell^{\otimes a-1},
		\udata_{\ell-1}^{\otimes i-1-a},
		\kerrdata_\ell^{\otimes n-i}\big)\\
	Z_n^{(3)}
	&:=
	\textstyle
	\sum_{i=1}^n \sum_{a=0}^{i-1} {n\choose i} 				
	\tfrac{i-1-a}{n!}\Bdata_n\big(
		\Delta_\ell,
		(\xvecjap\p_{\vec{x}})^{\alpha}\udata_{\ell-1},
		\udata_\ell^{\otimes a},
		\udata_{\ell-1}^{\otimes i-2-a},
		\kerrdata_\ell^{\otimes n-i}\big)\\
	Z_n^{(4)}
	&:=
	\textstyle
	\sum_{i=1}^n \sum_{a=0}^{i-1} {n\choose i} 				
	\tfrac{n-i}{n!}\Bdata_n\big(
		\Delta_\ell,
		(\xvecjap\p_{\vec{x}})^{\alpha}\kerrdata_\ell,
		\udata_\ell^{\otimes a},
		\udata_{\ell-1}^{\otimes i-1-a},
		\kerrdata_\ell^{\otimes n-i-1}\big)\\
	Z_n^{(5)}
	&:=
	\textstyle
	\sum_{i=1}^n \sum_{a=0}^{i-1} {n\choose i} 
	\tfrac{1}{n!}\CBdata_{n,\alpha}
	\big(
	\Delta_{\ell},
		\udata_\ell^{\otimes a},
		\udata_{\ell-1}^{\otimes i-1-a},
		\kerrdata_\ell^{\otimes n-i}\big)
\end{align*}
By \eqref{eq:Z12n} and \eqref{eq:ulCbHb1}, \eqref{eq:KCbHbN}:
\[ 
\textstyle\sum_{n\ge2}
\|Z_n^{(1)}((\xvecjap\p_{\vec{x}})^{\alpha}\Delta_{\ell})\|_{\Hb^{N-1,\delta+1}}
\lesssim_{N,\delta,b}
\|\ulinsol\|_{\Hb^{N,1}}
\|(\xvecjap\p_{\vec{x}})^{\alpha}\Delta_\ell\|_{\Hb^{N,\delta}}
\]
Using \smash{$\|\ulinsol\|_{\Hb^{N,1}}\lesssim_b\eps_0^{1/2}$},
an admissible smallness assumption on $\eps_0$ yields
\[ 
\textstyle
\sum_{n\ge2}
\|Z_n^{(1)}((\xvecjap\p_{\vec{x}})^{\alpha}\Delta_{\ell})\|_{\Hb^{N-1,\delta+1}}
\le
\frac12
\|(\xvecjap\p_{\vec{x}})^{\alpha}\Delta_\ell\|_{\Hb^{N,\delta}}
\]
Using \eqref{eq:uBestimates} with $n=2$, \eqref{eq:B2est}
and Corollary \ref{cor:bn3} (first item), one obtains
\begin{align*}
\tsum_{n\ge2}\smash{\|Z_{n}^{(2)}\|_{\Hb^{N-1,\delta+1}}}
	&\lesssim_{*}
	\|\Delta_\ell\|_{\Hb^{N,\delta}}
	\|u_\ell\|_{\Hb^{\ss,1}}\\
\tsum_{n\ge2}\|Z_{n}^{(3)}\|_{\Hb^{N-1,\delta+1}}
	&\lesssim_{*}
	\|\Delta_\ell\|_{\Hb^{N,\delta}}
	\|u_{\ell-1}\|_{\Hb^{\ss,1}}\\
\tsum_{n\ge2}\|Z_{n}^{(4)}\|_{\Hb^{N-1,\delta+1}}
	&\lesssim_{*}
	\|\Delta_\ell\|_{\Hb^{N,\delta}}
	\|\kerrdata_{\ell}\|_{\Cb^{\ss,1}}
\end{align*}
By \eqref{eq:dvNNEW} and \eqref{eq:b2},
each of these terms is bounded by 
\smash{$\lesssim_{*}2^{-\ell}$}.
By Lemma \ref{lem:uCBestimates}
and using the $\Cb^{4,0}$-bounds in \eqref{eq:ulCbHb1}, \eqref{eq:KCbHbN} 
and $\rhoB\le 1/(4\CBcom)$
and \eqref{eq:b2},
\begin{align*}
\tsum_{n\ge2}\|Z_{n}^{(5)}\|_{\Hb^{N-1,\delta+1}}
	&\lesssim_{*}
	\|\Delta_\ell\|_{\Hb^{\ss-1,\delta}}
\end{align*}
By the induction hypothesis $P_{N}',\dots,P_{\ss-1}'$, 
\begin{align}\label{eq:EstQ5_2}
\|\Delta_\ell\|_{\Hb^{\ss-1,\delta}}
\le 2^{-\ell} \big(\|\Delta_1\|_{\Hb^{\ss-1,\delta}} + p_0(\ell)\big)
\lesssim_{*}
2^{-\ell} \big(1 + p_0(\ell)\big)
\end{align}
with $p_0$ a polynomial whose coefficients
only depend on $N',\delta,b,b',\rfix$.
The last estimate uses $\Delta_1=\cdata_1$
and \eqref{eq:Ps-1X}, \ref{item:ubound'}.
Collecting terms, we obtain
\begin{align*}
\|Z\|_{\Hb^{N,\delta+1}}
\le 
\tfrac12\|(\xvecjap\p_{\vec{x}})^{\alpha}\Delta_\ell\|_{\Hb^{N,\delta}}
+2^{-\ell} p_1(\ell)
\end{align*} 
with $p_1$ a polynomial whose coefficients
only depend on $N',\delta,b,b',\rfix$.
\item 
$Z'$:
Using Proposition \ref{prop:Gxhomotopy},
\begin{align*}
\|Z'\|_{\Hb^{N,\delta+1}}
&\lesssim_{*}
\|\MC_{\ge2}(\udata_\ell + \kerrdata_\ell) -
\MC_{\ge2}(\udata_{\ell-1} + \kerrdata_{\ell})
\|_{\Hb^{\ss-2,\delta+1}}
\end{align*}
Now use \eqref{eq:MC5} in Lemma \ref{lem:MCinf}
with $\kappa=\kappa'$ and $\ss$ there given by $\ss-1$ here;
the assumptions \eqref{eq:ukappaassp} hold
(for a suitable choice of $b$ there)
using the $\Cb^{4,0}$-bounds in \eqref{eq:ulCbHb1}, \eqref{eq:KCbHbN}
and using \eqref{eq:thmkerrsol}, \eqref{eq:b2}.
This yields 
\begin{align*}
\|Z'\|_{\Hb^{N,\delta+1}}
&\lesssim_{*}
\|\Delta_{\ell}\|_{\Hb^{\ss-1,\delta}}
\lesssim_{*}
2^{-\ell} \big(1 + p_0(\ell)\big)
\end{align*}
where we use \eqref{eq:b2} to bound $\RFP_{0}$ in \eqref{eq:MC5}, 
and use \eqref{eq:EstQ5_2} in the last step.
\item 
$Q$, $Y$:
Using Proposition \ref{prop:Gxhomotopy}
and then \eqref{eq:MC5}, \eqref{eq:MC4}
in Lemma \ref{lem:MCinf} with $\ss$ there given by $\ss$ here,
one obtains
\begin{align*}
\|Q\|_{\Hb^{N,\delta+1}}+\|Y\|_{\Hb^{N,\delta+1}}
&\lesssim_{*}
(\|\udata_{\ell-1}\|_{\Hb^{\ss,\delta}}+1)
\|\kerrdata_{\ell}-\kerrdata_{\ell-1}\|_{\Cb^{\ss,1}}
\\
&\lesssim_{*}
|\tilde{z}_\ell-\tilde{z}_{\ell-1}|
\lesssim_{*}2^{-\ell}
\end{align*}
where for the second estimate we use 
\eqref{eq:b2} and \eqref{eq:thmkerr12'}, 
and for the last estimate we use \eqref{eq:dvNNEW}
and $\aa^{-1} = 4 \CFPapply b |z_*|^{\frac12} \lesssim_{N,\delta,b,\rfix} 1$.
\end{itemize}
Collecting terms we obtain $P_{\ss}'$,
which completes the induction step.

We conclude Part 2 with $\delta'=\delta$.
The statements $P_{N}',\dots,P_{N'}'$ imply
\[ 
\|\cc_{\ell+1}-\cc_{\ell}\|_{\Hb^{N',\delta+1}}
\le 
2^{-\ell}\big(\|\cc_1\|_{\Hb^{N',\delta}}+p_2(\ell)\big)
\]
with $p_2$ a polynomial whose coefficients only depend on $N',\delta,b,b',\rfix$.
Thus the sequence $\cc_{\ell}$ is Cauchy 
and converges in $\Hb^{N',\delta+1}$,
by uniqueness of the limit (in $\Hb^{N,\delta+1}$)
it is equal to $\cc$.
Then $P_N,\dots,P_{N'}$ imply the bound in the theorem.

%

\textit{Proof of Part 2 for general $\delta'$.}
Fix $\slashed{\delta}\in[\delta,\delta+1)$ 
such that $\delta'-\slashed\delta\in\Z_{\ge0}$.
For $\delta_0\in\{\slashed\delta,\slashed\delta+1,\dots,\delta'\}$
let $P_{\delta_0}$ be the statement
\begin{align*}
P_{\delta_0}:
\qquad
\|\cc\|_{\Hb^{N',\delta_0+1}}
	\lesssim_{N',\delta',b,b',\rfix} 
	\|\ulinsol\|_{\Hb^{N',1}}
	\|\ulinsol\|_{\Hb^{N',\delta_0}}
\end{align*}
We show by induction that all $P_{\delta_0}$ hold, 
$P_{\delta'}$ implies the statement in the theorem.

\textit{Induction base:}
By Part 2 with $\delta'=\delta$,
and by $\delta\le \slashed\delta\le\delta+1$ and \eqref{eq:monotone},
\begin{align} \label{eq:cn'sldelta}
\|\cc\|_{\Hb^{N',\slashed\delta}} 
\lesssim_{N',\delta,b,b',\rfix}
\|\ulinsol\|_{\Hb^{N',1}} 
\|\ulinsol\|_{\Hb^{N',\slashed\delta}} 
\end{align}
We use Lemma \ref{lem:phi} with the parameters 
specified in the 'decay, base' column of 
Table \ref{tab:LemmaParameters},
where $b''\ge1$ only depends on $N',\delta,b,b',\rfix$
and is chosen such that $\smash{\|\cc\|_{\Hb^{N',\slashed\delta}}}\le b''$;
this exists by \eqref{eq:cn'sldelta}, \ref{item:ubound'}.
In particular $(\cc,\tilde{z}) \in \smash{\FPspace^{N',\slashed\delta}}$.
Further the assumption \ref{item:thmu0small} of the lemma
is satisfied by \ref{item:ubound'} and $\slashed\delta\le\delta'$.
Now $P_{\slashed\delta}$ follows from \eqref{eq:vzFP} and \eqref{eq:phicest},
where we use \eqref{eq:z*HH} to bound $|z_*|$,
and  \eqref{eq:cn'sldelta}, \ref{item:ubound'} to bound
the $\cc$ terms on the right hand side of \eqref{eq:phicest}.

\textit{Induction step.}
Let $\delta_0\in\{\slashed\delta+1,\slashed\delta+2,\dots,\delta'\}$
and assume that $P_{\delta_0-1}$ holds.
To show $P_{\delta_0}$
we use Lemma \ref{lem:phi} with parameters 
specified in the 'decay, step' column of Table \ref{tab:LemmaParameters},
where $b'''\ge1$ only depends on $N',\delta,b,b',\rfix$
and is chosen such that $\smash{\|\cc\|_{\Hb^{N',\delta_0}}} \le b'''$;
this exists by $P_{\delta_0-1}$, \ref{item:ubound'}.
In particular $(\cc,\tilde{z})\in\FPspace^{N',\delta_0}$.
The assumption \ref{item:thmu0small} of the lemma
is satisfied by \ref{item:ubound'} and $\delta_0\le\delta'$.
Now $P_{\delta_0}$ follows from \eqref{eq:vzFP},
\eqref{eq:phicest}, $P_{\delta_0-1}$,
\eqref{eq:z*HH} and \ref{item:ubound'}.

\textit{Proof of Part 3.}
This follows from \ref{item:support} and \eqref{eq:thmkerrsol},
and using $\rfix\ge1$.
\qed
\end{proof}
\subsection{Proof of Theorem \ref{thm:P=0} and of Corollary \ref{cor:ApplicationMinkowskiPaper}}
	\label{sec:ProofThm1}

\begin{proof}[of Theorem \ref{thm:P=0}]
We use Theorem \ref{thm:mainL-infinity} with 
the parameters in Table \ref{tab:TheoremParameters}.
Let $\Capply{\Cmain},\epsapply{\eps}_0$ be the constants produced
by that theorem (called $\Cmain,\eps_0$ there).
They depend only on \eqref{eq:P=0_parameters}, $\KSconst_{N}$, in particular
$\eps,C$ are allowed to depend on $\Capply{\Cmain},\epsapply{\eps}_0$.
Set 
\begin{align*}
C = \Capply{\Cmain}
\qquad
\eps 
=\min\{
M,\ \epsapply{\eps}_0,\ b,\ (\Csob{4,0})^{-1}\tfrac{\rhoB}{3},\ \tfrac{\rhoB}{6\KSconst_4}
\}
\end{align*}
We check that the assumptions of Theorem \ref{thm:mainL-infinity} hold:
As required $\eps\in(0,\epsapply{\eps}_0]$ and 
$\ulinsol,\kerrdata$ are as in \eqref{eq:u0Kmain}.
\ref{item:thmu0sol}:
By \ref{item:MAINDPu0=0_MAINu0small} and \eqref{eq:DPB1}.
\ref{item:thmu0small}:
By \ref{item:MAINDPu0=0_MAINu0small},
\eqref{eq:sobolev} and $N\ge6$, and the choice of $\eps$.
\ref{item:thmz*}:
By Remark \ref{rem:B2charge}, 
the definition of $z_*$ in Theorem \ref{thm:P=0}
and in Theorem \ref{thm:mainL-infinity} agree,
thus \ref{item:thmz*} follows from \ref{item:MAINu0cone}.
\ref{item:thmkerr}: 
Let $z_1\in\Ccone_{\eps,\conenb{2^{-6}}}$.
\eqref{eq:thmkerr1}:
By \eqref{eq:KSest12} with $\ss=4$, \eqref{eq:monotone}, $|z_1|\le2\eps$ 
and the choice of $\eps$.
\eqref{eq:thmkerr2}, \eqref{eq:thmkerr3}: 
By \eqref{eq:KSest12} with $\ss=N$.
\eqref{eq:thmkerrsol}: 
By \eqref{eq:thmkerr1}, \eqref{eq:thmkerr2} and $N\ge8$ we have
$\KSdata(z_1)\in\Wconv^{8}$, thus \eqref{eq:thmkerrsol}
follows from \ref{item:IntroKSConstraints}, \eqref{eq:equivMCP}.
\eqref{eq:thmkerr4}: 
By \ref{item:IntroKSCharges}. 
Thus the assumptions hold.

\begin{table}
\centering
\footnotesize{
\begin{tabular}{cc|c}
&
	\begin{tabular}{@{}c@{}}
	Parameters \\ in Theorem \ref{thm:mainL-infinity}
	\end{tabular}
	&
	\begin{tabular}{@{}c@{}}
	Parameters \\ used to invoke Theorem \ref{thm:mainL-infinity}
	\end{tabular}
	\\
\hline
Input
&$N$, $\delta$, $b$, $\rfix$
	& $N$, $\delta$, $\max\{b,\KSconst_N\}$, $101$ \\
&$\eps$	 
	& $\eps$
	\\
&$\ulinsol$, $\kerrdata$
	& $\ulinsol$, $\KSdata$ restricted to $\Ccone_{\eps,\conenb{2^{-6}}}$\\
	&$N'$, $\delta'$, $b'$ (Part 2 only)
	& $N'$, $\delta'$, $\max\{b',\KSconst_{N'}\}$\\
\hline
Output
	&$C$, $\eps_0$, $\cdata$, $z$
	& $\Capply{\Cmain}$, $\epsapply{\eps}_0$, $\cdata$, $z$
\end{tabular}}
\captionsetup{width=115mm}
\caption{%
The first column lists the input and output 
parameters of Theorem \ref{thm:mainL-infinity}. 
The second column specifies the choice 
of the input parameters when invoking
Theorem \ref{thm:mainL-infinity}
in the proof of Theorem \ref{thm:P=0}.
The constants
produced by this invocation of
Theorem \ref{thm:mainL-infinity}
are denoted $\Capply{\Cmain}$, $\epsapply{\eps}_0$ 
and depend only on the parameters in the first row.}
\label{tab:TheoremParameters}
\end{table}

Let $\cdata$, $z$ be as in Theorem \ref{thm:mainL-infinity}.
We check that they have the properties stated in Theorem \ref{thm:P=0}:
\eqref{eq:Pu0Kv} holds by \eqref{eq:MCu0Kv}, \eqref{eq:equivMCP} and $N\ge8$,
the estimate thereafter holds by that after \eqref{eq:MCu0Kv}.
Part 1 follows from Part 1 of Theorem \ref{thm:mainL-infinity}
and $C = \Capply{\Cmain}$.
Part 2 follows from Part 2 of Theorem \ref{thm:mainL-infinity}
using the parameters in Table \ref{tab:TheoremParameters},
where the assumption \ref{item:ubound'} holds by 
the assumption of Part 2 in Theorem \ref{thm:P=0},
and \ref{item:kerrbound'} holds by \eqref{eq:KSest12} with $\ss=N'$.
Part 3 follows from Part 3 of Theorem \ref{thm:mainL-infinity}.
\qed
\end{proof}
\begin{proof}[of Corollary \ref{cor:ApplicationMinkowskiPaper}]
We will make finitely many smallness assumptions on $\lambda_0$
depending only on $N,\gamma,\MinkThmsfix,b,\vlinsol$.
\textit{First item:}
By Remark \ref{rem:v0u0exist}.
\textit{Second item:}
(d1): by \ref{item:KSMC} and the assumption $\MinkThmsfix\le 1/101$.
(d2), (d3), (d4): 
This follows from \eqref{eq:KerrDecay} with $\ss=N+3$,
$
|z| 
\le|z_*| + C|z_*|^{\frac32} 
= \lambda^2(|z_*^v| + \lambda C|z_*^v|^{\frac32} )
$
by \eqref{eq:K(z)estimqte}, 
and a smallness assumption on $\lambda_0$.
(d5): By \eqref{eq:Pu0Kv}.
For (d6), (d7), (d8) note that
\begin{subequations}\label{eq:uuu}
\begin{align}
\|\udata\|_{\widetilde{C}_b^{N+3}(\Dspdata_{\le\MinkThmsfix})} 
&\lesssim_{N,\MinkThmsfix}
\|\udata\|_{\Cb^{N+4,0}(\diamonddata)} 
\label{eq:usp}
\\
\|\udata\|_{\widetilde{H}^{N+1}(\diamond_{0,\MinkThmsfix})} 
&\lesssim_{N,\MinkThmsfix}
\|\udata\|_{\Hb^{N+2,0}(\diamonddata)} 
\lesssim_{N}
\|\udata\|_{\Cb^{N+2,1}(\diamonddata)} 
\label{eq:ubulk}
\\
\smash{\|\udata-\KSdata(z)\|_{\HdataApp^{\frac52+\gamma,N+3}(\Dspdata_{\le\MinkThmsfix})}}
&\lesssim_{N,\gamma,\MinkThmsfix}
\|\udata-\KSdata(z)\|_{\Hb^{N+4,\delta_N}(\diamonddata)}
\label{eq:udecay}
\end{align}
\end{subequations}
where we decorated norms from \cite{MinkowskiPaper}
with a tilde, in order to distinguish them from the norms
defined in this paper.
In \eqref{eq:usp}, \eqref{eq:ubulk}
the shift in the number of derivatives is due to the
fact that in \cite{MinkowskiPaper} the norms are 
defined using the same number of derivatives for $\lxdata$ and for $\Idata$,
unlike \eqref{eq:normsoffset}.
Further, the norms \smash{$\tilde{C}^{N+3}_b(\Dspdata_{\le\MinkThmsfix})$}
and $\HdataApp^{\frac52+\gamma,N+3}(\Dspdata_{\le\MinkThmsfix})$ are defined using
homogeneous bases near spacelike infinity
\cite[Definition 17]{MinkowskiPaper}, 
while the bases used here are only homogeneous to leading order,
see Lemma \ref{lem:basisgdata}, but this only contributes
a constant in the estimates \eqref{eq:usp}, \eqref{eq:udecay}.
Given \eqref{eq:uuu}, 
the assumptions (d6), (d7), (d8) follow from 
the estimate after \eqref{eq:Pu0Kv} and
\eqref{eq:T1cKconcl},
\eqref{eq:sobolev}
and a smallness assumption on $\lambda_0$.

We check the last statement in the corollary.
By Part 2 of Theorem \ref{thm:P=0}
with $N'$, $\delta'$, $b'$ there given by 
$N'+6$, $\delta_{N'}$, $\|\ulinsol\|_{\Hb^{N'+6,\smash{\delta_{N'}}}}+1$ here,
\eqref{eq:Part2c} holds with $N'$, $\delta'$ replaced by $N'+6$, $\delta_{N'}$.
We conclude that the quantities in (d9)-(d12) are bounded,
which implies existence of $b'$ for which (d9)-(d12) hold.
(d9): By \eqref{eq:KerrDecay} with $\ss=N'+3$.
For (d10), (d11), (d12) note that \eqref{eq:uuu} also holds
with $N$ replaced by $N'$. 
Then (d10), (d11) follow from \eqref{eq:KerrDecay} with $\ss=N'+4$,
\smash{$\|\ulinsol\|_{\Hb^{N'+6,\smash{\delta_{N'}}}}<\infty$},
\eqref{eq:Part2c}, \eqref{eq:sobolev};
(d12) follows from 
\smash{$\|\ulinsol\|_{\Hb^{N'+6,\smash{\delta_{N'}}}}<\infty$}, \eqref{eq:Part2c}.
\qed
\end{proof}

\appendix

\section{Kerr-Schild spacetime}
\label{sec:KSreparametrization}

We construct the Kerr-Schild elements \eqref{eq:KS}.

As in \cite{Thesis,MinkowskiPaper}, 
in this section we will view $\diamond=\R^4$
as a subset of the Einstein cylinder.
Recall that $y$ are coordinates near spacelike infinity
defined by Kelvin inversion $y=x/(-(x^0)^2+|\vec{x}|^2)$.
Let $U\subset\diamond$ be the neighborhood of spacelike infinity
given by $\{|y|\le\frac{1}{100}\}$,
and $\tilde{U}\subset\diamond$ the neighborhood given by $\{|y|\le\frac{1}{101}\}$.

For $m\in[0,2^{-10}]$ and $q\in\R^9$ with $|q|\le\frac{1}{10}$ let 
\begin{align*}
\kerr(m,q)\in \gx^1(U)
\qquad\qquad
q = (\vec{b},\vec{t},\vec{a})
\end{align*}
be the family of Kerr-elements in \cite[Theorem 20]{Thesis}.
Fix a cutoff function $\psi:\diamond\to[0,1]$ that
is smooth on $\overline{\diamond}$, equal to one on $\tilde{U}$,
equal to zero on the subset of $U$ where $|y|\ge \frac{1}{100+1/2}$,
and equal to zero on $\diamond\setminus U$.
Define
\[ 
\tilde{\kerr}(m,q) = \psi\kerr(m,q) \in \gx^1(\diamond)
\qquad
\tilde{\kerrdata}(m,q) =\tilde{\kerr}(m,q) |_{\diamonddata} \in \gxdata^1(\diamonddata)
\]
with the understanding that $\tilde{\kerr}(m,q)=0$ on $\diamond\setminus U$.
Then:
\begin{enumerate}[label=\textnormal{({f\arabic*})}]
\item \label{item:KthesisMC}
$\dg\tilde\kerr(m,q)+\frac12[\tilde\kerr(m,q),\tilde\kerr(m,q)]=0$
on $\tilde{U}$
\item \label{item:Kthesisnullinf}
$\tilde\kerr(m,q)$ extends smoothly to future and past null infinity.
\item \label{item:Kthesisest}
For every $\ss\in\Z_{\ge0}$ one has,
using notation analogous to \eqref{eq:KerrDecay},
\begin{align}\label{eq:KerrDecayAp}
|(|y|\p_y)^{\le\ss}\tilde\kerr(m,q)|
&\lesssim_{\ss}
m |y|(1+|\log|y||)
\qquad
\text{on $\tilde{U}$}
\end{align}

\item \label{item:Kdata}
$\tilde\kerrdata(m,q)$ is linear in $m$,
one has 
\begin{equation}\label{eq:PDP}
\Pconstraints(\tilde{\kerrdata}(m,q)) = 0
\qquad
D\Pconstraints(\tilde{\kerrdata}(m,q)) = 0
\qquad
\text{on $\tilde{U}\cap\diamonddata$}
\end{equation}
and for all $\ss\in\Z_{\ge1}$:
\begin{subequations}\label{eq:Uest}
\begin{align}
\|\tilde\kerrdata(m,q)\|_{\Cb^{\ss,1}(\diamonddata)}
	&\lesssim_{\ss}
	m
	\label{eq:k1qbound}\\
\|\tilde\kerrdata(m,q)-\tilde\kerrdata(m',q'))\|_{\Cb^{\ss,1}(\diamonddata)}
&\lesssim_{\ss}
|m-m'| + (m+m')|q-q'|
\label{eq:k1qlip}
\end{align}
\end{subequations}

\item \label{item:psicut}
Note that $\Bdata_1(\tilde{\kerrdata}(m,q))$
vanishes on $\tilde{U}\cap\diamonddata$, by 
the second identity in \eqref{eq:PDP} and by \eqref{eq:DPB1}.
Thus we can apply $\Charge$ to it.
One has, with $q=(\vec{b},\vec{t},\vec{a})$:
\[ 
\Charge\big(\Bdata_1(\tilde{\kerrdata}(m,q))\big)
=
4\pi m 
\left(\begin{smallmatrix}
\cosh(|\vec{b}|)\\
\sinh(|\vec{b}|)\frac{\vec{b}}{|\vec{b}|} \\
\cosh(|\vec{b}|)\vec{t}
+
\sinh(|\vec{b}|) (\vec{a}\times\frac{\vec{b}}{|\vec{b}|}) \\
\cosh(|\vec{b}|)\left( 
\vec{a}-(\vec{a}\cdot\frac{\vec{b}}{|\vec{b}|})\frac{\vec{b}}{|\vec{b}|} \right)
+(\vec{a} \cdot \frac{\vec{b}}{|\vec{b}|})\frac{\vec{b}}{|\vec{b}|} 
-\sinh(|\vec{b}|) (\vec{t}\times\frac{\vec{b}}{|\vec{b}|})
\end{smallmatrix}\right)
\]
\end{enumerate}

We check these properties:
\ref{item:KthesisMC}, \ref{item:Kthesisnullinf}, \ref{item:Kthesisest}
hold by \cite[Theorem 20, Lemma 113, Remark 59]{Thesis}\footnote{
The estimate \eqref{eq:KerrDecayAp} is there shown for $\vec{b},\vec{t}=0$,
but is easily seen to extend to nonzero $\vec{b},\vec{t}$.}.
To check \ref{item:Kdata}, we use the fact that 
on $U\cap\{|y^0| \le 2^{-4} |\vec{y}|\}$ the element 
$\kerr(m,q)$ is given by \cite[(447), (454) and the '+' case in (427)]{Thesis},
where one must also use the fact that the rotations, boost, translations
map 
$U\cap\{|y^0| \le 2^{-4} |\vec{y}|\}$
to 
$|y^0|\le \frac13|\vec{y}|$.
Then by \cite[Lemma 104]{Thesis}, 
on $U\cap\{|y^0| \le 2^{-4} |\vec{y}|\}$
the element $\kerr(m,q)$ is linear in $m$,
thus $\tilde{\kerrdata}(m,q)$ is linear in $m$,
and then \eqref{eq:PDP} follows from \ref{item:KthesisMC}.
Further \cite[Remark 56]{Thesis} implies that
for all $\lambda\in(0,1]$:
\begin{align}\label{eq:KSscaling}
S_{\lambda}^* \kerr(\lambda m, \vec{b},\lambda \vec{t},\lambda \vec{a})
=
\kerr(m, \vec{b},\vec{t},\vec{a})
\quad
\text{on}
\quad 
S_{1/\lambda}(U)\cap\{|y^0| \le 2^{-4} |\vec{y}|\}
\end{align}
using the $\R_+$-action in \cite[Definition 7]{MinkowskiPaper}\footnote{%
The $\R_+$-action used in \cite{Thesis} differs from 
that in \cite{MinkowskiPaper}. 
In \cite{Thesis} the action is defined purely by pullback,
while in \cite{MinkowskiPaper} the action also involves an 
explicit $\lambda^{-1}$-factor, see \cite[(58)]{MinkowskiPaper}.}.
By \eqref{eq:KSscaling}, linearity in $m$, and the fact that
the components of $\kerr(m,q)$ are smooth jointly in $(y,m,q)$,
one obtains\footnote{
A similar, more detailed, scaling argument is
e.g.~in \cite[Proof of Lemma 113, first item]{Thesis}.}
\begin{align}
|(|y|\p_y)^{\le\ss}\kerr(m,q)|
&\lesssim_{\ss}
m|y| \label{eq:ydataest}
\\
|(|y|\p_y)^{\le\ss}(\kerr(m,q)-\kerr(m,q'))|
&\lesssim_{\ss}
m|y|(|\vec{b}-\vec{b}'|
+
|y||\vec{t}-\vec{t}'|+|y||\vec{a}-\vec{a}'|)\nonumber
\end{align}
on $U\cap\{|y^0| \le 2^{-4} |\vec{y}|\}$, 
where 
$q=(\vec{b},\vec{t},\vec{a})$, 
$q'=(\vec{b}',\vec{t}',\vec{a}')$,
and where we use notation analogous to \eqref{eq:KerrDecay}.
From this, the estimates \eqref{eq:Uest} easily follow
(the basis and vector fields used in \eqref{eq:ydataest} 
are homogeneous of degree zero, see \cite[Definition 17]{MinkowskiPaper},
while the basis and vector fields used in the norms
in \eqref{eq:Uest} are only homogeneous to leading order,
but this change only contributes a constant in the estimates).
Further \ref{item:psicut} follows from \cite[Lemma 114, Corollary 94]{Thesis}.
%
%
%
%
%
%
%
%
\begin{prop}\label{prop:Kerr}
There exists $M\in(0,2^{-11}]$ and a map
\[ 
\phi:\; \Ccone_{M,2^{-6}} 
	\;\to\; \big\{(m,q)\in[0,2^{-10}]\times\R^{9}
	\mid  |q|\le\tfrac{1}{10}  \big\}
\]
such that $\KS=\tilde{\kerr}\circ \phi$ satisfies the properties
\ref{item:KSMC}-\ref{item:IntroKSCharges} in the introduction.
\end{prop}
The reparametrization $\phi$ is constructed in the next two lemmas.
\begin{lemma}\label{lem:kerrphi0}
Define $\phi_0:\Ccone_{2^{-10},2^{-5}} 
	\to \{(m,q)\in [0,2^{-10}]\times\R^{9}\mid |q|\le\frac{1}{10}\}$,
\[ 
\phi_0(m',q')
\;=\;
\left(\begin{smallmatrix}
\tfrac{1}{4\pi}\sqrt{(m')^2-|\vec{b}'|^2}\\
\arctanh(\frac{|\vec{b}'|}{m'})\frac{\vec{b}'}{|\vec{b}'|}\\
\tfrac{1}{(m')^2-|\vec{b}'|^2}
\big( m' \vec{t}' - 
\vec{a}'\times\vec{b}'-
\tfrac{1}{m'}(\vec{b}'\cdot\vec{t}')\vec{b}' \big)\\
\tfrac{1}{(m')^2-|\vec{b}'|^2}
\Big(
m'\Big(\vec{a}' 
-
\big(1-\tfrac{\sqrt{(m')^2-|\vec{b}'|^2} }{m'}\big)
(\vec{a}'\cdot\frac{\vec{b}'}{|\vec{b}'|})\frac{\vec{b}'}{|\vec{b}'|}
\Big)
-
\vec{b}' \times\vec{t}'
\Big)
\end{smallmatrix}\right)
\]
where $q' = (\vec{b}',\vec{t}',\vec{a}')$.
This satisfies, for all $z'_1,z'_2\in\Ccone_{2^{-10},2^{-5}}$:
\begin{subequations}
\begin{align}
\Charge\big( 
\Bdata_1\big(\tilde{\kerrdata}(\phi_0 (z_1'))\big)\big)
	&=z_1' \label{eq:phi0equation}\\
|\phi_0^m(z_1')|
	&\lesssim |z_1'|
	\label{eq:phi0_bd}\\
|\phi_0^m(z_1')-\phi_0^m(z_2')| 
+
(\phi_0^m(z_1')+\phi_0^m(z_2'))
|\phi_0^q(z_1')-\phi_0^q(z_2')|
	&\lesssim
	|z_1'-z_2'| \label{eq:phi0_Lip}
\end{align}
\end{subequations}
where $\phi_0^m$ and $\phi_0^q$ denote the first
respectively the last nine components of $\phi_0$.
\end{lemma}
\begin{proof}
By direct inspection, using \ref{item:psicut}.\qed
\end{proof}
\begin{lemma}\label{lem:kerrphi1}
There exists $M\in (0,2^{-11}]$ and a map
$\phi_1:\Ccone_{M,2^{-6}}\to\Ccone_{2M,2^{-5}}$
such that for all $z_1,z_2\in\Ccone_{M,2^{-6}}$: 
\begin{subequations}
\begin{align}
&\Charge \big(\MCinf\big(\tilde{\kerrdata}\big((\phi_0\circ\phi_1)(z_1)\big)\big)\big) 
	= z_1
	\label{eq:CMCinfeq}\\
&|z_1-\phi_1(z_1)|
	\le 2^{-7}|z_1| 
	\label{eq:phi1_bd}\\
%
&\tfrac23 |z_1-z_2|
	\le |\phi_1(z_1)-\phi_1(z_2)| 
	\le 2 |z_1-z_2|
	\label{eq:phi1_lip}
\end{align}
\end{subequations}
\end{lemma}
\begin{proof}
Abbreviate $\kappa = \tilde{\kerrdata}\circ\phi_0$.
By \eqref{eq:k1qbound} and \eqref{eq:phi0_bd}
we can choose $M'\in[0,2^{-10}]$ such that 
for all $z'\in\Ccone_{M',2^{-5}}$ one has $\kappa(z')\in\Wconv^{8}$
(we also use the fact that the 
$\Hb^{8,0}$-norm is bounded by the $\Cb^{8,1}$-norm).
Both $\MCinf(\kappa(z'))$ and $\MC_{\ge2}(\kappa(z'))$
vanish on $|\vec{x}|\ge101$, 
by \eqref{eq:PDP}, \eqref{eq:equivMCP}, \eqref{eq:DPB1}.
Then \eqref{eq:phi0equation} implies
\begin{equation}\label{eq:Mcinf2}
\Charge(\MCinf(\kappa(z')))
=
z' + \Charge(\MC_{\ge2}(\kappa(z')))
\end{equation}
Denote $F(z') = \Charge(\MC_{\ge2}(\kappa(z')))$,
which is a map $F:\Ccone_{M',2^{-5}}\to\R^{10}$.
Using Lemma \ref{lem:ChargeEst} and a calculation 
analogous to the proof of \eqref{eq:MC3} respectively \eqref{eq:MC4},
\begin{align*}
|F(z_1')| 
	&\lesssim
	\|\kappa(z_1')\|_{\Cb^{7,1}}^2
	\lesssim
	|z_1'|^2\\
|F(z_1')-F(z_2')| 
	&\lesssim
	(\|\kappa(z_1')\|_{\Cb^{7,1}}+\|\kappa(z_2')\|_{\Cb^{7,1}})
		\|\kappa(z_1')-\kappa(z_2')\|_{\Cb^{7,1}}\\
	&\lesssim
	(|z_1'|+|z_2'|)|z_1'-z_2'|
\end{align*}
where in both estimates, the second step holds by 
\eqref{eq:Uest}, \eqref{eq:phi0_bd}, \eqref{eq:phi0_Lip}.
Choose $M\in[0,M'/2]$ sufficiently small such that
for all $z_1',z_2'\in\Ccone_{2M,2^{-5}}$: 
\begin{equation}\label{eq:Fest1/2}
|F(z_1')| \le 2^{-8}|z_1'|
\qquad
\qquad
|F(z_1')-F(z_2')| \le 2^{-8}|z_1'-z_2'|
\end{equation}
Then for every $z\in \Ccone_{M,2^{-6}}$, the map
$z' \mapsto z-F(z')$
is contracting and a self map on
$B_{z}=\{z'\in\R^{10}\mid |z'-z| \le 2^{-7}|z|\}\subset\Ccone_{2M,2^{-5}}$.
%
Thus by the Banach fixed point theorem, there exists
a unique $z'\in B_z$ that satisfies
\begin{equation}\label{eq:zz'Fs}
z = z'+F(z')
\end{equation}
Set $\phi_1(z)=z'$.
This satisfies \eqref{eq:CMCinfeq}
by \eqref{eq:zz'Fs} and \eqref{eq:Mcinf2};
the estimate \eqref{eq:phi1_bd} holds because $\phi_1(z)\in B_{z}$;
and \eqref{eq:phi1_lip} 
follows from \eqref{eq:Fest1/2}, \eqref{eq:zz'Fs}.
\qed
\end{proof}
\begin{proof}[of Proposition \ref{prop:Kerr}]
Set $\phi=\phi_0\circ\phi_1$ using 
Lemma \ref{lem:kerrphi0} and \ref{lem:kerrphi1}.
Let $\phi^m$, $\phi^q$ the first respectively
last nine components of $\phi$. 
For all $z_1,z_2\in\Ccone_{M,2^{-6}}$:%
\begin{subequations}\label{eqPphibdlip}
\begin{align}
|\phi^m(z_1)| 
	&\lesssim |z_1| \label{eq:phibd}\\
|\phi^m(z_1)-\phi^m(z_2)|
+
(\phi^m(z_1)+\phi^m(z_2))|\phi^q(z_1)-\phi^q(z_2)| 
	&\lesssim |z_1-z_2| \label{eq:philip}
\end{align}
\end{subequations}
where the first holds by \eqref{eq:phi0_bd}, \eqref{eq:phi1_bd};
the second by \eqref{eq:phi0_Lip}, \eqref{eq:phi1_lip}.
We conclude that $\KS$ has the stated properties.
\ref{item:KSMC}: by \ref{item:KthesisMC}.
\ref{item:LKSsmoothnullinf}: by \ref{item:Kthesisnullinf}.
\ref{item:Kerr4dest}: 
by \ref{item:Kthesisest}, \eqref{eq:phibd}.
\ref{item:IntroKSConstraints}: by \eqref{eq:PDP}.
\ref{item:IntroKSest}: by \eqref{eq:Uest}, \eqref{eqPphibdlip}.
\ref{item:IntroKSCharges}: by \eqref{eq:CMCinfeq}.
\qed
\end{proof}

\footnotesize
\step

  \noindent\textsc{Department of Mathematics, Stockholm University,
  Stockholm, Sweden}\par\nopagebreak
 \noindent\textit{Email address:} \texttt{andrea.nuetzi@math.su.se}

\end{document}